\documentclass[11pt]{article}
\usepackage{fullpage}
\usepackage{amssymb}
\usepackage{amsthm}
\usepackage{amsmath}
\usepackage{graphicx}

\usepackage{subcaption}
\usepackage{url}  
\usepackage[]{algorithm2e}
\usepackage{multirow}
\usepackage{mathrsfs}

\newcommand{\UFRCASecure}[4]{\mathbf{UF}-\mathbf{RCA}\paren{#1,#2,#3,#4}-secure}

\newcommand{\sharedCues}{Shared Cues}
\newcommand{\ExtraRehearsals}[2]{ER_{#1,#2}}
\newcommand{\Hash}{\mathbf{H}}
\newcommand{\TotalExtraRehearsals}[1]{ER_{#1}}

\newcommand{\AssociationStrength}[1]{\mathbf{s}_{#1}}
\newcommand{\Adversary}{\mathcal{A}}

\newcommand{\paren}[1]{\left( #1 \right)}
\newcommand{\cut}[1]{}

\newcommand{\guess}[1]{q_{#1}}
\newcommand{\timet}{\hat{t}}
\newcommand{\stepst}{\tilde{t}}
\newcommand{\challengelength}{\lambda}
\newcommand{\close}{\epsilon}
\newcommand{\D}{{\mathcal D}}
\newcommand{\Z}{{\mathcal Z}}

\newcommand{\VSTAT}{{\mbox{VSTAT}}}

\newcommand{\MVSTAT}{{\mbox{MVSTAT}}}
\newcommand{\MSAMPLE}{{\mbox{1-MSTAT}}}

\newcommand{\dc}{\kappa_2}
\newcommand{\SDN}{{\mathrm{SDN}}}

\newtheorem{fact}{Fact}

\newtheorem{definition}{Definition}
\newtheorem{requirement}{Requirement}
\newtheorem{claim}{Claim}
\newtheorem{theorem}{Theorem}
\newtheorem{conjecture}{Conjecture}
\newtheorem{lemma}{Lemma}
\newtheorem{remark}{Remark}

\def \QED {\hfill{$\Box$}}

\newenvironment{proofof}[1]{\noindent {\em Proof of #1.  }}{\QED}

\newenvironment{remindertheorem}[1]{\medskip \noindent {\bf Reminder of Theorem #1.  }\em}{}

\newenvironment{reminderlemma}[1]{\medskip \noindent {\bf Reminder of Lemma #1.  }\em}{}

\newenvironment{reminderclaim}[1]{\medskip \noindent {\bf Reminder of Claim #1.  }\em}{}

\begin{document}


\title{Towards Human Computable Passwords\thanks{This work was partially supported by the NSF Science and Technology TRUST and the AFOSR MURI on Science of Cybersecurity. This work was completed in part while the first author was at Carnegie Mellon University where the first author partially supported by an NSF Graduate Fellowship.}}


%
%

\author{Jeremiah Blocki\\ Purdue University\\ \texttt{jblocki@purdue.edu}    \and Manuel Blum \\ Carnegie Mellon University\\  \texttt{mblum@cs.cmu.edu} \and Anupam Datta \\ Carnegie Mellon University\\  \texttt{danupam@cmu.edu} \and Santosh Vempala \\ Georgia Tech \\ \texttt{vempala@cc.gatech.edu}}



%
%

\maketitle

\begin{abstract}
An interesting challenge for the cryptography community is to design authentication protocols that are so simple that a human can execute them 
without relying on a fully trusted computer. We propose several candidate authentication protocols for a setting in which the human user can 
only receive assistance from a semi-trusted computer---a computer that stores information and performs computations correctly 
but does not provide confidentiality. Our schemes use a semi-trusted computer to store and display public challenges $C_i\in[n]^k$. 
The human user memorizes a random secret mapping $\sigma:[n]\rightarrow \mathbb{Z}_d$ and authenticates by computing responses $f(\sigma(C_i))$ to 
a sequence of public challenges where $f:\mathbb{Z}_d^k\rightarrow \mathbb{Z}_d$ is a function that is easy for the human to evaluate. We prove 
that any statistical adversary needs to sample $m=\tilde{\Omega}\paren{n^{s(f)}}$ challenge-response pairs to recover $\sigma$, for a security 
parameter $s(f)$ that depends on two key properties of $f$. Our lower bound generalizes recent results of Feldman et al. \cite{feldman2013complexity} 
who proved analogous results for the special case $d=2$. To obtain our results, we apply the general hypercontractivity theorem \cite{o2007analysis} 
to lower bound the {\em statistical dimension } of the distribution over challenge-response pairs induced by $f$ and $\sigma$. 
Our {\em statistical dimension} lower bounds apply to arbitrary functions $f:\mathbb{Z}_d^k\rightarrow \mathbb{Z}_d$ (not just to functions that 
are easy for a human to evaluate). As an application, we propose a family of human computable password 
functions $f_{k_1,k_2}$ in which the user needs to perform $2k_1+2k_2+1$ primitive operations (e.g., adding two digits or remembering a 
secret value $\sigma(i)$), and we show that $s(f) = \min\{k_1+1, (k_2+1)/2\}$. For these schemes, we prove that forging passwords is 
equivalent to recovering the secret mapping. Thus, our human computable password schemes can maintain strong security guarantees even after 
an adversary has observed the user login to many different accounts. 
\end{abstract}




\section{Introduction} \label{sec:Introduction}
A typical computer user has many different online accounts which require some form of authentication. While passwords are still the dominant form of authentication, users struggle to remember their passwords. As a result users often adopt insecure password practices (e.g., reuse, weak passwords) \cite{florencio2007large,center2010consumer,kruger2008empirical,bonneau2012science} or end up having to frequently reset their passwords. 
Recent large-scale password breaches highlight the importance of this problem \cite{breach:CERT-Warning,center2010consumer,breach:militaryHACK,breach:natoHACK,noPlaintextPassword,breach:Zappos,breach:Atlassian,breach:apple,breach:sony,breach:linkedin,breach:IEEE,breach:Adobe}. An important research goal is to develop usable and secure password management scheme ---  a systematic strategy to help users create and remember multiple passwords. Blocki et al. \cite{NaturallyRehearsingPasswords} and Blum and Vempala~\cite{blum2015publishable} recently proposed password management schemes that maintain some security guarantees after a small constant number of breaches (e.g., an adversary who sees three of the user's passwords still has some uncertainty about the user's remaining passwords). 

%

In this work we focus on the goal of developing human computable password management schemes  in which security guarantees are strongly maintained after {\em many} breaches (e.g., an adversary who sees one-hundred of the user's passwords still has high uncertainty about the user's remaining passwords).  In a human computable password management scheme the user reconstructs each of his passwords by {\em computing} the response to a public challenge. 

Our human computable password schemes admittedly require more human effort than the password management schemes of Blocki et al. \cite{NaturallyRehearsingPasswords} and Blum and Vempala~\cite{blum2015publishable}, and, unlike Blocki et al.  \cite{NaturallyRehearsingPasswords}, our scheme requires users to do simple mental arithmetic (e.g., add two single-digit numbers) in their head.  However, our proposed schemes are still human usable in the sense that a motivated, security-conscious user would be able to learn to use the scheme and memorize all associated secrets in a few hours.  In particular, the human computation in our schemes only involves a few very simple operations (e.g., addition modulo $10$) over secret values (digits) that the user has memorized. More specifically, in our candidate human computable password schemes the user learns to compute a simple function $f:\mathbb{Z}_d^k\rightarrow\mathbb{Z}_d$,\footnote{ In our security analysis we consider arbitrary bases $d$. However, our specific schemes use the base  $d = 10$ that is most familiar to human users.} and memorizes a secret mapping $\sigma:[n]\rightarrow \mathbb{Z}_d$. The user authenticates by responding to a sequence of single digit challenges, i.e.,  a challenge-response pair $\paren{C,f\paren{\sigma\paren{C}}}$ is a challenge $C \in X_k \subseteq [n]^k$ and the corresponding response.

One of our candidate human computable password schemes involves  the function
\[ f\paren{x_0,x_1,x_2,x_3,x_4,x_5,\ldots,x_{13}} = x_{13} + x_{12} + x_{\paren{x_{10}+x_{11} \mod{10}}} \mod{10} \ . \]
If the user memorizes a secret mapping $\sigma$ from $n$ images to digits then each challenge $C=(I_0,\ldots,I_{13})$ would correspond to an ordered subset of $14$ of these images and the response to the challenge is $f\big(\sigma(I_0),\ldots,\sigma(I_{13})\big)$. We observe that a human would only need to perform three addition operations modulo $10$ to evaluate this function.  The user would respond by (1) adding the secret digits associated with challenge images $I_{10}$ and $I_{11}$ to get a secret index $0 \leq i \leq 9$, (2) finding image $I_i$, (3) adding the secret digits associated with images $I_i$, $I_{12}$ and $I_{13}$ to produce the final response. To amplify security the user may respond to $\challengelength \geq 1$ single-digit challenges $C_1,\ldots,C_\challengelength$ to obtain a $\challengelength$ digit password $f(\sigma(C_1)),\ldots,f(\sigma(C_\challengelength))$.  We note that the challenge $C$ does not need to be kept secret and thus the images can be arranged on the screen in helpful manner for the human user --- see Figure \ref{fig:HumanComputablePasswordsDigitChallengeAndResponse} for an example and see Appendix \ref{apdx:AuthenticationProcess} for more discussion of the user interface.  

\begin{figure}
\centering
\includegraphics[scale=0.5]{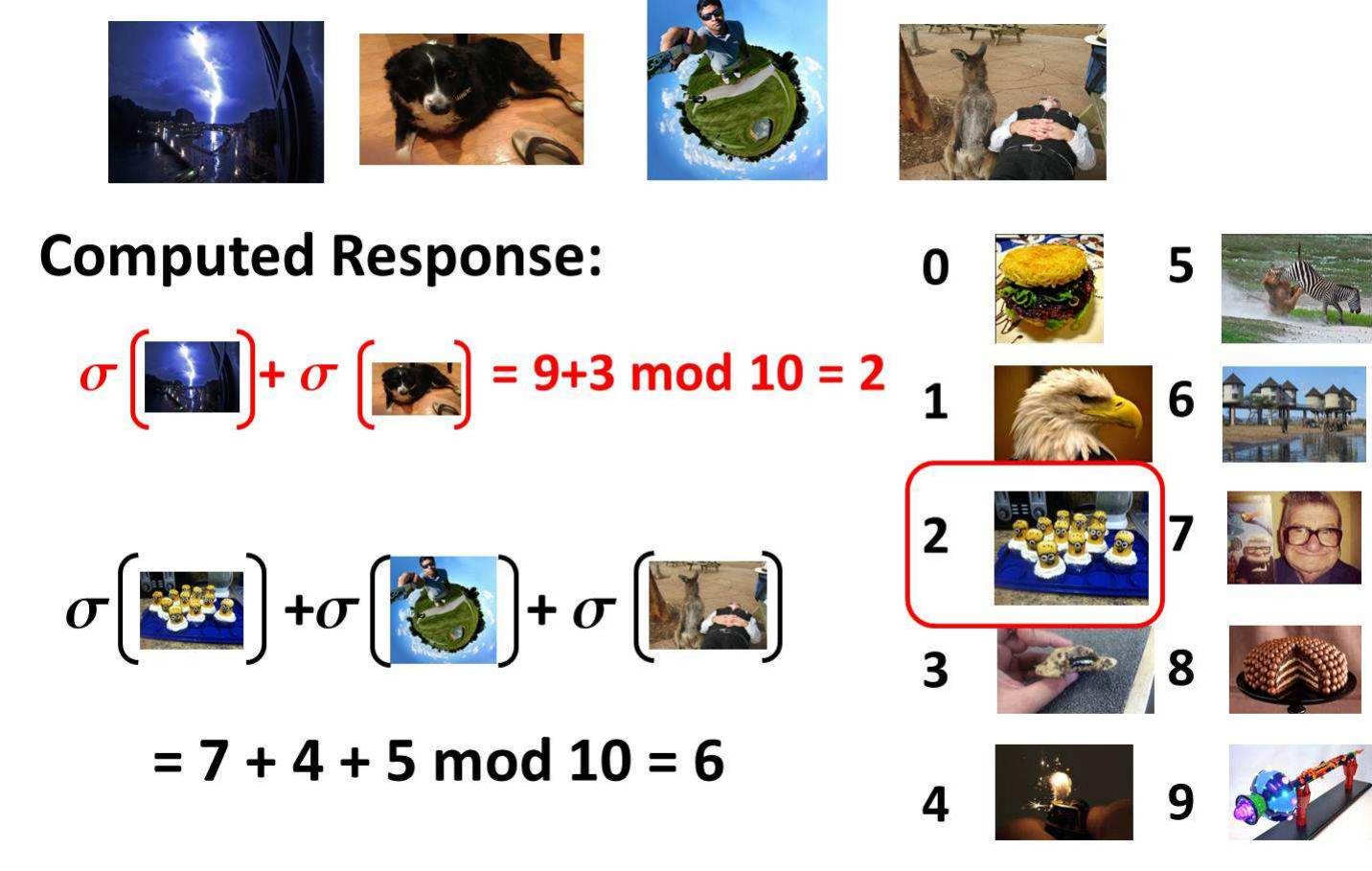}
\caption{Computing the response $f(\sigma(C)) = 6$ to a single digit challenge $C$.}
\label{fig:HumanComputablePasswordsDigitChallengeAndResponse}
\end{figure}

We present a natural conjecture which implies that a polynomial time attacker will need to see the responses to $\tilde{\Omega}\paren{n^{s(f)}}$ random challenges before he can forge the user's passwords (accurately predict the responses to randomly selected challenges)\footnote{We stress that, unlike \cite{NaturallyRehearsingPasswords,blum2015publishable}, our security guarantees are not information theoretic. In fact, a computationally unbounded adversary would need to see at most $O\paren{n}$ challenge-response pairs to break the human computable password management scheme.}. Here, $s(f)$ is a security parameter that depends on two key properties of the function $f$ (in our above example $s(f) = 3/2$). Furthermore, we provide strong evidence for our conjecture by ruling out a broad class of algorithmic techniques that the adversary might use to attack our scheme. 

Following Blocki et al. \cite{NaturallyRehearsingPasswords} we consider a setting where a user has two types of memory: {\em persistent memory} (e.g., a sticky note or a text file on his computer) and {\em associative memory} (e.g., his own human memory). We assume that persistent memory is reliable and convenient but not private (i.e., an adversary can view all challenges stored in persistent memory, but he cannot tamper with them). In contrast, a user's associative memory is private but lossy---if the user does not
rehearse a memory it may be forgotten. Thus, the user can store a password challenge $C \in X_k$ in persistent memory, but the mapping $\sigma$ must be stored in associative memory (e.g., memorized and rehearsed). We allow the user to receive assistance from a semi-trusted computer. A semi-trusted computer will perform computations accurately (e.g., it can be trusted to show the user the correct challenge), but it will not ensure privacy of its inputs or outputs. This means that a human computable password management scheme should be based on a function  $f$ that the user can compute entirely in his head.

\paragraph{Contributions.} 
We provide precise notions of security and usability for a human computable password management scheme (Section~\ref{sec:Definitions}). 
We introduce the notion of UF-RCA security (Unforgeability under Random Challenge Attacks). Informally, a human computable password 
scheme is UF-RCA secure if an adversary cannot forge passwords after seeing many example challenge-response pairs.

We present the design of a candidate family of human computable password management schemes $f_{k_1,k_2}$, and analyze the usability and 
security of these schemes (Section \ref{sec:Candidates}). Our usability analysis indicates that to compute $f_{k_1,k_2}\paren{\sigma\paren{C}}$ 
the user needs to execute $2k_1+2k_2+1$ simple operations (e.g., addition of single digits modulo 10). The main technical result of this section (Theorem \ref{thm:UFRCASecure}) 
states that our scheme is UF-RCA secure given a plausible 
conjecture about the hardness of random planted constraint satisfiability problems (\emph{RP-CSP}). 
Our conjecture is that any polynomial time adversary needs to see at least $m = n^{\min\{r(f)/2,g(f)+1-\epsilon\}}$ 
challenge-response pairs $\paren{C,f\paren{\sigma\paren{C}}}$ to recover the secret mapping $\sigma$. Here, 
$s(f) = \min\{r(f)/2,g(f)+1\}$ is a composite security parameter involving $g(f)$ (how many inputs to $f$ need to be 
fixed to make $f$ linear?) and $r(f)$ (what is the largest value of $r$ such that the distribution over 
challenge-response pairs are $\paren{r-1}$-wise independent?). We prove that $g\paren{f_{k_1,k_2}} = k_1$ and 
$r\paren{f_{k_1,k_2}} = \paren{k_2+1}/2$. 

Next we prove that any statistical adversary needs at least $\tilde{\Omega}\paren{n^{r(f)/2}}$ 
challenge-response pairs $\paren{C,f\paren{\sigma\paren{C}}}$ to find a secret mapping $\sigma'$ that is $\epsilon$-correlated with 
$\sigma$ (Section \ref{sec:TechnicalResults}). This result may be interpreted as strong evidence in favor of the \emph{RP-CSP} hardness assumption as most natural algorithmic techniques have statistical analogues (see   discussion in Section \ref{sec:TechnicalResults}). 
While Gaussian Elimination is a notable exception, our composite security parameter accounts for attacks based on Gaussian 
Elimination---an adversary needs to see $m=\tilde{\Omega}\paren{n^{1+g(f)}}$ challenge-response pairs to recover $\sigma$ using Gaussian 
Elimination. Moving beyond asymptotic analysis we also provide empirical evidence that our human computable password management scheme is hard to crack. In particular, we used a CSP solver to try to recover 
$\sigma \in \mathbb{Z}_{10}^n$ given $m$ challenge-response pairs using the functions $f_{1,3}$ and $f_{2,2}$. Our CSP solver failed to 
find the secret mapping $\sigma \in \mathbb{Z}_{10}^{50}$ given $m=1000$ random challenge-response pairs with both functions $f_{1,3}$ and $f_{2,2}$.  Additionally, we constructed public challenges for cryptographers to break our human computable password management schemes under various parameters (e.g., $n=100$, $m=1000$). 

Our lower bound for statistical adversaries is based on the {\em statistical dimension} of the distribution over challenge-response pairs 
induced by $f$ and $\sigma$. We stress that our analysis of the statistical dimension applies to arbitrary functions $f:\mathbb{Z}_d^k\rightarrow\mathbb{Z}_d$, 
not just to functions that are easy for humans to compute. Our analysis of the statistical dimension generalizes recent results of 
Feldman et al. \cite{feldman2013complexity} for binary predicates and may be of independent interest. While the analysis is similar at a high level, we stress that our proofs do require some new ideas. Because our function $f$ is not a binary predicate we cannot use the Walsh basis functions to express 
the Fourier decomposition of $f$ and analyze the statistical dimension of our distribution over challenge-response pairs as 
Feldman et al. \cite{feldman2013complexity} do. Instead, we use a generalized set of Fourier basis functions to take the Fourier decomposition of $f$, and we apply the general hypercontractivity theorem \cite{o2007analysis} to obtain our bounds on the statistical dimension. 

We complete the proof of Theorem \ref{thm:UFRCASecure} in Section \ref{sec:SecurityAnalysis} by proving that forging passwords and 
approximately recovering the secret mapping are equivalent problems for a broad class of human computable password schemes, 
including our candidate family $f_{k_1,k_2}$. This result implies that any adversary who can predict the response $f(C)$ to a random challenge $C$ with better accuracy than random guessing can be used as a blackbox to approximately recover the secret mapping.

\section{Definitions} \label{sec:Definitions}

\subsection{Notation} Given two strings $\alpha_1,\alpha_2 \in \mathbb{Z}_d^n$ we use $H\paren{\alpha_1,\alpha_2} \doteq \left| \left\{i\in[n]~\vline~ \alpha_1[i] \neq \alpha_2[i] \right\} \right|$ to denote the Hamming distance between them. We will also use $H\paren{\alpha_1} \doteq H\paren{\alpha_1, \vec{0}}$ to denote the Hamming weight of $\alpha_1$. We use $\sigma:[n]\rightarrow \mathbb{Z}_d$ to denote a secret random mapping that the user will memorize. We will sometimes abuse notation and think of $\sigma \in \mathbb{Z}_d^n$ as a string which encodes the mapping, and we will use $\sigma \sim \mathbb{Z}_d^n$ to denote a random mapping chosen from $\mathbb{Z}_d^n$ uniformly at random. Given a distribution $\mathcal{D}$ we will use $x \sim \mathcal{D}$ to denote a random sample from this distribution. We also use $x \sim S$ to denote an element chosen uniformly at random from a finite set $S$. 

\begin{definition} \label{def:EpsilonClose}
We say that two mappings $\sigma_1,\sigma_2 \in \mathbb{Z}_d^n$ are $\close$-correlated if $\frac{H\paren{\sigma_1,\sigma_2}}{n} \leq \frac{d-1}{d} - \close$, and we say that a mapping $\sigma \in \mathbb{Z}_d^n$ is $\delta$-balanced if $\max_{i \in \{0,\ldots,d-1\}} \left|\frac{H\paren{\sigma , \vec{i} }}{n} - \frac{d-1}{d} \right| \leq \delta$.
\end{definition}
Note that for a random mapping $\sigma_2$ we expect $\sigma_1$ and $\sigma_2$ to differ at $\mathbb{E}_{\sigma_2 \sim \mathbb{Z}_d^n}\left[H\paren{\sigma_1,\sigma_2} \right] = n\paren{\frac{d-1}{d}}$ locations, and for a random mapping $\sigma$ and $i \sim \{0,\ldots,d-1\}$ we expect $\sigma$ to differ from $\vec{i}$ at $\mathbb{E}_{i\sim\mathbb{Z}_d,\sigma \sim \mathbb{Z}_d^n}\left[H\paren{\sigma,\vec{i}} \right] = n\paren{\frac{d-1}{d}}$ locations. Thus, with probability $1-o(1)$ a random mapping $\sigma_2$ will not be $\epsilon$-correlated with $\sigma_1$, but a random mapping $\sigma$ will be $\delta$-balanced with probability $1-o(1)$. 

 We let $X_k\subseteq [n]^k$ denote the space of ordered clauses of $k$ variables without repetition. We use $C \sim X_k$ to denote a clause $C$ chosen uniformly at random from $X_k$ and we use $\sigma\paren{C} \in \mathbb{Z}_d^k$ to denote the values of the corresponding variables in $C$. For example, if $d=10$, $C = \paren{3, 10,59}$ and $\sigma\paren{i} = \paren{i +1\mod 10}$ then $\sigma\paren{C} = \paren{4,1,0}$. 

We view each clause $C \in X_k$ as a {\em single-digit challenge}. The user responds to a challenge $C$ by computing $f\paren{\sigma\paren{C}}$, where $f: \mathbb{Z}_d^k \rightarrow \mathbb{Z}_d$ is a {\em human computable} function (see discussion below) and $\sigma:[n]\rightarrow \mathbb{Z}_d$ is the secret mapping that the user has memorized. For example, if $d=10$, $C = \paren{3, 10,59}$, $\sigma\paren{i} = \paren{i+1 \mod 10}$ and $f\paren{x,y,z} = \paren{x-y+z \mod{10}}$ then $f\paren{\sigma\paren{C}} = \paren{4 - 1+0  \mod{10}} = 3$. A length-$\challengelength$ password challenge $\vec{C} = \left\langle C_1,\ldots,C_\challengelength\right\rangle \in \paren{X_k}^\challengelength$ is a sequence of $t$ single digit challenges, and $f\paren{\sigma\paren{\vec{c}}} = \left\langle f\paren{\sigma\paren{C_1}}, \ldots, f\paren{\sigma\paren{C_\challengelength}} \right\rangle \in \mathbb{Z}_d^\challengelength$ denotes the corresponding response (e.g., a password).

Let's suppose that the user has $m$ accounts $A_1,\ldots,A_m$. In a human computable password management scheme we will generate $m$ length-$\challengelength$ password challenges $\vec{C}_1,\ldots,\vec{C}_m \in \paren{X_k}^\challengelength$. These challenges  will be stored in persistent memory so they are always accessible to the user as well as the adversary. When our user needs to authenticate to account $A_i$ he will be shown the length-$t$ password challenge $\vec{C}_i = \left\langle C_1^i,\ldots, C_\challengelength^i\right\rangle$. The user will respond by computing his password $p_i = \left\langle  f\paren{\sigma\paren{C_1^i}},\ldots,f\paren{\sigma\paren{C_\challengelength^i}} \right\rangle\in \mathbb{Z}_d^\challengelength$. \\

\subsection{Requirements for a Human Computable Function}
In our setting we require that the composite function $f\circ\sigma:X_k\rightarrow\mathbb{Z}_d$ is human computable. Informally, we say that a function $f$ is {\em human-computable} if a human user can evaluate $f$ {\em quickly} in his head.
\begin{requirement}
\label{req:humanComputable} A function $f$ is $\timet$-{\em human computable} for a human user $H$ if $H$ can reliably evaluate $f$  in his head in $\timet$ seconds.
\end{requirement}

We argue that a function $f$ will be {\em human-computable} whenever there is a fast streaming algorithm~\cite{alon1996space} to compute $f$ using only very simple primitive operations. A streaming algorithm is an algorithm for processing a data stream in which the input (e.g., the challenge $C$) is presented as a sequence of items that can only be examined once. In our context the streaming algorithm must have a very low memory footprint because a typical person can only keep $7 \pm 2$ `chunks' of information in working memory \cite{memory:chunks:miller1956} at any given time. Our streaming algorithm can only involve primitive operations that a person could execute quickly in his head (e.g., adding two digits modulo $10$, recalling a value $\sigma(i)$ from memory). 

\begin{definition} \label{def:HumanComputable}
Let $P$ be a  set of primitive operations. We say that a function $f$ is $(P,\stepst,\hat{m})$-computable if there is a space $\hat{m}$ streaming algorithm $\mathcal{A}$ to compute $f$ using only $\stepst$ operations from $P$. 
\end{definition}

In this paper we consider the following primitive operations $P$: $\mathbf{Add}$, $\mathbf{Recall}$ and $\mathbf{TableLookup}$. $\mathbf{Add}: \mathbb{Z}_{10} \times \mathbb{Z}_{10} \rightarrow \mathbb{Z}_{10}$ takes two digits $x_1$ and $x_2$ and returns $x_1 + x_2 \mod{10}$. $\mathbf{Recall}:[n]\rightarrow \mathbb{Z}_{10}$ takes an index $i$ and returns the secret value $\sigma(i)$ that the user has memorized. $\mathbf{TableLookup}:\mathbb{Z}_{10} \times [n]^{10} \rightarrow [n] $ takes a digit $x_1$ and finds the $x_1$'th value from a table of $10$ indices. We take the view that no human computable function should require users to store intermediate values in long-term memory because the memorization process would necessarily slow down computation. Therefore, we restrict our attention to space $\hat{m}$ streaming algorithms and do not include any primitive operation like $\mathbf{MemorizeValue}$. \\

\noindent{\bf Example: } The function $f\circ \sigma(i_1,\ldots,i_5) = \sigma\paren{i_1}+\ldots+\sigma\paren{i_5}$ requires $9$ primitive operations (five $\mathbf{Recall}$ operations and four $\mathbf{Add}$ operations) and requires space $\hat{m} =3$ (e.g., we need one slot to store the current total, one slot to store the next value from the data stream and one free slot to execute a primitive operation).

Similar primitive operations have been studied by cognitive physchologists (e.g., \cite{sternberg2004memory}). The time $\gamma_H$ it takes a human user $H$ to execute one primitive operation will typically improve with practice (e.g., \cite{hitch1978role}). We note that we allow this computation speed constant $\gamma_H$ to vary from user to user in the same way that two computers might operate at slightly different speeds. 
We conjecture that, after training, a human user $H$ with a moderate mathematical background will be able to evaluate a $(P,\stepst,3)$-computable function in $\timet \leq \stepst$ seconds --- the first author of this paper found that (after some practice) he could  evaluate $(P,9,3)$-computable functions in $7.5$-seconds ($\gamma_H \leq 1$). 

\begin{conjecture} \label{conj:humanComputable}
Let $P = \left\{\mathbf{Add},\mathbf{Recall},\mathbf{TableLookup} \right\}$. For each human user $H$ there is a small constant $\gamma_H > 0$ such that any  $(P,\stepst,3)$-computable function $f$ will be $\timet$-human computable for $H$ with $\timet = \gamma_H\stepst$. 
\end{conjecture}

\subsection{Password Unforgeability}
In the password forgeability game the adversary attempts to guess the user's password for a randomly selected account after he has seen the user's passwords at $m/\challengelength$ other randomly selected accounts. We say that a scheme is UF-RCA (Unforgeability against Random Challenge Attacks) secure if any probabilistic polynomial time adversary fails to guess the user's password with high probability. In the password forgeability game we select the secret mapping $\sigma:[n]\rightarrow \mathbb{Z}_d$ uniformly at random along with challenges $C_1,\ldots, C_{m+\challengelength} \sim X_k$. The adversary is given the function $f:\mathbb{Z}_d^k\rightarrow \mathbb{Z}_d$ and is shown the challenges $C_1,\ldots, C_{m+\challengelength}$ as well as the values $f\paren{\sigma\paren{C_i}}$ for $i \in \{1,\ldots, m\}$. The game ends when the adversary $\Adversary$ outputs a guess $\left\langle \guess{1},\ldots,\guess{\challengelength} \right\rangle \in \mathbb{Z}_d^\challengelength$ for the value of $\left\langle f\paren{\sigma\paren{C_{m+1}}},\ldots,f\paren{\sigma\paren{C_{m+\challengelength}}}\right\rangle$. We say that the adversary wins if he correctly guesses the responses to all of the challenges $C_{m+1},\ldots,C_{m+\challengelength}$, and we use  $\mathbf{Wins}\paren{\Adversary,n,m,\challengelength}$ to denote the event that the adversary wins the game (e.g., $\forall i \in \{1,\ldots,\challengelength\}. \guess{i} = f\paren{\sigma\paren{C_{m+i}}}$ ). We are interested in understanding how many example single digit challenge-response pairs the adversary needs to see before he can start breaking the user's passwords.

\begin{definition} \label{def:security} (Security) We say that a function $f:\mathbb{Z}_d^k\rightarrow \mathbb{Z}_d$ is $\UFRCASecure{n}{m}{\challengelength}{\delta}$ if for every probabilistic polynomial time (in $n,m$) adversary $\mathcal{A}$ we have
 $ \Pr\left[\mathbf{Wins}\paren{\mathcal{A},n,m,\challengelength} \right] \leq \delta$, where the randomness is taken over the selection of the secret mapping $\sigma \sim \mathbb{Z}_d^n$, the challenges $C_1,\ldots, C_{m+\challengelength}$ as well as the adversary\rq{}s coins.
\end{definition}

\paragraph{Discussion} Our security model is different from the security model of Blocki et al. \cite{NaturallyRehearsingPasswords} in which the adversary gets to adaptively select which accounts to compromise and which account to attack. While our security model may seem weaker at first glance because the adversary does not get to select which account to compromise/attack, we observe that the password management schemes of  Blocki et al. \cite{NaturallyRehearsingPasswords} are only secure against one to three adaptive breaches. By contrast, our goal is to design human computable password schemes that satisfy $\mathbf{UF}-\mathbf{RCA}$ security for large values of $m$ (e.g. $1000$), which means that it is reasonable to believe that the user has at most $m/\challengelength$ password protected accounts. If the user has at most $m/\challengelength$ accounts then union bounds imply that an adaptive adversary --- who gets to compromise all but one account --- will not be able to forge the password at any remaining account with probability greater than $m\delta/\challengelength$ (typically, $m \ll \challengelength/\delta$)\footnote{We assume in our analysis that the adversary does not get to pick the challenges $C$ that the user will solve. }. 

\subsection{Security Parameters of $f$}
Given a function $f:\mathbb{Z}_d^k\rightarrow \mathbb{Z}_d$ we define the function $Q^f:\mathbb{Z}_d^{k+1}\rightarrow\{\pm 1\}$ s.t. $Q^{f}\paren{x,i} = 1$ if $f(x) = i$; otherwise  $Q^{f}\paren{x,i} = -1$. We use $Q^f_\sigma$ to define a distribution over $X_k \times \mathbb{Z}_d$ (challenge-response pairs) as follows:
$\Pr_{ Q^{f}_\sigma}\left[ C,i\right] \doteq \frac{Q^{f}\paren{\sigma\paren{C},i}+1}{2\left|X_k \right|}$. 
Intuitively, $Q^f_\sigma$ is the uniform distribution over challenge response pairs $(C,j)$ s.t. $f\paren{\sigma\paren{C}} = j$. We also use 
 $Q^{f,j}:\mathbb{Z}_d^{k}\rightarrow\{\pm 1\}$ ($Q^{f,j}\paren{x} = Q^f\paren{x,j}$) to define a distribution over $X_k$. $\Pr_{Q^{f,j}_\sigma}[C] = \frac{Q^{f,j}(f(\sigma(C)))+1}{2\left|\{C' \in X_k: f(\sigma(C'))=j \right|} =  \Pr_{ Q^{f}_\sigma}\left[ (C,i)~\vline ~ i=j\right]$. We write the Fourier decomposition of a function $Q:\mathbb{Z}_d^k\rightarrow \{\pm 1\}$ as follows
 \[Q(x) = \sum_{\alpha \in \mathbb{Z}_d^{k}} \hat{Q}_\alpha \cdot \chi_\alpha\paren{x} \ ~ , \mbox{where the basis functions are~~~}  \chi_\alpha\paren{x} \doteq \exp\paren{\frac{-2\pi\sqrt{-1} \paren{x\cdot \alpha} }{d}} \ . \] 
We say that a function $Q$ has degree $\ell$ if $\ell = \max \left\{H\paren{\alpha}~\vline~\alpha\in\mathbb{Z}_d^k \wedge \hat{Q}_{\alpha} \neq 0 \right\}$ --- equivalently if $Q(x) = \sum_{i} Q_i(x)$ can be expressed as a sum of functions where each function $Q_i:\mathbb{Z}_d^k\rightarrow \mathbb{R}$ depends on at most $\ell$ variables.

\begin{definition} \label{def:distributionalComplexity}
We use $r(Q) \doteq \min\left\{H\paren{\alpha}~\vline ~\exists \alpha \in \mathbb{Z}_d^k. \hat{Q}_\alpha \neq 0 \wedge \alpha \neq \vec{0} \right\}$ to denote the distributional complexity of $Q$, and we use $r(f) = \min\left\{r\paren{Q^{f,j}}~\vline~j\in \mathbb{Z}_d \right\}$ to denote the distributional complexity of $f$. We use $$g(f) \doteq \min\left\{\ell \in \mathbb{N}\cup\{0\}~\vline~\exists \alpha \in \mathbb{Z}_d^\ell, S \subseteq [k], \hat{d} \in \mathbb{Z}_d. \mbox{s.t $|S|=\ell$ and $f_{|S,\alpha}$ is a linear function $\mod{\hat{d}}$} \right\} \ ,$$ to denote the minimum number of variables that must be fixed to make $f$ a linear function. Here, $f_{|S,\alpha}:\mathbb{Z}_d^{k-\ell}\rightarrow \mathbb{Z}_d$ denotes the function $f$ after fixing the variables at the indices specified by $S$ to $\alpha$. Finally, we use $s(f) \doteq \min\{r(f)/2, g(f)+1\}$ as our composite security measure.
\end{definition}

We conjecture that a polynomial time adversary will need to see $m = n^{s(f)}$ challenge-response pairs before he can approximately recover the secret mapping $\sigma$. We call this conjecture about the hardness of random planted constraint satisfiability problems RP-CSP (Conjecture \ref{conj:hardnessConjecture}). In support of RP-CSP we prove that any statistical algorithm needs to see at least $m=\tilde{\Omega}\paren{n^{r(f)/2}}$ challenge response pairs to (approximately) recover the secret mapping $\sigma$ and we observe that a polynomial time adversary would need to see $m=O\paren{n^{g(f)+1}}$ challenge-response pairs to recover $\sigma$ using Gaussian Elimination. In Section \ref{sec:SecurityAnalysis} we show that the human computable password scheme will be UF-RCA secure provided that RP-CSP holds and that $f$ satisfies a few moderate properties (e.g., the output of $f$ is evenly distributed).

\begin{conjecture}[RP-CSP] \label{conj:hardnessConjecture}
For every  probabilistic polynomial time adversary $\Adversary$ and every $\epsilon,\epsilon' >0$ there is an integer $N$ s.t. for all $n > N$, $m \leq n^{\min \left\{r(f)/2,g(f)+1-\epsilon' \right\}}$ we have
$ \Pr\left[ \mathbf{Success}\paren{\Adversary,n,m,\epsilon} \right] \leq \mu(n) $,
where $\mathbf{Success}\paren{\Adversary,n,m,\epsilon}$ denotes the event that $\Adversary$ finds a mapping $\sigma'$ that is $\epsilon$-correlated with $\sigma$ given $m$ randomly selected challenge response pairs $\paren{C_1,f\paren{\sigma\paren{C_1}}},\ldots,\paren{C_m,f\paren{\sigma\paren{C_m}}}$ and  $\mu(n)$ is a negligible function. The probability is over the selection of the random mapping $\sigma$, the challenges $C_1,\ldots,C_m$ and the random coins of the adversary.
\end{conjecture}

\section{Candidate Secure Human Computable Functions} \label{sec:Candidates}
In this section we present a family of candidate human computable functions. We consider the usability of these human computable password schemes in Section \ref{subsec:Usability}, and we analyze the security of our schemes in Section \ref{subsec:SecurityAnalysis}. 

We first introduce our family of candidate human computable functions (for all of our candidate human computable functions $f:\mathbb{Z}_d^k\rightarrow\mathbb{Z}_d$ we fix $d = 10$ because most humans are used to performing arithmetic operations on digits). Given integers $k_1 > 0$ and $k_2 > 0$ we define the function $f_{k_1,k_2}:\mathbb{Z}_{10}^{10+k_1+k_2}$ as follows 
\[f_{k_1,k_2}\paren{x_0,\ldots,x_{9+k_1+k_2}} = x_{j} + \sum_{i= 10+k_1}^{9+k_1+k_2} x_i \mod{10} \mbox{,~~~where~~} j = \paren{\sum_{i=10}^{9+k_1} x_i} \mod{10} \ . \]

\paragraph{Authentication Process} We briefly overview the authentication process ---  see Algorithms \ref{alg:GenStories} and \ref{alg:Authenticate} in Appendix \ref{apdx:AuthenticationProcess} for a more formal presentation of the authentication process. We assume that the mapping $\sigma:\{1,...,n\}\rightarrow \mathbb{Z}_{10}$ is generated by the user's local computer in secret. The user may be shown mnemonic helpers (see discussion below) to help memorize $\sigma$, but these mnemonic helpers are discarded immediately afterward. After the user has memorized $\sigma$ he can create a password $pw_i$ for an account $A_i$ as follows: the user's local computer generates $\challengelength$ random single-digit challenges $C_1^i,\ldots,C_\challengelength^i \in X_\challengelength$ and the user computes $pw_i = f\paren{\sigma\paren{C_1}},\ldots,f\paren{\sigma\paren{C_\challengelength}}$. The authentication server for account $A_i$ stores the cryptographic hash of $pw_i$, while the challenges $C_1^i,\ldots,C_\challengelength^i \in X_\challengelength$ are stored in public memory (e.g., on the user's local computer), which means that they can be viewed by the adversary as well as the legitimate user. To authenticate the user retrieves the public challenges $C_1^i,\ldots,C_\challengelength^i$ for account $A_i$ and computes $pw_i$. The server for $A_i$ verifies that the cryptographic hash of $pw_i$ matches its records. To protect users from offline attacks in the event of a server breach, the password $pw_i$ should be stored using a slow cryptographic hash function $\Hash$ like BCRYPT \cite{bcrypt}. 

\subsection{ Usability } \label{subsec:Usability} In our discussion of usability we focus on the time it would take a human user to compute a password once he has memorized the secret mapping $\sigma$. Other important considerations include the challenge of memorizing and rehearsing the secret mapping $\sigma$ to ensure that the user remembers the secret mapping $\sigma$ over time. 

\subsubsection{Computation Time} Given a challenge $C = \paren{c_0,\ldots,c_{9+k_1+k_2}} \in X_{10+k_1+k_2}$ we can compute $f_{k_1,k_2}\paren{\sigma\paren{C}}$ we compute $j =  \sum_{i=10}^{9+k_1} \sigma\paren{c_i} \mod{10}$ using $k_1-1$ $\mathbf{Add}$ operations and $k_1$ $\mathbf{Recall}$ operations. We then execute  $\mathbf{TableLookup}\paren{j,c_0,\ldots,c_9}$ to obtain $c_j$. Now we need $k_2$ $\mathbf{Add}$ operations and $k_2+1$ $\mathbf{Recall}$ operations to compute the final response $\sigma\paren{c_j} +  \sigma\paren{c_{10+k_1}} + \ldots + \sigma\paren{c_{9+k_1+k_2}}$. 

\begin{fact} \label{fact:HumanComputable}
Let $P = \left\{\mathbf{Add},\mathbf{Recall},\mathbf{TableLookup} \right\}$ then
$f_{k_1,k_2}\circ \sigma$ is $(P,2k_1+2k_2+1,3)$-computable.
\end{fact}

Fact \ref{fact:HumanComputable} and Conjecture \ref{conj:humanComputable} would imply that $f_{1,3}$ and $f_{2,2}$ are  $\timet$-human computable with $\timet = 9$ seconds for humans $H$ with computation constant $\gamma_H \leq 1$. The functions $f_{1,3}$ and $f_{2,2}$ were both $\timet$-human computable with $\timet = 7.5$ seconds for the main author of this paper. While the value of $\gamma_H$ might be larger for many human users who are less comfortable with mental arithmentic, we note we may have $\gamma_H \ll 1$ for many human users after training (e.g., see \url{https://youtu.be/_-2L6ZxFacg} for a particularly impressive demonstration of mental arithmetic by young children.).


\subsubsection{Memorizing and Rehearsing $\sigma$} Memorizing the secret mapping might be the most difficult part of our schemes. In practice, we envision that the user memorizes a mapping from $n$ objects (e.g., images) to digits. For example, if  $n=26$ and $d=10$ then the user might memorize a random mapping from characters to digits. The first author of this paper was able to memorize a mapping from $n=100$ images to digits in about 2 hours. We conjecture that the process could be further expedited using mnemonic helpers --- see discussion in the appendix.

After the user memorizes $\sigma$ he may need to rehearse parts of the mapping periodically to ensure that he does not forget it. One of the benefits of our human computable password schemes is that the user will get lots of practice rehearsing the secret mapping each time he computes a password. In fact users who authenticate frequently enough will not need to spend any extra time rehearsing the secret mapping as they will get sufficient natural practice to remember $\sigma$.

\subsection{Security Analysis} \label{subsec:SecurityAnalysis}

 Claim \ref{claim:CandidateFamilySecurityParameters} demonstrates that  $s\paren{f_{k_1,k_2}} = \min\left\{(k_2+1)/2,k_1+1 \right\}$. Intuitively, the security of our human computable password management scheme will increase with $k_1$ and $k_2$. However, the work that the user needs to do to respond to each single-digit challenge is proportional to $2k_1+2k_2+1$ (See Fact \ref{fact:HumanComputable}). 

\newcommand{\clmCandidateFamilySecurityParameters}{Let $0 \le k_1$ and  $k_2>0$ be given and let $f = f_{k_1,k_2}$ we have $g(f) = \min\{ k_1,10\}$, $r(f) = k_2+1$ and $s(f) = \min\left\{\frac{k_2+1}{2}, k_1+1, 11 \right\}$.}
\begin{claim}\label{claim:CandidateFamilySecurityParameters}
\clmCandidateFamilySecurityParameters
\end{claim}

An intuitive way to see that $r\paren{f_{k_1,k_2}} > k_2$ is to observe that we cannot bias the output of $f_{k_1,k_2}$ by fixing $k_2$ variables. Fix the value of {\em any} $k_2$ variables and draw the values for the other $k_1+10$ variables uniformly at random from $\mathbb{Z}_{10}$. One of the $k_2+1$ variables in the sum $x_j+\sum_{i=10+k_1}^{9+k_1+k_2} x_i \mod{10}$ will not be fixed. Thus, the probability that the final output of $f_{k_1,k_2}\paren{x_0,\ldots,x_{9+k_1+k_2}}$ will be $r$ is exactly $1/10$ for each digit $r \in \mathbb{Z}_{10}$. Similarly, an intuitive way to see that $r\paren{f_{k_1,k_2}} \leq k_2+1$ is to observe that we can bias the value of $f_{k_1,k_2}\paren{x_0,\ldots,x_{9+k_1+k_2}}$ by fixing the value of $k_2+1$ variables. In particular if we fix the variables $x_0,x_{10+k_1},\ldots,x_{9+k_1+k_2}$ so that $0 = x_0+\sum_{i=10+k_1}^{9+k_1+k_2} x_i \mod{10}$  then the output of $f_{k_1,k_2}\paren{x_0,\ldots,x_{9+k_1+k_2}}$ is more likely to be $0$ than any other digit. The full proof of Claim \ref{claim:CandidateFamilySecurityParameters} can be found in Appendix \ref{apdx:ProofOfClaims}.

Theorem \ref{thm:UFRCASecure} states that our human computable password management scheme is UF-RCA secure as long as RP-CSP (Conjecture \ref{conj:hardnessConjecture}) holds. In Section \ref{sec:TechnicalResults} we provide strong evidence in support of RP-CSP. In particular, no statistical algorithm can approximately recover the secret mapping given $m = \tilde{O}\paren{n^{r(f)/2}}$ challenge-response pairs. To prove Theorem \ref{thm:UFRCASecure} we need to show that an adversary that breaks UF-RCA security for $f_{k_1,k_2}$ can be used to approximately recover the secret mapping $\sigma$. We prove a more general result in Section \ref{sec:SecurityAnalysis}.

\begin{theorem} \label{thm:UFRCASecure}
Let $\epsilon, \epsilon' > 0, \challengelength \geq 1$ be given. Under the RP-CSP conjecture (Conjecture \ref{conj:hardnessConjecture}) the human computable password scheme defined by $f_{k_1,k_2}$ is $\UFRCASecure{n}{m}{\challengelength}{\delta}$ for any $m \leq n^{\min\{\paren{k_2+1}/2,k_1+1-\epsilon' \}}-\challengelength$ and $\delta > \paren{\frac{1}{10}+\epsilon}^\challengelength$.
\end{theorem}

\begin{remark} \label{rmk:TightBounds}
In the Appendix we demonstrate that our security bounds are asymptotically tight. In particular, there is a statistical algorithm to break our human computable password schemes $(f_{k_1,k_2})$ which requires $m=\tilde{O}\paren{n^{\paren{k_2+1}/2}}$ to $\MSAMPLE$ to recover $\sigma$ (See Theorem \ref{thm:StatisticalUpperBound} in Section \ref{apdx:subsec:securityUpperBound}). We also demonstrate that there is a attack based on Gaussian Elimination that uses $m = \tilde{O}\paren{n^{k_1+1}}$ challenge-response pairs to recover $\sigma$.  
\end{remark} 
\subsubsection{Exact Security Bounds} We used the Constraint Satisfaction Problem solver from the Microsoft Solver Foundations library to attack our human computable password scheme\footnote{Thanks to David Wagner for suggesting the use of SAT solvers.}. In each instance we generated a random mapping $\sigma:[n]\rightarrow \mathbb{Z}_{10}$ and $m$ random challenge response pairs $\paren{C,f\paren{\sigma\paren{C}}}$ using the functions $f_{2,2}$ and $f_{1,3}$. We gave the CSP solver $2.5$ days to find $\sigma$ on a computer with a 2.83 GHz Intel Core2 Quad CPU and 4 GB of RAM. The full results of our experiments are in Appendix \ref{apdx:CSPSolver}. Briefly, the solver failed to find the random mapping in the following instances with $f = f_{2,2}$ and $f = f_{1,3}$: (1) $n=30$ and $m = 100$, (2) $n=50$ and $m=1,000$ and (3) $n = 100$ and $m= 10,000$.

\begin{remark}
While the theoretical security parameter for $f_{1,3}$ $\paren{s\paren{f_{1,3}} = 2}$ is slightly better than the security parameter for $f_{2,2}$ $\paren{s\paren{f_{2,2}} = 1.5}$, we conjecture that $f_{2,2}$ may be more secure for small values of $n$ (e.g., $n \leq 100$) because it is less vulnerable to attacks based on Gaussian Elimination. In particular, there is a polynomial time attack on $f_{1,3}$ based on Gaussian Elimination that requires at most $n^2$ examples to recover $\sigma$, while the same attack would require $n^3$ examples with $f_{2,2}$. Our CSP solver was not able to crack $\sigma\in\mathbb{Z}_{10}^{100}$ given $10,000 = 100^2$ challenge response pairs with $f_{2,2}$.
\end{remark}

\paragraph{Human Computable Password Challenge.}  We are challenging the security and cryptography community to break our human computable password scheme for instances that our CSP solver failed to crack (see Appendix \ref{sec:challenge} for more details about the challenge). Briefly, for each challenge we selected a random secret mapping $\sigma \in \mathbb{Z}_{10}^n$, and published (1) $m$ single digit challenge-response pairs $\paren{C_1,f\paren{\sigma\paren{C_1}}}$,$\ldots$, $ \paren{C_m,f\paren{\sigma\paren{C_m}}}$, where each clause $C_i$ is chosen uniformly at random from $X_k$, and (2) $20$ length--$\challengelength=10$ password challenges $\vec{C}_1,\ldots,\vec{C}_{20} \in \paren{X_k}^{10}$. The goal of each challenge is to correctly guess one of the secret passwords $p_i = f\paren{\sigma\paren{\vec{C}_i}}$ for some $i \in [20]$. The challenges can be found at \url{http://www.cs.cmu.edu/~jblocki/HumanComputablePasswordsChallenge/challenge.htm}. There is a $\$20$ prize associated with each individual challenge (total: $\$360$). We remark that these challenges remain unsolved even after they were presented during the rump sessions at a cryptography conference and a security conference\cite{HumanComputablePasswordChallenge}. \\



\section{Statistical Adversaries and Lower Bounds} \label{sec:TechnicalResults}

Our main technical result (Theorem \ref{thm:SecurityLowerBound}) is a lower bound on the number of single digit challenge-response pairs that a statistical algorithm needs to see to (approximately) recover the secret mapping $\sigma$. Our results are quite general and may be of independent interest. Given {\em any} function $f:\mathbb{Z}_d^k\rightarrow \mathbb{Z}_d$ we prove that {\em any} statistical algorithm needs $\tilde{\Omega}\paren{n^{r(f)/2}}$ examples before it can find a secret mapping  $\sigma' \in \mathbb{Z}_d^n$ such that $\sigma'$ is $\epsilon$-correlated with $\sigma$.
 We first introduce statistical algorithms in Section \ref{subsec:StatisticalAlgorithm} before stating our main lower bound for statistical algorithms in Section \ref{subsec:StatisticalDimensionLowerBounds}. We also provide a high level overview of our proof in Section \ref{subsec:StatisticalDimensionLowerBounds}.   

\subsection{Statistical Algorithms} \label{subsec:StatisticalAlgorithm}
A statistical algorithm is an algorithm that solves a distributional search problem $\Z$. In our case the distributional search problem $\Z_{\close,f}$ is to find a mapping $\tau$ that is $\close$-correlated with the secret mapping $\sigma$ given access to $m$ samples from $ Q^f_\sigma$ -- the distribution over challenge response pairs induced by $\sigma$ and $f$. A statistical algorithm can access the input distribution $ Q^f_\sigma$ by querying the $\MSAMPLE$ oracle or by querying the $\VSTAT$ oracle (Definition \ref{def:StatOracle}).

\begin{definition}\cite{feldman2013complexity} [$\MSAMPLE(L)$ oracle and $\VSTAT$ oracle] \label{def:StatOracle}
  Let $D$ be the input distribution over the domain $X$.  Given any function $h: X \rightarrow \{0,1,\ldots,L-1\}$,
  $\MSAMPLE(L)$ takes a random sample $x$ from $D$ and returns $h(x)$. For an integer parameter $T > 0$ and any query function $h:X \rightarrow \{0,1\}$, $\VSTAT\paren{T}$ returns a value $v \in \left[ p-\tau,p+\tau\right]$ where $p = \mathbb{E}_{x\sim D}\left[ h(x)\right]$ and $\tau = \max\left\{\frac{1}{T},\sqrt{\frac{p(1-p)}{T}} \right\}$. 
\end{definition}

In our context the domain $X = X_k\times \mathbb{Z}_d$ is the set of all challenge response pairs and the distribution $D = Q^f_\sigma$ is the uniform distribution over challenge-response pairs induced by $\sigma$ and $f$. Feldman et al. \cite{feldman2013complexity}  used the notion of statistical dimension (Definition \ref{def:sdima} )  to lower bound the number of oracle queries necessary to solve a distributional search problem (Theorem \ref{thm:avgvstat-random}). Before we can present the definition of statistical dimension we need to introduce the {\em discrimination norm}. Intuitively, if the discrimination norm is small then a statistical algorithm will (whp) not be able to distinguish between honest samples $(C,f(\sigma(C))$ and samples from reference distribution $T$ over $X_k \times \mathbb{Z}_d$ which is completely independent of $\sigma$ \footnote{Observe that this implies that a statistical algorithm cannot find the secret $\sigma$. In particular, because the distribution $T$ is independent of the secret mapping $\sigma$ samples from $T$ will not leak any information about $\sigma$. }. We define our reference distribution as follows:
\[ \Pr_{T}\left[(C,i)\right] = \frac{\Pr_{x \sim \mathbb{Z}_d^k}\left[f(x) = i \right]}{|X_k|} \ .\]
 Now given a set $\D' \subseteq \mathbb{Z}_d^n$ of secret mappings the discrimination norm of $\D'$ is denoted by $\dc(\D')$ and defined as follows:  
\begin{align*}
\dc(\D') \doteq \max_{h, \|h\|=1} \left\{ \mathbb{E}_{\sigma \sim \D'}\left[\left| \Delta\paren{h,\sigma} \right| \right] \right\} \ ,
\end{align*}
where $h:X_k \times \mathbb{Z}_d\rightarrow \mathbb{R}$, $\|h\| \doteq \sqrt{\mathbb{E}_{(C,i) \sim X_k \times \mathbb{Z}_d}\left[h^2\paren{C,i} \right]}$ and \[ \Delta\paren{h,\sigma} \doteq  \mathbb{E}_{C\sim X_k}\left[h\paren{C,f\paren{\sigma\paren{C}}}  \right]-\mathbb{E}_{(C,i)\sim T}\left[ h\paren{C,i}\right] \ .\] 

\begin{definition} \cite{feldman2013complexity}\footnote{For the sake of simplicity we define the discrimination norm and the statistical dimension using our particular distributional search problem $\Z_{\epsilon,f}$. Our definition is equivalent to the definition in \cite{feldman2013complexity} once we fix the reference distribution $T$.  }. \label{def:sdima}
  For $\kappa>0$, $\eta >0$, $\epsilon >0$, let $d'$ be the largest integer such that for any mapping $\sigma \in \mathbb{Z}_d^n$ the set $\D_\sigma = \mathbb{Z}_d^n \setminus \left\{\sigma' \in \mathbb{Z}_d^n  ~\vline~\mbox{ $\sigma'$ is $\close$-correlated with $\sigma$ } \right\}$ has size at least $\paren{1-\eta} \cdot \left|\mathbb{Z}_d^n \right|$ and for any subset $\D' \subseteq \D_\sigma$ where $|\D'| \ge |\D_{\sigma}|/d'$, we have $\dc(\D') \leq \kappa$. The \textbf{statistical dimension} with discrimination norm $\kappa$ and error parameter $\eta$ is $d'$ and denoted by $\SDN(\Z_{\close,f},\kappa,\eta)$.
\end{definition}

Feldman et al. \cite{feldman2013complexity} proved the following lower bound on the number of $\MSAMPLE$ and $\VSTAT$ queries needed to solve a distributional search problem. Intuitively, Theorem \ref{thm:avgvstat-random} implies that many queries are needed to solve a distributional search problem with high statistical dimension. In Section \ref{subsec:StatisticalDimensionLowerBounds} we argue that the statistical dimension our distributional search problem (finding $\sigma'$ that is $\close$-correlated with the secret mapping $\sigma$ given $m$ samples from the distribution $Q^f_\sigma$) is high.

\newcommand{\thmAverageStatRandom}{For $\kappa > 0$ and $\eta\in (0,1)$ let $d' = \SDN(\Z_{\close,f},\kappa,\eta)$ be the statistical dimension of the distributional search problem $\Z_{\close,f}$. Any randomized statistical algorithm that, given access to a $\VSTAT\paren{\frac{1}{3\kappa^2}}$ oracle (resp. $\MSAMPLE\paren{L}$) for the distribution $Q_{\sigma}^f$ for a secret mapping $\sigma$ chosen randomly and uniformly from $\mathbb{Z}_d^n$, succeeds in finding a mapping $\tau \in \mathbb{Z}_d^n$ that is $\epsilon$-correlated with $\sigma$ with probability $\Lambda > \eta$ over the choice of distribution and internal randomness requires at least $\frac{\Lambda-\eta}{1-\eta}d'$ (resp. $\Omega\paren{\frac{1}{L}\min\left\{\frac{d'\paren{\Lambda-\eta}}{1-\eta}, \frac{\paren{\Lambda-\eta}^2}{\kappa^2} \right\}}$) calls to the oracle. }
\begin{theorem}
\cite[Theorems 10 and 12]{feldman2013complexity}
\label{thm:avgvstat-random}
\thmAverageStatRandom
\end{theorem}

 As Feldman et al. \cite{feldman2013complexity} observe, almost all known algorithmic techniques can be modeled within the statistical query framework. In particular, techniques like Expectation Maximization\cite{dempster1977maximum}, local search, MCMC optimization\cite{gelfand1990sampling}, first and second order methods for convex optimization, PCA, ICA, k-means can be modeled as a statistical algorithm even with $L=2$ --- see \cite{blum2005practical} and \cite{chu2007map} for proofs. One issue is that a statistical simulation might need polynomially more samples. However, for $L > 2$ we can think of our queries to $\MSAMPLE(L)$ as evaluating $L$ disjoint functions on a random sample. Indeed, Feldman et al. \cite{feldman2013complexity} demonstrate that there is a statistical algorithm for binary planted satisfiability problems using $\tilde{O}\paren{n^{r(f)/2}}$ calls to $\MSAMPLE\paren{n^{\lceil r(f)/2\rceil}}$.  

\begin{remark} \label{rmk:OtherOracles}
We can also use the statistical dimension to lower bound the number of queries that an algorithm would need to make to other types of statistical oracles to solve a distributional search problem. For example, we could also consider an oracle $\MVSTAT(L,T)$ that takes a query $h:X\rightarrow\{0,\ldots,L-1\}$ and a set $\mathcal{S}$ of subsets of $\{0,\ldots,L-1\}$ and returns a vector $v \in \mathbb{R}^L$ s.t for every $Z \in \mathcal{S}$
\[\left| \sum_{i \in Z} v[i]-p_Z\right| \leq \max\left\{\frac{1}{T},\sqrt{\frac{p_Z\paren{1-p_Z}}{T}} \right\}, \]
 where $p_Z = \Pr_{x\sim D}\left[h(x) \in Z\right]$ and the cost of the query is $\left|\mathcal{S}\right|$. Feldman et al. \cite[Theorem 7]{feldman2013complexity} proved lower bounds similar to Theorem \ref{thm:avgvstat-random} for the $\MVSTAT$ oracle. In this paper we focus on the $\MSAMPLE$ and $\VSTAT$ oracles for simplicity of presentation.
\end{remark}

\subsection{Statistical Dimension Lower Bounds} \label{subsec:StatisticalDimensionLowerBounds}
We are now ready to state our main technical result\footnote{We remark that for our particular family of human computable functions $f_{k_1,k_2}$ we could get a theorem similar to Theorem \ref{thm:SecurityLowerBound} by selecting $\sigma \sim \{0,5\}^n$ and appealing directly to results of Feldman et al. \cite{feldman2013complexity}. However, this theorem would be weaker than Theorem \ref{thm:SecurityLowerBound} as it would only imply that a statistical algorithm cannot find an assignment $\sigma'$ that is $\frac{1}{2}-\frac{1}{10}+\epsilon$-correlated with $\sigma$ for $\epsilon>0$. In contrast, our theorem implies that we cannot find $\sigma'$ that is $\epsilon$-correlated.}.

\newcommand{\thmSecurityLowerBound}{Let $\sigma \in \mathbb{Z}_d^n$ denote a secret mapping chosen uniformly at random, let $Q^f_\sigma$ be the distribution over $X_k \times \mathbb{Z}_d$ induced by a function $f:\mathbb{Z}_d^k \rightarrow \mathbb{Z}_d$ with distributional complexity $r = r(f)$. Any randomized statistical algorithm that finds an assignment $\tau$ such that $\tau$ is $\paren{\sqrt{\frac{-2\ln \paren{ \eta/2}}{n}}}$-correlated with $\sigma$ with probability at least $\Lambda > \eta$ over the choice of $\sigma$ and the internal randomness of the algorithm needs at least $m$ calls to the  $\MSAMPLE(L)$ oracle (resp. $\VSTAT\paren{\frac{n^r}{2\paren{\log n}^{2r}}}$ oracle) with $m\cdot L \geq c_1 \paren{\frac{n}{\log n}}^r$ (resp. $m \geq n^{c_1 \log n}$) for a constant $c_1 = \Omega_{k,1/\paren{\Lambda-\eta}}(1)$ which depends only on the values $k$ and $\Lambda-\eta$.  In particular if we set $L = \paren{\frac{n}{\log n}}^{r/2}$ then our algorithms needs at least $m \geq c_1 \paren{\frac{n}{\log n}}^{r/2}$ calls to $\MSAMPLE(L)$.}
\begin{theorem} \label{thm:SecurityLowerBound}
\thmSecurityLowerBound
\end{theorem}

The proof of Theorem \ref{thm:SecurityLowerBound} follows from Theorems \ref{thm:StatisticalDimension} and \ref{thm:avgvstat-random}. Theorems \ref{thm:SecurityLowerBound} and \ref{thm:StatisticalDimension} generalize results of Feldman et al. \cite{feldman2013complexity} which only apply for binary predicates $f:\{0,1\}^k\rightarrow\{0,1\}$. An interested reader can find our proofs in Appendix \ref{apdx:StatisticalDimension}. At a high level our proof proceeds as follows: Given any function $h:X_k\times \mathbb{Z}_d \rightarrow \mathbb{R}$ we show that $\Delta\paren{\sigma, h}$ can be expressed in the following form:
$\Delta\paren{\sigma, h} = \sum_{\ell=r(f)}^k \frac{1}{\left|X_\ell\right|} b_\ell\paren{\sigma}$,
where $\left| X_\ell\right| = \Theta\paren{n^{\ell}}$ and each function $b_\ell$ has degree $\ell$ (Lemma \ref{lemma:degree}). We then use the general hypercontractivity theorem \cite[Theorem 10.23]{o2007analysis} to obtain the following concentration bound.

\newcommand{\lemmaConcentrationBoundRestricted}{Let $b:\mathbb{Z}_d^n\rightarrow \mathbb{R}$ be any function with degree at most $\ell$, and let $\D' \subseteq \mathbb{Z}_d^n$ be a set of assignments for which $d' = d^n/\left|\D'\right| \geq e^\ell$. Then $\mathbb{E}_{\sigma\sim \D'} \left[\left| b\paren{\sigma}\right| \right] \leq  2 (\ln d'/c_0)^{\ell/2} \|b \|_2$, where $c_0 = \ell\paren{\frac{1}{2ed}}$ and $\|b \|_2 = \sqrt{\mathbb{E}_{x\sim \mathbb{Z}_d^n }\left[b\paren{x}^2 \right]}$.}
\begin{lemma} \label{lemma:concentrationRestricted}
\lemmaConcentrationBoundRestricted
\end{lemma}

 We then use Lemma \ref{lemma:concentrationRestricted} to bound $\mathbf{E}_{\sigma \sim \D'}\left[\Delta\paren{\sigma, h} \right]$ for any set $\D' \subseteq \mathbb{Z}_d^n$ such that $\left|\D'\right|  = \left| \mathbb{Z}_d^k\right|/d'$ (Lemma \ref{lemma:difference}). This leads to the following bound on $\dc(\D') = O_k\paren{\paren{\ln d'/n}^{r(f)/2}}$. 

\newcommand{\thmStatisticalDimension}{There exists a constant $c_Q > 0$ such that for any $\close > 1/\sqrt{n}$ and $q \geq n$ we have \[
\SDN\paren{\Z_{\close,f},\frac{c_Q\paren{\log q}^{r/2}}{n^{r/2}}, 2e^{-n\cdot \close^2/2}} \geq q \ ,\]
where $r=r(f)$ is the distributional complexity of $f$.}

\begin{theorem} \label{thm:StatisticalDimension}
\thmStatisticalDimension
\end{theorem}

\paragraph{Discussion} We view Theorem \ref{thm:SecurityLowerBound} as strong evidence for RP-CSP (Conjecture \ref{conj:hardnessConjecture}) because almost all known algorithmic techniques can be modeled within the statistical query framework\cite{blum2005practical,chu2007map}. Thus, Theorem \ref{thm:SecurityLowerBound} rules out most known attacks that an adversary might mount. It also implies that many popular heuristic based SAT solvers (e.g., DPLL\cite{Davis:1960:CPQ:321033.321034}) will not be able to recover $\sigma$ in polynomial time. While Theorem \ref{thm:SecurityLowerBound} does not rule our attacks based on Gaussian Elimination we consider this class of attacks separately. We need $m=\tilde{O}\paren{n^{g(f)+1}}$ examples to extract $O(n)$ linear constraints and solve for $\sigma$ (see Appendix \ref{subsec:GuassianElimination}). However, our composite security parameter $s(f)\geq g(f)+1$ accounts for attacks based on Gaussian Elimination.


\section{Security Analysis} \label{sec:SecurityAnalysis}
In the last section we presented evidence in support of RP-CSP (Conjecture \ref{conj:hardnessConjecture}) by showing that any statistical adversary needs $m = \tilde{\Omega}\paren{n^{r(f)/2}}$ examples to (approximately) recover $\sigma$. However, RP-CSP only says that it is hard to (approximately) recover the secret mapping $\sigma$, not that it is hard to forge passwords. As an example consider the following NP-hard problem from learning theory: find a 2-term DNF that is consistent with the labels in a given dataset. Just because 2-DNF is hard to learn in the {\em proper} learning model does not mean that it is NP-hard to learn a good classifier for 2-DNF. Indeed, if we allow our learning algorithm to output a linear classifier instead of a 2-term DNF then 2-DNF is easy to learn \cite{kearns1994introduction}. Could an adversary win our password security game without properly learning the secret mapping? 

Theorem \ref{thm:secFinal}, our main result in this section, implies that the answer is no. Informally, Theorem \ref{thm:secFinal} states that any adversary that breaks UF-RCA security of our human computable password scheme $f_{k_1,k_2}$ can also (approximately) recover the secret mapping $\sigma$. This implies that our human computable password scheme is UF-RCA secure as long as RP-CSP holds. Of course, for some functions it is very easy to predict challenge-response pairs without learning $\sigma$. For example, if $f$ is the constant function --- or any function highly correlated with the constant function --- then it is easy to predict the value of $f\paren{\sigma\paren{C}}$. However, any function that is highly correlated with a constant function is a poor choice for a human computable passwords scheme. We argue that any adversary that can win the  password game can be converted into an adversary that properly learns $\sigma$ provided that the output of function $f$ is evenly distributed (Definition \ref{def:unpredictableFunction}). 

\begin{definition} \label{def:unpredictableFunction}
We say that the output of a function $f:\mathbb{Z}_d^k\rightarrow \mathbb{Z}_d$ is {\em evenly distributed} if there exists a function  $g:\mathbf{Z}_d^{k-1}\rightarrow \mathbb{Z}_d$  such that $f\paren{x_1,\ldots,x_k} = g\paren{x_1,\ldots,x_{k-1}} + x_k \mod{d}$. 
\end{definition}

Clearly, our family $f_{k_1,k_2}$ has evenly distributed output. To see this we simply set $g = f_{k_1,k_2-1}$. We are now ready to state our main result from this section.

\newcommand{\thmSecurityFinal}{Suppose that $f$ has evenly distributed output, but that $f$ is not $\UFRCASecure{n}{m}{\challengelength}{\delta}$ for $\delta > \paren{\frac{1}{d}+\epsilon}^\challengelength$. Then there is a probabilistic polynomial time algorithm (in $n$, $m$, $\challengelength$ and $1/\epsilon$) that extracts a string $\sigma' \in \mathbb{Z}_d^n$ that is $\epsilon/8$-correlated with $\sigma$ with probability at least $\frac{\epsilon^3}{(8d)^2}$ after seeing $m+\challengelength$ example challenge response pairs.}
\begin{theorem} \label{thm:secFinal}
\thmSecurityFinal
\end{theorem}

The proof of Theorem \ref{thm:secFinal} is in Appendix \ref{apdx:securityProofs}. We overview the proof here. The proof of Theorem \ref{thm:secFinal} uses Theorem \ref{thm:EvenDistributedPredict} as a subroutine. Theorem \ref{thm:EvenDistributedPredict} shows that we can, with reasonable probability, find a mapping $\sigma'$ that is correlated with $\sigma$ given predictions of $f\paren{\sigma\paren{C}}$ for each clause as long as the probability that each prediction is accurate is slightly better than a random guess (e.g., $\frac{1}{d}+\delta$). The proof of Theorem \ref{thm:EvenDistributedPredict} is in Appendix \ref{apdx:securityProofs}. 

\newcommand{\ThmEvenDistributedPredict}{Let $f$ be a function with evenly distributed output (Definition \ref{def:unpredictableFunction}), let $\sigma \sim \mathbb{Z}_d^n$ denote the secret mapping, let $\epsilon > 0$ be any constant and suppose that for every $C \in X_k$ we are given labels $\ell_C \in \mathbb{Z}_d$ s.t. 
$\Pr_{C \sim X_k}\left[f\paren{\sigma\paren{C}} = \ell_C \right] \geq \frac{1}{d} + \epsilon$.
There is a polynomial time algorithm (in $n$, $m$,$1/\epsilon$) that finds a mapping $\sigma' \in \mathbb{Z}_d^n$ such that $\sigma'$ is $\epsilon/2$-correlated with $\sigma$ with probability at least $\frac{\epsilon}{2d^2}$}

\begin{theorem}\label{thm:EvenDistributedPredict}
\ThmEvenDistributedPredict
\end{theorem}

The remaining challenge in the proof of Theorem \ref{thm:secFinal} is to show that there is an efficient algorithm to extract predictions of $f\paren{\sigma\paren{C}}$ given blackbox access to an adversary $\Adversary$ that breaks UF-RCA security. However, just because the adversary $\Adversary$ gives the correct response to an entire password challenge $C_1,\ldots,C_\challengelength$ with probability greater than $\paren{\frac{1}{d}+\epsilon}^\challengelength$ it does not mean that the response to each individual challenge $C$ is correct with probability $\frac{1}{d}+\epsilon$. To obtain our predictions for individual clauses $C$ we draw $\challengelength$ extra example challenge response pairs $\paren{C_1',f\paren{\sigma\paren{C_1'}}},\ldots,\paren{C_\challengelength',f\paren{\sigma\paren{C_\challengelength'}}}$, which we use to check the adversary. To obtain the label for a clause $C$ we select a random index $i \in [\challengelength]$ and give $\Adversary$ the password challenge $C_1',\ldots,C_\challengelength'$, replacing $C_i'$ with $C$. If for some $j < i$ the label for clause $C_j'$ is not correct (e.g., $\neq f\paren{\sigma\paren{C_j'}}$) then we discard the label and try again. Claim \ref{claim:secFinal} in Appendix \ref{apdx:securityProofs} shows that this process will give us predictions for individual clauses that are accurate with probability at least $\frac{1}{d}+\epsilon$.


\section{Related Work} \label{sec:related}
The literature on passwords has grown rapidly over the past decade (e.g., see \cite{massey1994guessing,pliam2000incomparability,bonneau2012science,boztas1999entropies,usability:compositionPolicies,blocki2013optimizing}.) Perhaps  most related to our paper is the work of Blocki et al. \cite{NaturallyRehearsingPasswords,spacedRepetitionAndMnemonics} and Blum and Vempala~\cite{blum2015publishable} on developing usable and secure password management schemes. While the password management schemes proposed in these works are easier to use (e.g., involve less memorization and/or computation) than our human computable password scheme, these schemes only remain secure up to their information theoretic limit --- after a very small (e.g., $1$--$6$) number of breaches security guarantees start to break down. By contrast, our schemes remain secure after a large (e.g., $100$) number or breaches. 

In contrast to our work, password management software (e.g., PwdHash~ \cite{ross2005stronger} or KeePass~\cite{reichl2013keepass}) relies strong trust assumptions about the user's computational devices. The recent breach at LastPass\footnote{See \url{https://blog.lastpass.com/2015/06/lastpass-security-notice.html/} (Retrieved 9/1/2015).} highlights the potential danger of such strong assumptions.

Hopper and Blum \cite{hopper2001secure} designed a Human Identification Protocol based on noisy parity, a learning problem  that is believed to be hard\footnote{Subsequent work~\cite{juels2005authenticating,gilbert2005active,bringer2006hb,katz2006parallel} has  explored the use of the Hopper-Blum protocol for authentication on pervasive devices like smartcards. }. We emphasize a few fundamental differences between our work and the work of Hopper and Blum. First, a single digit challenge in their protocol consists of an $n$-digit vector $x \in \mathbb{Z}_{10}^n$ and the user responds with the $\mod{10}$ sum of the digits at $\ell \leq n$ secret locations (occasionally the user is supposed to respond with a random digit instead of the correct response so that the adversary cannot simply use Gaussian Elimination to find the secret locations). By contrast, a single digit challenge in our protocol consists of an ordered clause of length $k \ll n$. Second, their protocols allow for an $O\big(n^{\ell/2}\big)$-time attack called Meet-In-The-Middle~\cite{hopper2001secure} after the adversary has seen $\tilde{O}( \log {n \choose \ell})$ challenge-response pairs. Thus, it is critically important to select $\ell$ sufficiently large (e.g., $\ell = \Omega(\log(n))$) in the Hopper-Blum protocol to defend against this Meet-In-The-Middle attack. By contrast, we focus on computation of {\em very simple} functions over a {\em constant} number of variables so that a human can compute the response to each challenge quickly. In particular, we provide strong evidence that our scheme is secure against any polynomial time attacker even if the adversary has seen up to $O\left(n^{c\cdot k} \right)$ challenge-response pairs for some constant $c\ge 1$. Finally, computations in our protocols are deterministic. This is significant because humans are not good at consciously generating random numbers \cite{wagenaar1972generation,humanRandom:figurska2008humans,seventeenMostRandom} (e.g., noisy parity could be easy to learn when humans are providing source of noise)\footnote{Hopper and Blum also proposed a deterministic variant of their protocol called sum of $k$-mins, but this variant is {\em much} less secure. See additional discussion in the appendix.}.

Naor and Pinkas\cite{naor1997visual} proposed using visual cryptography\cite{naor1995visual} to address a related problem: how can a human verify that a message he received from a trusted server has not been tampered with by an adversary? Their protocol requires the human to carry a visual transparency (a shared secret between the human and the trusted server in the visual cryptography scheme), which he will use to verify that messages from the trusted server have not been altered.

A related goal in cryptography, constructing pseudorandom generators in {\sf $NC^0$}, was proposed by Goldreich~\cite{goldreich2000candidate} and by Cryan and Miltersen~\cite{cryan2001pseudorandom}. In Goldreich's construction we fix $C_1,\ldots,C_m \in [n]^k$ once and for all, and a binary predicate $P:\{0,1\}^k\rightarrow\{0,1\}$. The pseudorandom generator is a function $G:\{0,1\}^n\rightarrow\{0,1\}^m$, whose $i$'th bit $G(x)[i]$ is given by $P$ applied to the bits of $x$ specified by $C_i$. O'Donnel and Witmer gave evidence that the ``Tri-Sum-And'' predicate ($TSA\paren{x_1,\ldots,x_5} = x_1 + x_2 +x_3 + x_4x_5 \mod{2}$) provides near-optimal stretch. In particular, they showed that for $m=n^{1.5-\epsilon}$ Goldreich's construction with the TSA predicate is secure against subexponential-time attacks using SDP hierarchies. Our candidate human-computable password schemes use functions $f:\mathbb{Z}_{10}^k \rightarrow \mathbb{Z}_{10}$ instead of binary predicates. While our candidate functions are contained in {\sf $NC^0$}, we note that an arbitrary function in {\sf $NC^0$} is not necessarily human computable.

On a technical level our statistical dimension lower bounds extend work of
Feldman et al. \cite{feldman2013complexity}, who considered the problem of finding a planted solution in a random {\em binary} planted constraint satisfiability problem. We extend their analysis to handle non-binary planted constraint satisfiability problems, and argue that our candidate human computable password schemes are secure. We will discuss this work in more detail later in the paper. 


\section{Discussion} \label{sec:OpenQuestions}

\subsection{Improving Response Time}
The easiest way to improve response time is to decrease $\challengelength$ --the number of single-digit challenges that the user needs to solve. However, if the user wants to ensure that each of his passwords are strong enough to resist offline dictionary attacks then he would need to select a larger value of $\challengelength$ (e.g., $\challengelength \geq 10$). Fortunately, there is a natural way to circumvent this problem. The user could save time by memorizing a mapping $w:\mathbb{Z}_{10} \rightarrow \left\{x ~\vline   ~\mbox{$x$ is one of 10,000 most common english words} \right\}$ and responding to each challenge $C$ with $w\paren{f\paren{\sigma\paren{C}}}$ -- the word corresponding to the digit $f\paren{\sigma\paren{C}}$. Now the user can create passwords strong enough to resist offline dictionary attacks by responding to just $3$--$5$ challenges. Even if the adversary learns the words in the user's set he won't be able to mount online attacks. Predicting $w\paren{f\paren{\sigma\paren{C}}}$ is at least as hard as predicting $f\paren{\sigma\paren{C}}$  even if the adversary knows the exact mapping $w$\footnote{In fact, it is quite likely that it is much harder for the adversary to predict $w\paren{f\paren{\sigma\paren{C}}}$  because the adversary will not see which word corresponds with each digit.}.


\subsection{One-Time Challenges}

\paragraph{Malware} Consider the following scenario: the adversary infects the user's computer with a keylogger which is never detected over the user's lifetime. We claim that it is possible to protect the user in this extreme scenario using our scheme by generating multiple (e.g., $10^6$) one-time passwords for each of the user's accounts. When we initially generate the secret mapping $\sigma \sim \mathbb{Z}_d^{n}$ we could also generate cryptographic hashes for millions of one-time passwords $\Hash\paren{\vec{C},f_{k_1,k_2}\paren{\sigma\paren{\vec{C}}}}$.  While usability concerns make this approach infeasible in a traditional password scheme (it would be far too difficult for the user to memorize a million one-time passwords for each of his accounts), it may be feasible to do this using a human computable password scheme. In our human computable password scheme we could select $k_1$ and $k_2$ large enough that $s\paren{f_{k_1,k_2}} = \min\{k_1+1,(k_2+1)/2\} \geq 6$. Assuming that the user authenticates fewer than $10^6$ times over his lifetime a polynomial time adversary would never obtain enough challenge-response examples to learn $\sigma$. The drawback is that $f_{k_1,k_2}$ will take longer for a user to execute in his head. 

\paragraph{Secure Cryptography in a Panoptic World} Standard cryptographic algorithms could be easily broken in a panoptic world  where the user only has access to a semi-trusted computer (e.g., if a user asks a semi-trusted computer to sign a message $m$ using a secret key $sk$ stored on the hard drive then the computer will respond with the correct value $Sign\paren{sk,m}$, but the adversary will learn the values $m$ and $sk$). Our human computable password schemes could also be used to secure some cryptographic operations (e.g., signatures) in a panoptic world by leveraging recent breakthroughs in program obfuscation \cite{garg2013candidate}. The basic idea is to obfuscate a ``password locked" circuit $P_{\sigma,sk,r}$ that can sign messages under a secret key $sk$ --- we need a trusted setup phase for this step. The circuit $P_{\sigma,sk,r}$ will only sign a message $m$ if the user provides the correct response to a unique (pseudorandomly generated) challenge for $m$.

\subsection{Open Questions}

\paragraph{Eliminating the Semi-Trusted Computer} Our current scheme relies on a semi-trusted computer to generate and store random public challenges. An adversary with full control over the user's computer might be able to extract the user's secret if he is able to see the user's responses to $O(n)$ adaptively selected password challenges. Can we eliminate the need for a semi-trusted computer? 

\cut{\paragraph{Cryptography for Humans?} We presented a scheme that allows a human to create multiple passwords from a short secret by responding to public challenges. What other cryptographic functions could be performed by humans? Could we develop a human computable protocol for bit commitment or coin flipping? What about a human computable encryption scheme?}

\paragraph{Exact Security Bounds} While we provided asymptotic security proofs for our human computable password schemes, it is still important to understand how much effort an adversary would need to expend to crack the secret mapping for specific values of $n$ and $m$. Our attacks with a SAT solver (see Appendix \ref{apdx:CSPSolver}) indicate that the value $n=26$ is too small to provide UF-RCA security even with small values of $m$ (e.g., $m=50$). As $n$ increases the problem rapidly gets harder for our SAT solver (e.g., with $n=50$ and $m=1000$ the solver failed to find $\sigma$).  We also present a public challenge with specific values of $n$ and $m$ to encourage cryptography and security researchers to find other techniques to attack our scheme.

\cut{
\paragraph{Improving Usability}
Is it possible to improve usability by designing a human computable function $f:\mathbb{Z}_{10}^k \rightarrow \mathbb{Z}_{100}$? This could potentially allow the user to generate a secure length $t$ password after responding to only $t/2$ challenges. Our statistical dimension lower bounds also hold for functions $f:\mathbb{Z}_{d_1}^k\rightarrow\mathbb{Z}_{d_2}$. As before the challenge would be designing a function that is human computable and has strong security properties (e.g., $s(f)$ is large and $f$ is $\paren{\delta_1,\delta_2}$-hard to guess). \cut{We would also like to conduct user studies to see how long it takes a typical user to compute passwords using our challenge response schemes.}}

\bibliographystyle{alpha}
\bibliography{humanComputablePasswords}
\appendix

\section{Authentication Process} \label{apdx:AuthenticationProcess}
In this section of the appendix we illustrate our human computable password schemes graphically. In our examples, we use the function $f = f_{2,2}$. To compute the response $f\paren{\sigma\paren{C}}$ to a challenge $C = \{x_0,\ldots,x_{13}\}$ the user computes $f\paren{\sigma\paren{C}}=\sigma\paren{x_{\sigma\paren{x_{10}}+\sigma\paren{x_{11}}\mod{10}}} + \sigma\paren{x_{12}} + \sigma\paren{x_{13}} \mod{10}$

\paragraph{Memorizing a Random Mapping.} To begin using our human computable password schemes the user begins by memorizing a secret random mapping $\sigma:[n]\rightarrow \{0,\ldots,9\}$ from $n$ objects (e.g., letters, pictures) to digits. See Figure \ref{fig:HumanComputablePasswordsRandomMapping} for an example.

The computer can provide the user with mnemonics to help memorize the secret mapping $\sigma$ --- see Figures \ref{fig:MnemonicEagle2} and \ref{fig:MnemonicEagle6}. For example, if we wanted to help the user remember that $\sigma\paren{eagle} = 2$ we would show the user Figure \ref{fig:MnemonicEagle2}. We observe that a $10\times n$ table of mnemonic images would be sufficient to help the user memorize {\em any} random mapping $\sigma$. We stress that the computer will only save the original image (e.g., Figure \ref{fig:MnemonicEagle}). The mnemonic image (e.g., Figure \ref{fig:MnemonicEagle2} or \ref{fig:MnemonicEagle6}) would be discarded after the user memorizes $\sigma\paren{eagle}$.

\begin{figure}
\centering
\begin{subfigure}[b]{0.30\textwidth}
\includegraphics[scale=0.4]{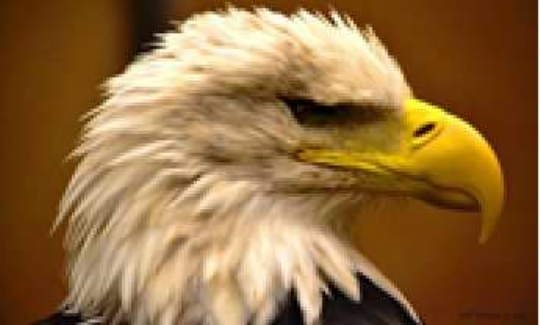}
\caption{Original photo (an eagle).\\}
\label{fig:MnemonicEagle}
\end{subfigure}
\begin{subfigure}[b]{0.30\textwidth}
\includegraphics[scale=0.4]{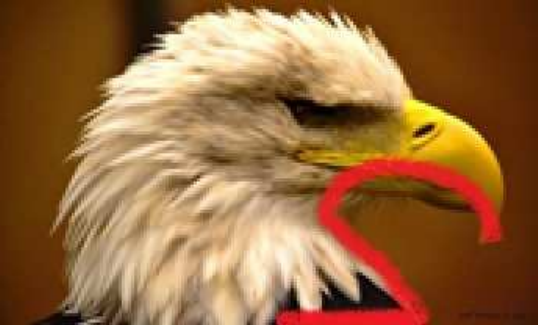}
\caption{Mnemonic to help the user remember \\ $\sigma\paren{eagle} = 2$.}
\label{fig:MnemonicEagle2}
\end{subfigure}
\begin{subfigure}[b]{0.30\textwidth}
\includegraphics[scale=0.4]{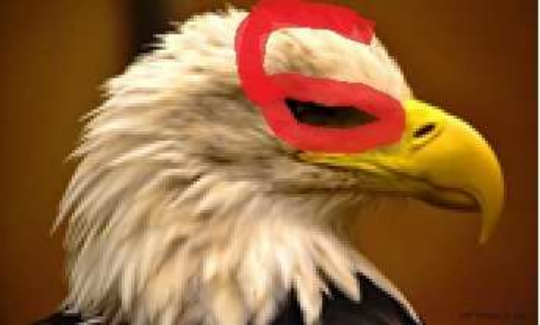}
\caption{Mnemonic to help the user remember \\ $\sigma\paren{eagle} = 6$.} \label{fig:MnemonicEagle6}
\end{subfigure}
\caption{Mnemonics to help memorize the secret mapping $\sigma$}
\end{figure}

\begin{figure}
\centering
\begin{subfigure}[b]{0.45\textwidth}
\includegraphics[scale=0.35]{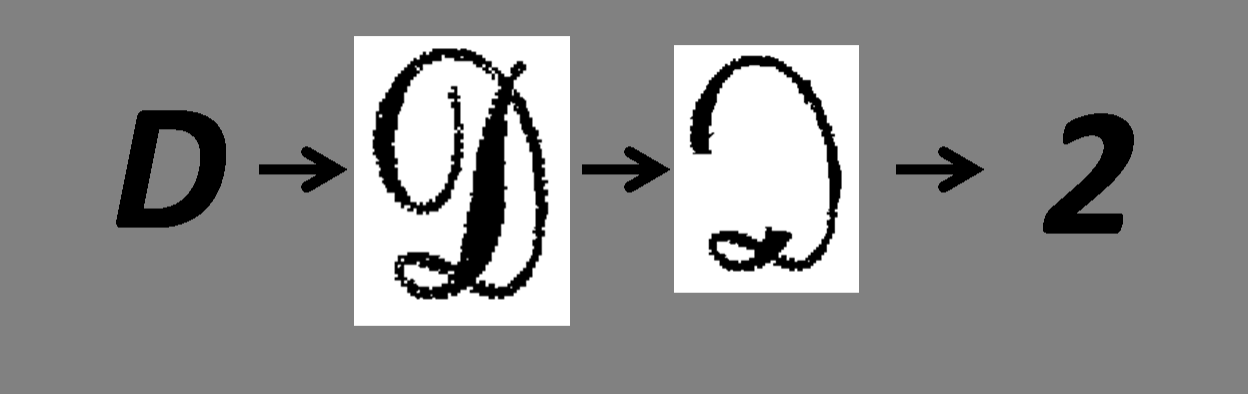}
\caption{$M_{D,2}$} \label{fig:MnemonicDto2}
\end{subfigure}
\begin{subfigure}[b]{0.45\textwidth}
\includegraphics[scale=0.35]{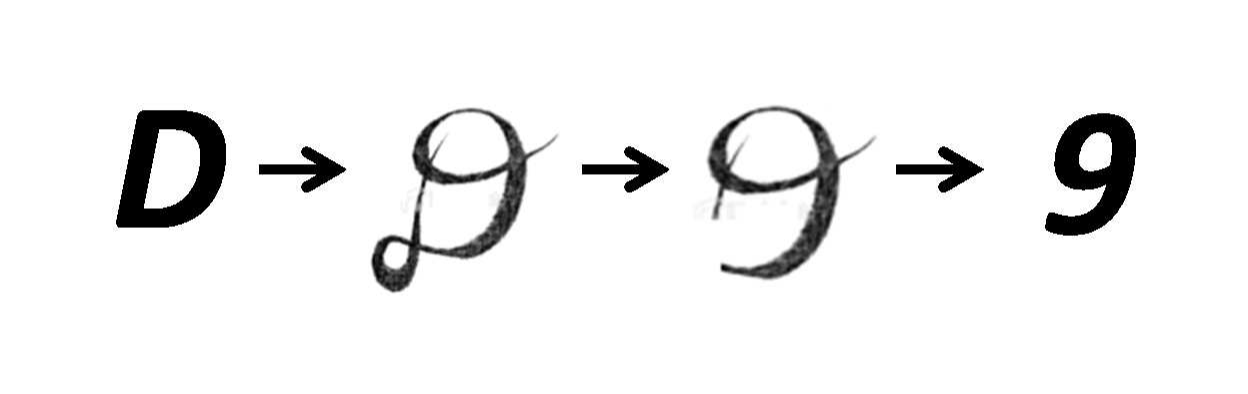}
\caption{$M_{D,9}$} \label{fig:MnemonicDto9}
\end{subfigure}
\caption{Mnemonics to help the user memorize the secret mapping $\sigma$}
\end{figure}

\begin{figure}
\centering
\includegraphics[scale=0.4]{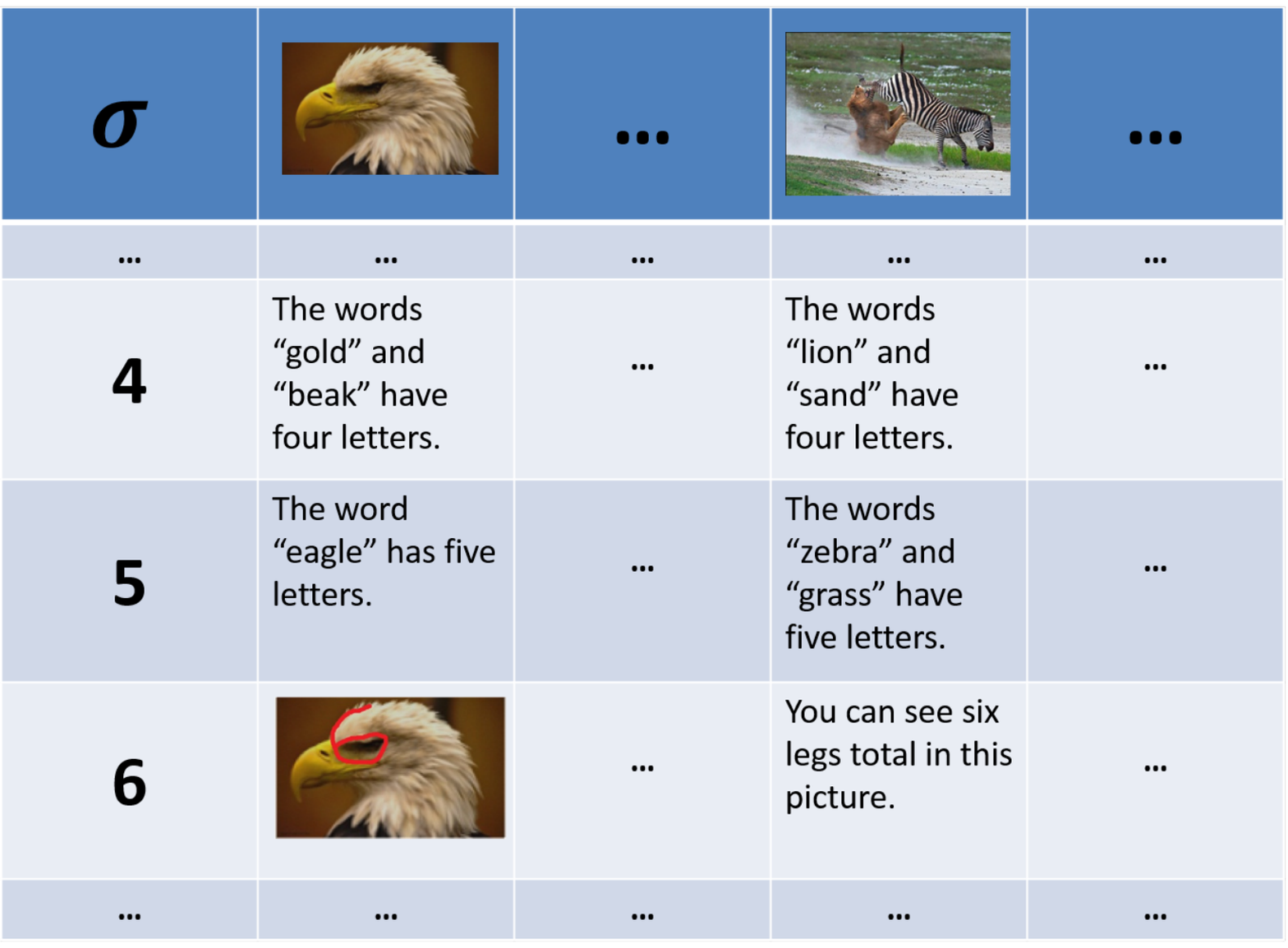}
\caption{Table of Mnemonic Helpers to Help Learn Any Secret Mapping}
\label{fig:MnemonicTable}
\end{figure}

\paragraph{Single-Digit Challenges.} In our scheme the user computes each of his passwords by responding to a sequence of single-digit challenges. For $f = f_{2,2}$ a single-digit challenge is a tuple $C \in [n]^{14}$ of fourteen objects. See Figure \ref{fig:HumanComputablePasswordsDigitChallenge} for an example. To compute the response $f\paren{\sigma\paren{C}}$ to a challenge $C = \{x_0,\ldots,x_{13}\}$ the user computes $f\paren{\sigma\paren{C}}=\sigma\paren{x_{\sigma\paren{x_{10}}+\sigma\paren{x_{11}}\mod{10}}} + \sigma\paren{x_{12}} + \sigma\paren{x_{13}} \mod{10}$. Observe that this computation involves just three addition operations modulo ten. See Figure \ref{fig:HumanComputablePasswordsDigitChallengeAndResponse} for an example. In this example the response to the challenge $C= \{x_0 =\mbox{burger},x_1=\mbox{eagle},\ldots,x_{10}=\mbox{lightning}, x_{11}=\mbox{dog},\allowbreak x_{12}=\mbox{man standing on world},\allowbreak x_{13}=\mbox{kangaroo}\}$ is \begin{eqnarray*}
f\paren{\sigma\paren{C}} &=& \sigma\paren{x_{\sigma\paren{x_{10}}+\sigma\paren{x_{11}}\mod{10}}} + \sigma\paren{x_{12}} + \sigma\paren{x_{13}} \mod{10} \\
& =& \sigma\paren{x_{\sigma\paren{\mbox{lightning}}+\sigma\paren{\mbox{dog}}\mod{10}}} \\ & &+   \sigma\paren{\mbox{man standing on world}} + \sigma\paren{\mbox{kangaroo}} \mod{10}   \\
& =& \sigma\paren{x_{9+3\mod{10}}} + \sigma\paren{\mbox{man standing on world}} + \sigma\paren{\mbox{kangaroo}} \mod{10}   \\
 &=& \sigma\paren{\mbox{minions}} + \sigma\paren{\mbox{man standing on world}} + \sigma\paren{\mbox{kangaroo}} \mod{10}  \\
&=& 7+4+5 \mod{10} = 6 \ . \end{eqnarray*}
 We stress that this computation is done entirely in the user's head. It takes the main author $7.5$ seconds on average to compute each response.

\begin{figure}
\centering
\includegraphics[scale=0.5]{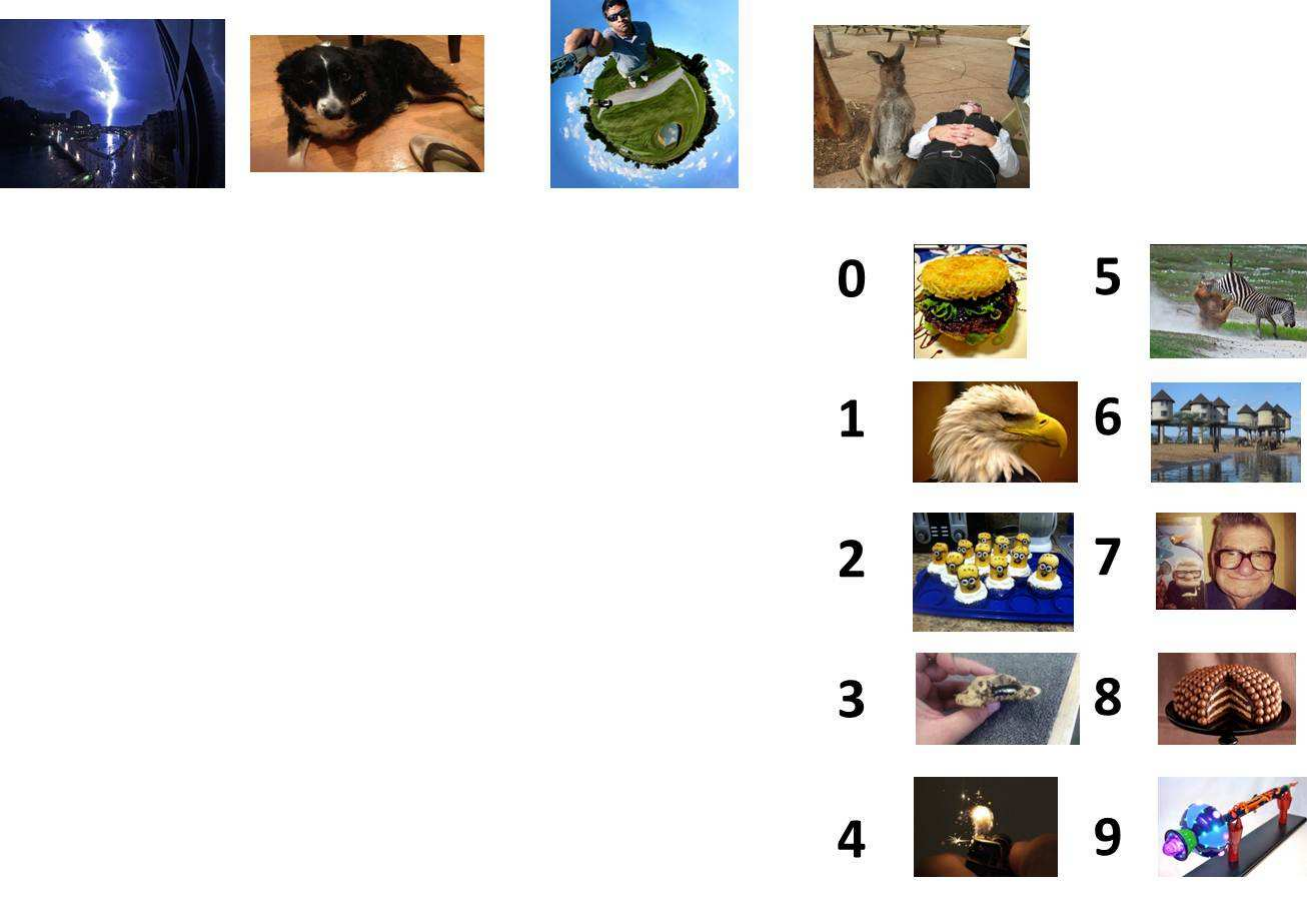}
\caption{A single-digit challenge}
\label{fig:HumanComputablePasswordsDigitChallenge}
\end{figure}

\paragraph{Creating an Account.} To help the user create an account the computer would first pick a sequence of single-digit challenges $C_1,\ldots, C_\challengelength$, where the security parameter is typically $\challengelength=10$, and would display the first challenge $C_1$ to the user --- see Figure \ref{fig:HumanComputablePasswordsLoginScreen1} for an example. To compute the first digit of his password the user would compute $f\paren{\sigma\paren{C_1}}$. After the user types in the first digit $f\paren{\sigma\paren{C_1}}$ of his password the computer will display the second challenge $C_2$ to the user --- see Figure \ref{fig:HumanComputablePasswordsLoginScreen2}. After the user creates his account the computer will store the challenges $C_1,\ldots,C_{10}$ in public memory. The password $pw = f\paren{\sigma\paren{C_1}}\ldots f\paren{\sigma\paren{C_\challengelength}}$ will not be stored on the user's local computer (the authentication server may store the cryptographic hash of $pw$).

\begin{figure}
\centering
\includegraphics[scale=0.5]{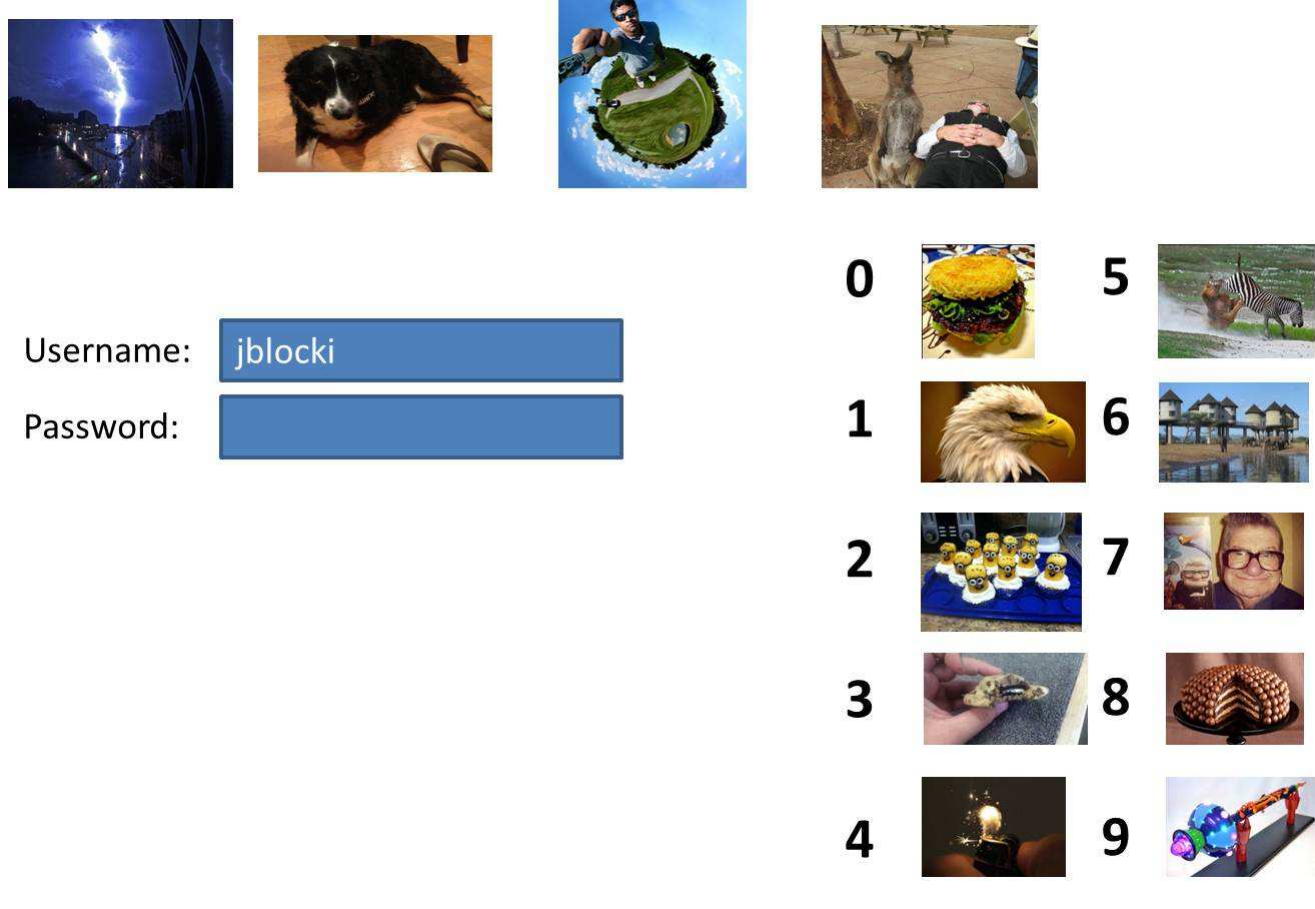}
\caption{Login Screen}
\label{fig:HumanComputablePasswordsLoginScreen1}
\end{figure}

\begin{figure}
\centering
\includegraphics[scale=0.5]{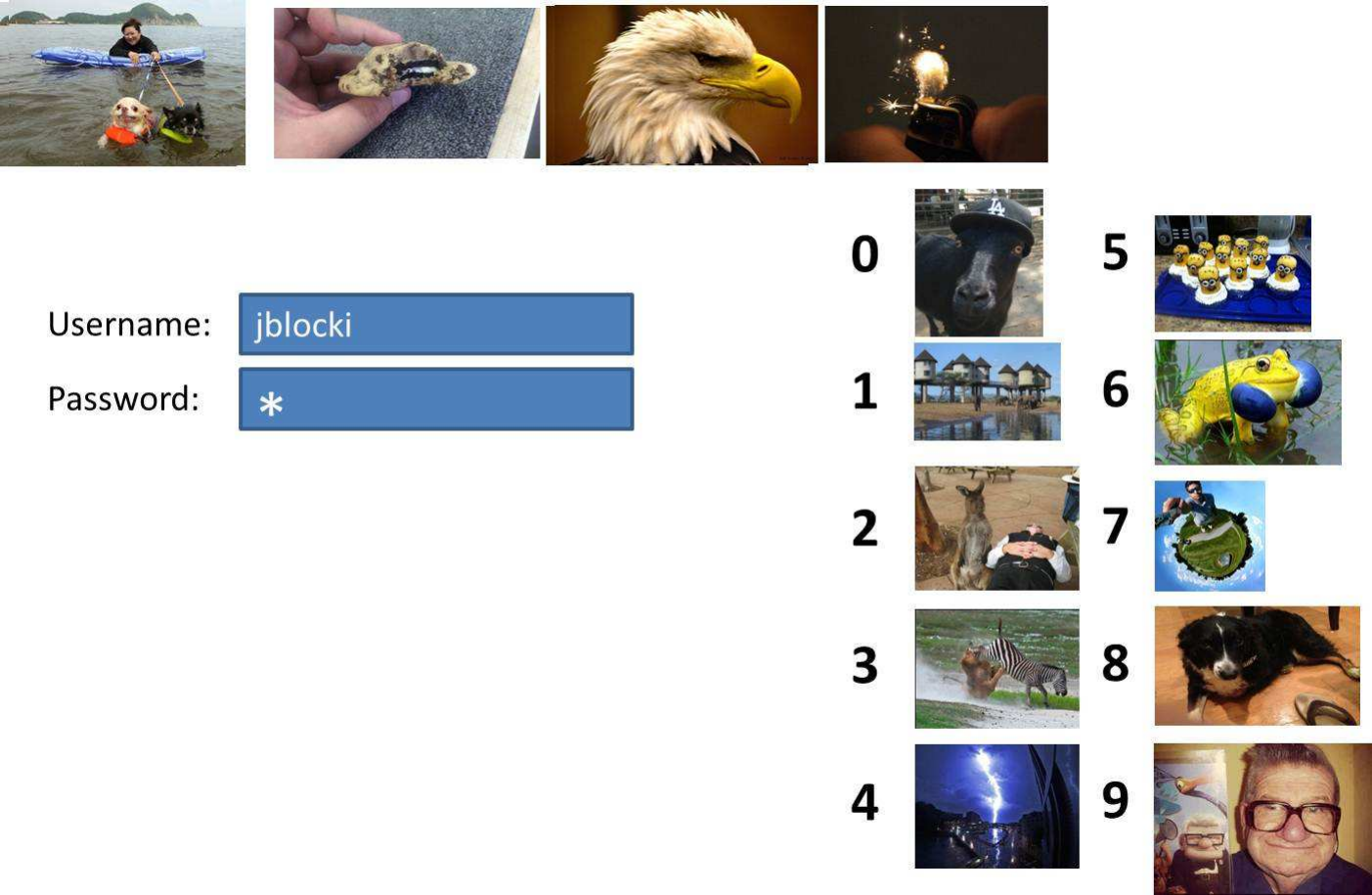}
\caption{Login Screen after the user responds to the first single-digit challenge}
\label{fig:HumanComputablePasswordsLoginScreen2}
\end{figure}

\paragraph{Authentication.} Authenticating is very similar to creating an account. To help the user recompute his password for an account the computer first looks up the challenges $C_1,\ldots,C_{\challengelength}$ which were stored in public memory, and the user authenticates by computing his password $pw = f\paren{\sigma\paren{C_1}}\ldots f\paren{\sigma\paren{C_\challengelength}}$. We stress that the single-digit challenges the user sees during authentication will be the same single-digit challenges that the user saw when he created the account. The authentication server verifies that the cryptographic hash of $pw$ matches its record.

\paragraph{Helping the user remember his secret mapping.} The computer keeps track of when the user rehearses each value of his secret mapping (e.g., $\paren{i,\sigma\paren{i}}$ for each $i \in [n]$), and reminds the user to rehearse any part of his secret mapping that he hasn't used in a long time. One advantage of our human computable password scheme (compared with the~\sharedCues~scheme of Blocki et al.\cite{NaturallyRehearsingPasswords}  ) is that most users will use each part of their secret mapping often enough that they will not need to be reminded to rehearse --- see discussion in Section \ref{subsec:MemorizeAndRehearse}. The disadvantage is that we require the user to spend extra effort computing his passwords each time he authenticates.

\subsection{Formal View}
We now present a formal overview of the authentication process. Algorithm \ref{alg:Memorize} outlines the initialization process in which the user memorizes a secret random mapping $\sigma$ generated by the user's computer, and Algorithm \ref{alg:GenStories} outlines the account creation process. In Algorithm \ref{alg:GenStories} the user  generates the password for an account $i$ by computing the response to a sequence of random challenges $C$ generated by the user's computer. The sequence of challenges are stored in public memory. We assume that all steps in algorithms \ref{alg:Memorize}, \ref{alg:GenStories} and \ref{alg:Authenticate} are executed on the user's local computer unless otherwise indicated. We also assume that the initialization phase (Algorithms \ref{alg:Memorize} and \ref{alg:GenStories}) is carried out in secret (e.g., we assume that the secret mapping is chosen in secret), but we do not assume that the challenges are kept secret.  We use {\bf (User)} to denote a step completed by the human user and we use {\bf (Server) } to denote a step completed by a third-party server. Algorithm \ref{alg:Authenticate} illustrates the authentication process. 

\begin{algorithm}[H]
\caption{$\mathbf{MemorizeMapping}$} \label{alg:Memorize}
\SetKwInOut{Input}{input}

\Input{ $I_1,...,I_n$, $d$, $b$ and $M_{i,j}$ for $i \in [n], j \in \{0,\ldots,d-1\}$.}
\emph{ Given $j \in \{0,\ldots, d-1\}$ and $i \in \{1,\ldots,n\}$ $M_{i,j}$ is a mnemonic to help the user associate image $I_i$ with the number $j$. $d$ is whatever base the user is familiar with (typically $d=10$), and the value $b$ contains random bits which are used to select the secret mapping $\sigma$. }\;

\BlankLine
\emph{Generate and Memorize Secret Mapping}\;

\For{$i\leftarrow 1$ \KwTo $n$}{
\tcp{\tt Using random bits b}
$\sigma\paren{i} \sim \{0,\ldots,d-1\}$ \;
$M_i \leftarrow M_{i,\sigma(i)}$ \;
{\bf (User) } Using $M_i$ memorizes the association $\paren{I_i,\sigma\paren{i}}$ for $i \in [n]$. \;
}
\end{algorithm}

\begin{algorithm}[H]
\caption{$\mathbf{CreateChallenge}$} \label{alg:GenStories}
\SetKwInOut{Input}{input}
\Input{ $n$, $i$, $\challengelength$, $d$, $b$ and $I_1,...,I_n$.}

\BlankLine
\emph{ Generate Password Challenge for account $A_i$. $t$ is a security parameter which specifies how many digits the password will contain.}\;

\For{$j=1 \to \challengelength$}{
$C^i_j \sim X_k$  \;  \tcp{\tt Using random bits b} 
}
$\vec{C}_i \leftarrow \left\langle C^i_1,\ldots,C^i_\challengelength\right\rangle$ \;
 Store record  $\paren{i,\vec{C}_i}$ \;
\For{$j=1 \to \challengelength$}{
		 Load images from $C_j^i$. \;
		Display $C_j^i$ for the user. \;
{\bf (User) } Computes $\guess{j} \leftarrow f\paren{\sigma\paren{C^j_i}}$ \;
}
 Send $\left\langle\guess{1},\ldots,\guess{\challengelength}\right\rangle =f\paren{\sigma\paren{\vec{C}_i}}$ to server $i$\;
\tcp{$\Hash$ is a strong cryptographic hash function} 
 {\bf (Server $i$)} Stores $h_i= \Hash\paren{\vec{C}_i,\left\langle\guess{1},\ldots,\guess{\challengelength}\right\rangle }$ \;

\end{algorithm}

\begin{algorithm}[H]
\caption{$\mathbf{Authenticate}$} \label{alg:Authenticate}

\SetKwInOut{Input}{input}
\Input{Security parameter $\challengelength$. Account $i \in [m]$. Challenges $\vec{C}_1,\ldots,\vec{C}_m$.}

$ \left\langle C^i_1,\ldots,C^i_\challengelength\right\rangle \leftarrow \vec{C}_i$ \;
\tcp{Display Single Digit Challenges}
\For{$j=1 \to \challengelength$}{
 {\bf (Semi-Trusted Computer)}  Load images from $C_j^i$.\;
		 {\bf (Semi-Trusted Computer)} Displays $C_j^i$ to the user. \;
 {\bf (User) } Computes $\guess{j} \leftarrow f\paren{\sigma\paren{C^j_i}}$. \;
}
 {\bf (Semi-Trusted Computer)} Sends $\left\langle \guess{1},\ldots,\guess{\challengelength}\right\rangle$ to the server for account $i$. \;
{\bf (Server)} Verifies that $\Hash\paren{\vec{C}_i,\left\langle \guess{1},\ldots,\guess{\challengelength}\right\rangle} = h_i $ \;

\end{algorithm}

\subsection{Memorizing and Rehearsing $\sigma$} \label{subsec:MemorizeAndRehearse}
After the user memorizes $\sigma$ he may need to rehearse parts of the mapping periodically to ensure that he does not forget it. How much effort does this require? Blocki et al.\cite{NaturallyRehearsingPasswords} introduced a usability model to estimate how much extra effort that a user would need to spend rehearsing the mapping $\sigma$. We used this model to obtain the predictions in Table \ref{tab:Usability}. 

Imagine that we had a program keep track of how often the user rehearsed each association $\paren{i,\sigma\paren{i}}$ and predict how much longer the user will safely remember the association $\paren{i,\sigma\paren{i}}$ without rehearsing again --- updating this prediction after each rehearsal. The user rehearses the association $\paren{i,\sigma\paren{i}}$ {\em naturally} whenever he needs to recall the value of $\sigma\paren{i}$ while computing the password for any of his password protected accounts. If the user is in danger of forgetting the value $\sigma\paren{i}$ then the program sends the user a reminder to rehearse (Blocki et al.\cite{NaturallyRehearsingPasswords} call this an {\em extra rehearsal}). Table \ref{tab:Usability} shows the value of  $\mathbb{E}\left[\TotalExtraRehearsals{365}\right]$, the expected number of extra rehearsals that the user will be required to do to remember the secret mapping $\sigma$  during the first year. This value will depend on how often the user rehearses $\sigma$ naturally. We consider four types of users: Active, Typical, Occasional and Infrequent. An Active user visits his accounts more frequently than a Infrequent user. See Appendix \ref{apdx:Rehearsal} for specific details about how Table \ref{tab:Usability} was computed.
\begin{table}
\parbox{.65\linewidth}{
\centering
\begin{tabular}{| l | l | l | l || l | l | l | }
\hline
 &\multicolumn{3}{c||}{Our Scheme $\paren{\sigma \in \mathbb{Z}_{10}^n}$} & \multicolumn{3}{|c|}{\sharedCues} \\
\hline
 User & $n=100$ & $n=50$ & $n=30$ & SC-0 & SC-1 & SC-2 \\
\hline
Very Active & $0.396$ & $0.001$ & $\approx 0$ & $\approx 0$ & $3.93$ & $7.54$ \\
\hline
Typical & $2.14$ & $0.039$ & $\approx 0$ & $\approx 0$ & $10.89$ & $19.89$ \\
\hline 
Occasional & $2.50$ & $0.053$ & $\approx 0$ & $\approx 0$ & $22.07$ & $34.23$ \\
\hline
Infrequent & $70.7$ & $22.3$ & $6.1 $ & $\approx 2.44$ & $119.77$ & $173.92$  \\
\hline
\end{tabular} 
\caption{$\mathbb{E}\left[\TotalExtraRehearsals{365}\right]$: Extra Rehearsals over the first year to remember $\sigma$ in our scheme with $f_{2,2}$ or $f_{1,3}$. Compared with  \sharedCues~schemes SC-0,SC-1 and SC-2\cite{NaturallyRehearsingPasswords}. }
\label{tab:Usability}
}
\hfill
\parbox{.23\linewidth}{
\begin{tabular}{| l | l || l | l | }
\hline
A & B & C & D \\
\hline \hline
0 & E & 5 & J   \\
1 & F & 6&K \\
2 & G & 7&L \\
3 & H & 8&M \\
4 & I & 9&N \\
\hline
\end{tabular}
\caption{Single-Digit Challenge Layout in Scheme 1} \label{tbl:challengeLayout}
}
\hfill
\end{table}

\paragraph{Discussion} One of the advantages of our human computable passwords schemes is that memorization is essentially a one time cost for our Very Active, Typical and Occasional users. Once the user has memorized the mapping $\sigma:\{1,...,n\}\rightarrow\mathbb{Z}_d$ he will get sufficient natural rehearsal to maintain this memory. In fact, our schemes require the user to expend {\em less} extra effort rehearsing his secret mapping than the \sharedCues~password management scheme of Blocki et al. \cite{NaturallyRehearsingPasswords} (with the exception of SC-0 -- the least secure \sharedCues~ scheme). Intuitively, this is because human computable password schemes require to recall $\sigma(i)$ for multiple different values of $i$ to respond to each single-digit challenge $C$. To compute $f_{2,2}\paren{\sigma\paren{\{0,\ldots,13 \}}}$ the user would need to recall the values of $\sigma(10),\sigma(11),\sigma(12),\sigma(13)$ and $\sigma(j)$, where $j  = \sigma(10)+\sigma(11) \mod{10}$. If the user has $10$ digit passwords then he will naturally rehearse the value of $\sigma(i)$ for up to fifty different values of $i$ each time he computes one of his passwords. While the user needs to spend extra time computing his password each time he authenticates in our human computable password scheme, this extra computation time gives the user more opportunities to rehearse his secret mapping.

\section{Human Computable Passwords Challenge} \label{sec:challenge}

\begin{table}[h]
\centering
\caption{Human Computable Password Challenges \\ $n$ --- Secret Length \\ $m$---$\#$ Challenge-Response Pairs}
\label{tab:challenge}
\begin{center}

\begin{tabular}{ | l | l | l || l | l| }
\hline
&\multicolumn{2}{c||}{Scheme 1 $\paren{f_{2,2}}$} & \multicolumn{2}{|c|}{Scheme 2 $\paren{f_{1,3}}$} \\
\hline
\hline
  $n$ & $m$ & Winner & $m$ & Winner \\
\hline
 \multirow{3}{*}{$100$ digits} & $1000$ & N/A  & $500$ & N/A \\
\cline{2-5} 
  & $500$ & N/A & $300$ & N/A  \\
\cline{2-5} 
  & $300$ & N/A & $200$ & N/A \\ 
\hline
 \multirow{3}{*}{$50$ digits} & $500$ & N/A & $300$ & N/A \\
\cline{2-5} 
 & $300$ & N/A & $150$ & N/A \\
\cline{2-5} 
 & $150$ & N/A & $100$ &N/A \\ 
\hline
 \multirow{3}{*}{$30$ digits} & $300$ & CSP Solver & $150$ &N/A \\
\cline{2-5} 
 & $100$ & N/A  &100  &N/A\\
\cline{2-5}
 & $50$ & N/A & 50 & N/A \\  \hline
\end{tabular}
\centering

\end{center}
\end{table}

While we provided asymptotic security bounds for our human computable password schemes in our context it is particularly important to understand the constant factors. In our context, it may be reasonable to assume that $n \leq 100$ (e.g., the user may be unwilling to memorize longer mappings). In this case it would be feasible for the adversary to execute an attack that takes time proportional to $10^{\sqrt{n}} \leq 10^{10}$.  We conjecture that in practice scheme 2 $(f_{1,3})$ is slightly weaker than scheme 1 $(f_{2,2})$ when $n \leq 100$ despite the fact that $s\paren{f_{2,2}} < s\paren{f_{1,3}}$ because of the attack described in Remark \ref{remark:security}. This attack requires $\tilde{O}\paren{n^{1+g\paren{f}/2}}$ examples, and the running time $O\paren{10^{\sqrt{n}}n^3}$ may be feasible for $n \leq 100$. To better understand the exact security bounds we created several public challenges for researchers to break our human computable password schemes under different parameters (see Table \ref{tab:challenge}). At this time these challenges remain unsolved even after they were presented during the rump sessions at a cryptography conference and a security conference\cite{HumanComputablePasswordChallenge}. The challenges can be found at \url{http://www.cs.cmu.edu/~jblocki/HumanComputablePasswordsChallenge/challenge.htm}. For each challenge we selected a random secret mapping $\sigma \in \mathbb{Z}_{10}^n$, and published (1) $m$ single digit challenge-response pairs $\paren{C_1,f\paren{\sigma\paren{C_1}}},\ldots, \paren{C_m,f\paren{\sigma\paren{C_m}}}$, where each clause $C_i$ is chosen uniformly at random from $X_k$, and (2) $20$ length---$10$ password challenges $\vec{C}_1,\ldots,\vec{C}_{20} \in \paren{X_k}^{10}$. The goal is to guess one of the secret passwords $p_i = f\paren{\sigma\paren{\vec{C}_i}}$ for some $i \in [20]$. \\

\section{CSP Solver Attacks} \label{apdx:CSPSolver}
Theorems \ref{thm:SecurityLowerBound} and \ref{thm:secFinal} provide asymptotic security bounds (e.g., an adversary needs to see $m = \tilde{\Omega}\paren{n^{s(f)}}$ challenge-response pairs to forge passwords). However, in our context $n$ is somewhat small (e.g., $n\leq 100$). Thus, it is also important to address the following question: how many challenge-response pairs does the adversary need to see before it becomes feasible for the adversary to recover the secret on a modern computer? To better understand the exact security bounds of our human computable password schemes we used a Constraint Satisfaction Problem (CSP) solver to attack our scheme. We also created several public challenges to break our candidate human computable password schemes (see Table \ref{tab:challenge}).

\paragraph{CSP Solver} Our computations were performed on a computer with a 2.83 GHz Intel Core2 Quad CPU and 4 GB of RAM. In each instance we generated a random mapping $\sigma:[n]\rightarrow \mathbb{Z}_{10}$ and $m$ random challenge response pairs $\paren{C,f\paren{\sigma\paren{C}}}$ using the functions $f_{2,2}$ and $f_{1,3}$. We used the Constraint Satisfaction Problem solver from the Microsoft Solver Foundations library to try to solve for $\sigma$\footnote{http://blogs.msdn.com/b/solverfoundation/ (Retrieved 9/15/2014). }. The results of this attack are shown in Tables \ref{tab:SATAttack} and \ref{tab:SATAttack2}. Due to limited computational resources we terminated each instance if the solver failed to find the secret mapping within $2.5$ days, and if our solver failed to find $\sigma$ in $2.5$ days on and instance $(n,m)$ we did not run the solver on strictly harder instances (e.g., $(n',m')$ with $n' \geq n$ and $m' \leq m$). We remark that our empirical results are consistent with the hypothesis that the time/space resources consumed by the CSP solver increase exponentially in $n$ (e.g., when we decrease $n$ from $30$ to $26$ with $m = 100$ examples the CSP solver can solve for $\sigma \in \mathbb{Z}_{10}^{26}$ in $40$ minutes, while the solver failed to find $\sigma \in \mathbb{Z}_{10}^{30}$ in $2.5$ days and we observe similar threshold behavior in other columns of the table. ).

\begin{table}[t]
\centering
\begin{tabular}{|c|c|c|c|c|c|c|}
\hline
& $m = 50$ & $m = 100$ & $m = 300$ & $m = 500$ & $m = 1000$ & $m = 10000$ \\
\hline
$n= 26$ & 23.5 hr & 40 min & 4.5 hr & 29 min & 10 min & 2 min \\
\hline
$n= 30$ & HARD & UNSOLVED & 2.33 hr & 35.5 min & 10 min & 20 s \\
\hline
$n= 50$ & HARD & HARD & HARD & HARD & UNSOLVED & 7 hr \\
\hline
$n= 100$ & HARD & HARD & HARD & HARD & HARD & UNSOLVED \\
\hline
\end{tabular}
\caption{CSP Solver Attack on $f_{2,2}$ \\ Key: UNSOLVED -- Solver failed to find solution in 2.5 days; HARD -- Instance is harder than an unsolved instance;  }
\label{tab:SATAttack}
\end{table}

\begin{table}[t]
\centering
\begin{tabular}{|c|c|c|c|c|c|c|}
\hline
        & $m = 50$ & $m = 100$ & $m = 300$ & $m = 500$ & $m = 1000$ & $m = 10000$ \\
\hline
$n= 26$ &  8.7 hr       &    53 min &   1.33 hr &   13.5 min&     6.3min &  2 min            \\
\hline
$n= 30$ &   HARD       &  UNSOLVED &    1 hr   &    41 min &    2 min   &   15 s           \\
\hline
$n= 50$ &    HARD      &    HARD       &   HARD        &   HARD        &   UNSOLVED         &     6.5 hr         \\
\hline
$n= 100$&    HARD      &    HARD       &       HARD    &    HARD       &     HARD       &         UNSOLVED     \\
\hline
\end{tabular}
\caption{CSP Solver Attack on $f_{1,3}$ \\ Key: UNSOLVED -- Solver failed to find solution in 2.5 days; HARD -- Instance is harder than an unsolved instance;  }
\label{tab:SATAttack2}
\end{table}


\section{Statistical Dimension} \label{apdx:StatisticalDimension}
At a high level our statistical dimension lower bounds closely mirror the lower bounds from \cite{feldman2013complexity} for binary predicates. For example, Lemmas \ref{lemma:degree}, \ref{lemma:concentrationHypercontractivity}, \ref{lemma:concentrationRestricted}, \ref{lemma:boundQueryh} and \ref{lemma:difference} are similar to Lemmas 2, 4, 5, 6 and 7 from \cite{feldman2013complexity} respectively. 

While the high level proof strategy is very similar, we stress that our lower bounds do requires new ideas because we are working with planted solutions $\sigma \in \mathbb{Z}_d^n$ instead of $\sigma \in \mathbb{Z}_d^n$. We use the basis functions $\chi_\alpha$ where for $\alpha \in \mathbb{Z}_d^n$ is \[\chi_\alpha\paren{x} = \exp\paren{\frac{-2\pi\sqrt{-1} \paren{x\cdot \alpha} }{d}} \ .\]
Note that if $d > 2$ then the Fourier coefficients $\hat{b}_\alpha$ of a function $b:\mathbb{Z}_d^k\rightarrow \mathbb{R}$ might include complex numbers. While we need to take care to deal with the possibility that a Fourier coefficients might be complex, we are still able to apply powerful tools from Fourier analysis. For example,  Parseval's identity \[\sum_{\alpha \in \mathbb{Z}_d^k} \left| \hat{b}_\alpha\right|^2 = \mathbb{E}_{x \sim \mathbb{Z}_d^k}\left[b\paren{x}^2\right]  \ ,\]
still applies and there are versions of the hypercontractivity theorem \cite[Chapter 10]{o2007analysis} that still apply even when Fourier coefficients are complex.

Another difference is that the reference distribution is defined over clauses and outputs $X_k \times \mathbb{Z}_d$ (instead of just  clauses $X_k$) because we are working with a function $f:\mathbb{Z}_d^n \rightarrow \mathbb{Z}_d$ with non-binary outputs. Some care is needed in finding the right reference distribution. Unlike~\cite{feldman2013complexity} we cannot just use the uniform distribution over $X_k \times \mathbb{Z}_d$ as a reference distribution --- instead the reference distribution inherently depends on the function $f$ (Of course it is must still be independent of $\sigma$). 


The following definition will be useful in our proofs. 

\begin{definition} 
Given a clause $C \in X_k$ and $S \subseteq [k]$ of size $\ell$, we let $C_{|S} \in X_\ell$ denote the clause of variables of $C$ at the positions with indices in $S$ (e.g., if $C = \paren{1,\ldots,k}$ and $S = \{1,5,k-2\}$ then $C_{|S} = \paren{1,5,k-2} \in X_3$). 
Given a function $h:X_k \times \mathbb{Z}_d \rightarrow \mathbb{R}$, a clause $C_\ell \in X_\ell$  and $i \in \mathbb{Z}_d$ and a set $S \subseteq [k]$ of size $\left|S\right|=\ell$ we define 
\[h_S^i\paren{C_\ell} = \frac{\left|X_\ell\right|}{\left|X_k \right|} \sum_{ C\in X_k, C_{|S} =C_\ell} h\paren{C,i}  \ .\]
\end{definition}

We first show that $\Delta\paren{\sigma,h}$  can be expressed in terms of the Fourier coefficients of $\hat{Q}$ as well as the functions $h_\ell$. In particular, given a function $h:X_k \times \mathbb{Z}_d \rightarrow \mathbb{R}$ we define the degree $\ell$ function $b_\ell:\mathbb{Z}_d^n\rightarrow \mathbb{C}$ as follows \[b_\ell\paren{\sigma} \doteq  \frac{1}{2} \sum_{\alpha \in \mathbb{Z}_d^\ell: H\paren{\alpha}=\ell} {k \choose \ell} \sum_{i=0}^{d-1}  \hat{Q}_{\alpha}^{f,i}   \sum_{C\in X_\ell}  \chi_\alpha\paren{\sigma\paren{C}}   h^i_\ell(C)    \ . \] 


Lemma \ref{lemma:degree} states that \[\Delta\paren{\sigma,h} = \sum_{\ell=1}^k  \frac{1}{\left|X_\ell \right|} b_\ell\paren{\sigma} \ .\] This observation will be important later because we can use hypercontractivity to bound the expected value of $\left|b_\ell\paren{\sigma}\right|$. Notice that if $Q$ has distributional complexity $r$ and $\ell \leq r$ then $b_\ell\paren{\sigma} = 0$ because $\hat{Q}^{f,i}_\alpha = 0$ for all $i \in \mathbb{Z}_d$ and $\alpha \in \mathbb{Z}_d^k$ s.t. $1 \leq H\paren{\alpha} \leq r$. This means that first $r$ terms of the sum in Lemma \ref{lemma:degree} will be zero. 

\begin{lemma} \label{lemma:degree}
For every $\sigma \in \mathbb{Z}_d^k$, $j \in \mathbb{Z}_d$ and $h : X_k \rightarrow \mathbb{R}$ we have
$\Delta\paren{\sigma,h} = \sum_{\ell=1}^k  \frac{1}{\left|X_\ell \right|} b_\ell\paren{\sigma} $ .
\end{lemma}

\begin{proofof}{Lemma \ref{lemma:degree}}
Before calculating we first observe that for any $j \in \mathbb{Z}_d$ we have  \[\hat{Q}_{\vec{0}}^{f,j} = \mathbb{E}_{x\sim \mathbb{Z}_d^k}\left[Q^{f,j}\paren{x}\chi_\alpha\paren{x} \right]= \mathbb{E}_{x\sim \mathbb{Z}_d^k}\left[Q^{f,j}\paren{x} \right]=\sum_{x \in \mathbb{Z}_d^k} \frac{Q^{f,j}\paren{x}}{d^k} = \frac{2-d}{d} \ .\]
Given $\alpha \in \mathbb{Z}_d^k$ we also define $S_\alpha \subset [k]$ to be the set of indices $i$ s.t. $\alpha_i \ne 0$ --- $\left| S_\alpha \right| = H(\alpha)$.  Now we note that 
\begin{eqnarray*}
\Delta\paren{\sigma,h} &=& \mathbb{E}_{C\sim X_k}\left[h\paren{C,f\paren{\sigma\paren{C}}}  \right]-\mathbb{E}_{(C,i)\sim T}\left[ h\paren{C,i}\right] \\
&=& \sum_{C \in X_k} \left( \frac{h\big(C,f(\sigma(C))\big)}{\left|X_k\right|} - \sum_{i=0}^{d-1} \Pr_{T}[(C,i)] h(C,i) \right) \\ 
&=& \sum_{C \in X_k} \left( \frac{h\big(C,f(\sigma(C))\big)}{\left|X_k\right|} - \sum_{i=0}^{d-1} \frac{\Pr_{x\sim \mathbb{Z}_d^k}[f(x)=i]}{\left|X_k\right|} h(C,i) \right) \\
&=& \sum_{C \in X_k} \sum_{i=0}^{d-1}  \left(  \frac{h(C,i)\left(\frac{Q^{f,i}_\sigma(C)+1}{2}-\Pr_{x\sim \mathbb{Z}_d^k}[f(x)=i]\right)}{\left|X_k\right|}  \right) \\ 
&=& \sum_{C \in X_k} \sum_{i=0}^{d-1}  \left( \frac{h(C,i)}{\left|X_k\right|}\left( \frac{1}{2}-\Pr_{x\sim \mathbb{Z}_d^k}[f(x)=i]+ \sum_{\alpha \in \mathbb{Z}_d^k} \frac{\hat{Q}_{\alpha}^{f,i}\chi_\alpha\paren{\sigma\paren{C}}}{2} \right)  \right) \\ 
&=& \sum_{C \in X_k} \sum_{i=0}^{d-1}  \left( \frac{h(C,i)}{\left|X_k\right|}\left( \frac{1}{d}-\Pr_{x\sim \mathbb{Z}_d^k}[f(x)=i]+ \sum_{\alpha \in \mathbb{Z}_d^k: \alpha \ne \vec{0}} \frac{\hat{Q}_{\alpha}^{f,i}\chi_\alpha\paren{\sigma\paren{C}}}{2} \right)  \right) \\ 
&=& \sum_{C \in X_k} \sum_{i=0}^{d-1}  \left( \frac{h(C,i)}{\left|X_k\right|}\left(  \sum_{\alpha \in \mathbb{Z}_d^k: \alpha \ne \vec{0}} \frac{\hat{Q}_{\alpha}^{f,i}\chi_\alpha\paren{\sigma\paren{C}}}{2} \right)  \right) \\ 
&=& \sum_{C \in X_k} \sum_{i=0}^{d-1}  \left( \frac{h(C,i)}{\left|X_k\right|}\left(  \sum_{\ell=1}^k \sum_{\alpha \in \mathbb{Z}_d^k: H\paren{\alpha}=\ell}  \frac{\hat{Q}_{\alpha}^{f,i}\chi_\alpha\paren{\sigma\paren{C}}}{2} \right)  \right) \\ 
&=& \frac{1}{2\left|X_k\right|} \sum_{\ell=1}^k  \sum_{C \in X_k}   \sum_{i=0}^{d-1}  \left( h(C,i)\left( \sum_{\alpha \in \mathbb{Z}_d^k: H\paren{\alpha}=\ell}  \hat{Q}_{\alpha}^{f,i}\chi_\alpha\paren{\sigma\paren{C}} \right)  \right) \\
&=& \frac{1}{2\left|X_k\right|} \sum_{\ell=1}^k \sum_{\alpha \in \mathbb{Z}_d^k: H\paren{\alpha}=\ell}   \sum_{C \in X_k}   \sum_{i=0}^{d-1}  \left( h(C,i)\left(  \hat{Q}_{\alpha}^{f,i}\chi_\alpha\paren{\sigma\paren{C}} \right)  \right) \\
&=& \frac{1}{2\left|X_k\right|} \sum_{\ell=1}^k \sum_{\alpha \in \mathbb{Z}_d^k: H\paren{\alpha}=\ell} \sum_{i=0}^{d-1}  \hat{Q}_{\alpha}^{f,i}   \sum_{C_\ell \in X_\ell}   \sum_{C \in X_k, C_{|S_\alpha} = C_\ell} \chi_\alpha\paren{\sigma\paren{C}}   h^i(C)   \\
&=& \frac{1}{2\left|X_k\right|} \sum_{\ell=1}^k \sum_{\alpha \in \mathbb{Z}_d^\ell: H\paren{\alpha}=\ell} {k \choose \ell} \sum_{i=0}^{d-1}  \hat{Q}_{\alpha}^{f,i}   \sum_{C_\ell \in X_\ell}   \sum_{C \in X_k, C_{|S_\alpha} = C_\ell} \chi_\alpha\paren{\sigma\paren{C_\ell}}   h(C,i)   \\
&=&  \sum_{\ell=1}^k \frac{1}{2\left|X_\ell\right|} \sum_{\alpha \in \mathbb{Z}_d^\ell: H\paren{\alpha}=\ell} {k \choose \ell} \sum_{i=0}^{d-1}  \hat{Q}_{\alpha}^{f,i}   \sum_{C_\ell \in X_\ell}  \chi_\alpha\paren{\sigma\paren{C_\ell}}   h^i_{S_\alpha}(C)   \\
&=&  \sum_{\ell=1}^k \frac{1}{\left|X_\ell\right|} b_\ell(\sigma) \ .
\end{eqnarray*}
 \end{proofof}

The following lemma is similar to Lemma 4 from \cite{feldman2013complexity}. Lemma \ref{lemma:concentrationHypercontractivity} is based on the general hypercontractivity theorem \cite[Chapter 10]{o2007analysis} and applies to more general (non-boolean) functions.
\begin{lemma} \cite[Theorem 10.23]{o2007analysis} \label{lemma:concentrationHypercontractivity}
If $b:\mathbb{Z}_d^n\rightarrow \mathbb{R}$ has degree at most $\ell$ then for any $t \geq \paren{\sqrt{2ed}}^\ell$,
\[ \Pr_{x\sim \mathbb{Z}_d^n } \left[ \left|b(x) \right| \geq t \|b \|_2 \right] \leq \frac{1}{d^\ell} \exp\paren{-\frac{\ell}{2ed} t^{2/\ell} }\ ,\]
where $\|b \|_2 = \sqrt{\mathbb{E}_{x\sim \mathbb{Z}_d^n }\left[b\paren{x}^2 \right]}$
\end{lemma}

Lemma \ref{lemma:concentrationRestricted} and its proof are almost identical to Lemma 5 in \cite{feldman2013complexity}. We simply replace their concentration bounds with the concentration bounds in Lemma \ref{lemma:concentrationHypercontractivity}. We include the proof for completeness.

\begin{reminderlemma}{\ref{lemma:concentrationRestricted}}
\lemmaConcentrationBoundRestricted
\end{reminderlemma}

\begin{proofof}{Lemma \ref{lemma:concentrationRestricted}}
The set $\D'$ contains $1/d'$ fraction of points in $\mathbb{Z}_d^n$. Therefore,
\[ \Pr_{x\sim \D'} \left[ \left|b(x) \right| \geq t \|b \|_2 \right] \leq \frac{d'}{d^\ell} \exp\paren{-\frac{\ell}{2ed} t^{2/\ell} }\ ,\]
for any $t \geq \paren{\sqrt{2ed}}^\ell$. For any random variable $Y$ and value $a \in \mathbb{R}$,
\[ \mathbb{E}[Y] \leq a   +  \int_a^\infty \Pr[Y \geq t]\,dt \ .\]
We set $Y = \left|b\paren{\sigma} \right|/\|b \|_2$ and $a = \paren{ \frac{\ln d'}{c_0}}^{\ell/2}$.  Assuming that $a > \paren{\sqrt{2ed}}^\ell$ we get
\begin{align*}\frac{\mathbb{E}_{\sigma \sim \D'}[|b(\sigma)|]}{\|b\|_2}  &\leq \left(\frac{\ln d'}{c_0}\right)^{\ell/2} +  \int_{(\ln d'/c_0)^{\ell/2}}^\infty \frac{d'}{d^\ell} \cdot e^{-c_0 t^{2/\ell}} dt  \\ &=
(\ln d'/c_0)^{\ell/2} +  \frac{\ell \cdot d'}{2d^\ell \cdot c_0^{\ell/2}} \cdot \int_{\ln d'}^\infty e^{-z} z^{\ell/2-1} dz \ . 
\end{align*}
Let $u(i) \doteq \int_{\ln d'}^\infty e^{-z} z^{\ell/2 - i} dz$. Applying integration by parts we have
\begin{eqnarray*}
u(i) \doteq \int_{\ln d'}^\infty e^{-z} z^{\ell/2 - i} dz &=& \left.\left( e^{-z} z^{\ell/2-i+1} \right)\right|_{\ln d'}^\infty + \int_{\ln d'}^\infty e^{-z} z^{\ell/2 - i+1} dz - \left(\frac{\ell}{2}-i \right)\int_{\ln d'}^\infty e^{-z} z^{\ell/2 - i} dz \\ &=&
\left.\left( e^{-z} z^{\ell/2-i+1} \right)\right|_{\ln d'}^\infty + u(i-1) - \left(\frac{\ell}{2}-i \right)u(i)
  \ . 
\end{eqnarray*}

Thus, \[u(i-1) = \left(\frac{\ell}{2}-i+1\right) u(i) -\left.\left( e^{-z} z^{\ell/2-i+1} \right)\right|_{\ln d'}^\infty = \left(\frac{\ell}{2}-i+1\right) u(i) + \frac{\left(\ln d' \right)^{\frac{\ell}{2}-i+1 }}{d'} \ . \] 
Unrolling the recurrence for $T \geq 1$ we get \[u(1) = \frac{\left(\ln d' \right)^{\frac{\ell}{2} }}{d'} +  \sum_{i=1}^{T-1} \frac{\left(\ln d' \right)^{\frac{\ell}{2}-i }}{d'} \prod_{j=1}^{i}  \left(\frac{\ell}{2}-j +1 \right) + u\left(T+1\right)\prod_{j=1}^{T}  \left(\frac{\ell}{2}-j +1 \right)  \ .  \]
We also note that for $T = \lceil \frac{\ell}{2}+1 \rceil$ we have $u\left(T+1\right)\prod_{j=1}^{T}  \left(\frac{\ell}{2}-j +1 \right) \leq 0$. This follows because $u(T+1) \geq 0$ for any integer $T \geq 0$ and $\prod_{j=1}^{T} \left(\frac{\ell}{2}-j +1 \right)  \leq 0$. It follows that 
 \[u(1) \leq \frac{\left(\ln d' \right)^{\frac{\ell}{2} }}{d'} +  \sum_{i=1}^{\lceil \frac{\ell}{2}+1 \rceil} \frac{\left(\ln d' \right)^{\frac{\ell}{2}-i }}{d'} \prod_{j=1}^{i}  \left(\frac{\ell}{2}-j +1 \right) \leq  \frac{\left(\ln d' \right)^{\frac{\ell}{2} }}{d'}\left( 1+ \sum_{i=1}^{\lceil \frac{\ell}{2}+1 \rceil} \ell^{-i} \left( \frac{\ell}{2}\right)^i \right) \leq  \frac{2\left(\ln d' \right)^{\frac{\ell}{2} }}{d'}\   \]
 where we used the condition $d' \geq e^{\ell}$ to obtain the second to last inequality. Now we have
 
 \begin{align*}\frac{\mathbb{E}_{\sigma \sim \D'}[|b(\sigma)|]}{\|b\|_2}  &\leq 
(\ln d'/c_0)^{\ell/2} +  \frac{\ell \cdot d'}{2d^\ell \cdot c_0^{\ell/2}} \cdot u(1) \\
&\leq 
(\ln d'/c_0)^{\ell/2} +  \frac{\ell \cdot d'}{2d^\ell \cdot c_0^{\ell/2}} \cdot \left(\frac{2\left(\ln d' \right)^{\frac{\ell}{2} }}{d'} \right)  \\
&\leq 
2(\ln d'/c_0)^{\ell/2} \  . 
\end{align*}
\end{proofof}

\begin{lemma} \label{lemma:boundQueryh}
Let $\D' \subseteq \{0,\ldots,d-1\}^n$ be a set of assignments for which $d' = d^n/\left| \D'\right|$. Then
\[\mathbb{E}_{\sigma\sim\D'} \left[ \left|\frac{1}{\left|X_\ell\right|} b_\ell\paren{\sigma}\right|\right]   \leq  2\paren{{k \choose \ell} \sqrt{ \ell!d}}  \paren{\ln d'/c_0}^{\ell/2} \max_{\alpha \in \mathbb{Z}_d^\ell} \sum_{i=0}^{d-1} \frac{\|h_{S_\alpha}\paren{\sigma} \|_2}{\sqrt{\left|X_\ell \right|}  } \ .  \]
\end{lemma}
\begin{proof}
For simplicity of notation we set $b = b_\ell$. Our first goal will be to find the Fourier coefficients of \[b(\sigma) =  \frac{1}{2} \sum_{\alpha \in \mathbb{Z}_d^\ell: H\paren{\alpha}=\ell} {k \choose \ell} \sum_{i=0}^{d-1}  \hat{Q}_{\alpha}^{f,i}   \sum_{C \in X_\ell}  \chi_\alpha\paren{\sigma\paren{C}}   h^i_{S_\alpha}(C)    \ . \] 
Given $\alpha = (\alpha_1,\ldots,\alpha_\ell) \in \mathbb{Z}_d^\ell$ with $H(\alpha) = \ell$ and a clause $C = (c_1,...,c_\ell) \in X_\ell$ we define the projection $\alpha^C \in \mathbb{Z}_d^\ell$ of $\alpha$ onto $C$ to be the unique vector s.t. $\alpha^C_{c_i} = \alpha_i$ for each $i \leq \ell$ and $\alpha_j = 0$ for each $j \notin \{c_1,\ldots,c_\ell\}$ --- note that $H(\alpha^C) = \ell$ . 

Given $\alpha,\alpha ' \in \mathbb{Z}_d^k$ with $H(\alpha)=H(\alpha')=\ell$ and $C, C \in X_\ell$ we say that the pairs $\paren{\alpha,C_\ell}$ and $\paren{\alpha ', C_\ell '}$ are equivalent if and only if their projections are equal  $\alpha^C = \alpha'^{C'}$\footnote{For example, if $C = (1,2,5)$, $C = (1,5,2)$ and $\alpha = (4,5,6)$ and $\alpha' = (4,6,5)$ then the pairs $\paren{\alpha,C_\ell}$ and $\paren{\alpha ', C_\ell '}$ are equivalent. }. We can partition the set $\{ \alpha \in \mathbb{Z}_{d}^k: H\paren{\alpha} = \ell\} \times X_\ell$ into equivalence classes $E_1,\ldots,E_t$. If the pairs $\paren{\alpha,C_\ell}$ and $\paren{\alpha ', C_\ell '}$ are equivalent then we observe that the clauses $C$ and $C '$ must contain the same variables though perhaps in a different order. Furthermore, given an equivalence class $E_j$ such that $(\alpha,C) \in E_j$ we have $(\alpha',C) \notin E_j$ for any $\alpha' \neq \alpha \in \mathbb{Z}_d^k$. Thus each equivalence class has size $\ell!$ because there are $\ell!$ ways to reorder the $\ell$ variables in a clause $C$. We can rewrite $b(\sigma)$ as 
\[b(\sigma) = \frac{{k \choose \ell}}{2} \sum_{j=1}^t\sum_{(\alpha,C) \in E_j} \sum_{i=0}^{d-1} \hat{Q}_{\alpha}^{f,i} h^i_{S_\alpha}(C)   \chi_{\alpha}(\sigma(C)) \]
  Let $(\alpha,C) \in E_j$ then the Fourier coefficient of $\alpha^C$ is
\[\hat{b}_j = \hat{b}_{\alpha^C} =  \frac{{k \choose \ell}}{2}\sum_{(\alpha',C') \in E_j} \sum_{i=0}^{d-1} \hat{Q}_{\alpha}^{f,i} h^i_\ell(C') \  . \] 
We also note that $t = \frac{ \left| X_\ell \right|(d-1)^k}{\ell!}$.

Now we can apply Parseval's identity along with the Cauchy-Schwarz inequality to obtain
\begin{eqnarray*}
\mathbb{E}_{\sigma\sim\mathbb{Z}_d^n} \left[b\paren{\sigma} \overline{b\paren{\sigma}}\right] &=& \mathbb{E}_{\sigma\sim\mathbb{Z}_d^n} \left[\left|b\paren{\sigma}\right|^2\right]  \\ 
&=& \mathbb{E}_{\sigma\sim\mathbb{Z}_d^n}\left[ \sum_{j=1}^t \left|\hat{b}_j  \right|^2 \right] \\
&=& \frac{{k \choose \ell}^2}{4} \mathbb{E}_{\sigma\sim\mathbb{Z}_d^n}\left[ \sum_{j=1}^t \left|\sum_{(C,\alpha) \in E_j}\sum_{i=0}^{d-1} \hat{Q}_{\alpha}^{f,i} h^i_{S_{\alpha}}(C)   \right|^2 \right] \\
&\leq& \frac{{k \choose \ell}^2}{4} \mathbb{E}_{\sigma\sim\mathbb{Z}_d^n}\left[ \sum_{j=1}^t \left(\sum_{(C,\alpha) \in E_j}\sum_{i=0}^{d-1} \left| \hat{Q}_{\alpha}^{f,i} \right|^2 \right) \left(\sum_{(C,\alpha) \in E_j}\sum_{i=0}^{d-1} \left|h^i_{S_{\alpha}} (C)  \right|^2 \right)   \right] \\
&\leq& \frac{{k \choose \ell}^2}{4} \mathbb{E}_{\sigma\sim\mathbb{Z}_d^n}\left[ \sum_{j=1}^t \left(\ell!\max_{j \leq t, (C,\alpha) \in E_j } \sum_{i=0}^{d-1} \left| \hat{Q}_{\alpha}^{f,i} \right|^2 \right) \left(\sum_{(C,\alpha) \in E_j}\sum_{i=0}^{d-1} \left|h^i_{S_{\alpha}}(C)  \right|^2 \right)   \right] \\
&\leq& \frac{{k \choose \ell}^2}{4} \mathbb{E}_{\sigma\sim\mathbb{Z}_d^n}\left[  \left(\ell!\max_{j \leq t, (C,\alpha) \in E_j } \sum_{i=0}^{d-1} \left| \hat{Q}_{\alpha}^{f,i} \right|^2 \right) \left(\sum_{j=1}^t \sum_{(C,\alpha) \in E_j}\sum_{i=0}^{d-1} \left|h^i_{S_{\alpha}}(C)  \right|^2 \right)   \right] \\
&\leq& \frac{{k \choose \ell}^2}{4} \mathbb{E}_{\sigma\sim\mathbb{Z}_d^n}\left[  \left(\ell!\max_{j \leq t, (C,\alpha) \in E_j } \sum_{i=0}^{d-1} \left| \hat{Q}_{\alpha}^{f,i} \right|^2 \right) \left(\max_{\alpha \in \mathbb{Z}_d^\ell}\sum_{C \in X_\ell}\sum_{i=0}^{d-1} \left|h^i_{S_{\alpha}}(C)   \right|^2 \right)   \right] \\
&\leq& \frac{{k \choose \ell}^2}{4} \mathbb{E}_{\sigma\sim\mathbb{Z}_d^n}\left[  \left(\ell!\max_{j \leq t, (C,\alpha) \in E_j } \sum_{i=0}^{d-1} \left| \hat{Q}_{\alpha}^{f,i} \right|^2 \right) \left(\left|X_\ell\right|\max_{\alpha \in \mathbb{Z}_d^\ell} \sum_{i=0}^{d-1} \mathbb{E}_{C \sim X_\ell} \left[h^i_{S_\alpha}(C)^2   \right] \right)   \right] \\
&\leq& \frac{{k \choose \ell}^2 d \ell! \left|X_\ell\right|}{4} \mathbb{E}_{\sigma\sim\mathbb{Z}_d^n}\left[  \left(\max_{j \leq t, (C,\alpha) \in E_j } \sum_{i=0}^{d-1} \left| \hat{Q}_{\alpha}^{f,i} \right|^2 \right)   \right] \max_{\alpha \in \mathbb{Z}_d^\ell} \sum_{i=0}^{d-1} \| h_{S_\alpha}^i\|_2^2 \\
&\leq& \frac{{k \choose \ell}^2 d \ell! \left|X_\ell\right|}{4}   \max_{\alpha \in \mathbb{Z}_d^\ell} \sum_{i=0}^{d-1} \|h_{S_\alpha}^i\|_2^2 \ .
\end{eqnarray*} 
Before we can apply Lemma \ref{lemma:concentrationRestricted} we must address a technicality. The range of $b = b_\ell$ might include complex numbers, but Lemma \ref{lemma:concentrationRestricted} only applies to functions $b$ with range $\mathbb{R}$. For $c,d \in \mathbb{R}$ we adopt the notation  $Im\paren{c+d\sqrt{-1}} = d$ and $Re\paren{c+d\sqrt{-1}} = c$. We observe that 
\begin{eqnarray*}
\mathbb{E}_{\sigma\sim\mathbb{Z}_d^n} \left[b\paren{\sigma} \overline{b\paren{\sigma}}\right] &=& \mathbb{E}_{\sigma\sim\mathbb{Z}_d^n} \left[Re\paren{b\paren{\sigma}}^2+Im\paren{b\paren{\sigma}}^2 \right] \\
&=& \|Re\paren{b} \|^2_2 + \|Im\paren{b}\|^2_2 \ .
\end{eqnarray*}

We first observe that $Re\paren{b}$ and $Im\paren{b}$ are both degree $\ell$ functions because $b$ is a degree $\ell$ function.


Now we can apply Lemma \ref{lemma:concentrationRestricted} to get
\begin{eqnarray*}
\mathbf{E}_{\sigma\sim\D'}\left[\left|Re\paren{b\paren{\sigma}} \right|\right] &\leq& \frac{2 \paren{\ln d'/c_0}^{\ell/2}}{d^\ell} \|Re(b) \|_2 \\ 
&\leq&  \frac{2 \paren{\ln d'/c_0}^{\ell/2}}{d^\ell} \sqrt{\mathbb{E}_{\sigma\sim\mathbb{Z}_d^n} \left[b\paren{\sigma} \overline{b\paren{\sigma}}\right] }  \\
&\leq&  \paren{{k \choose \ell}  \sqrt{\ell!d}}  \paren{\ln d'/c_0}^{\ell/2} \sqrt{\left|X_\ell \right|} \max_{\alpha \in \mathbb{Z}_d^\ell} \sum_{i=0}^d   \|h_{S_\alpha}\paren{\sigma} \|_2  \ . 
\end{eqnarray*}
A symmetric argument can be used to bound $\mathbf{E}_{\sigma\sim\D'}\left[Im\paren{b\paren{\sigma}} \right]$. 
Now because
\[\left| b\paren{\sigma} \right| \leq \left| Re\paren{b\paren{\sigma}} \right| + \left|Im\paren{b\paren{\sigma}} \right|  \  , \]
it follows that
\begin{eqnarray*}
\mathbf{E}_{\sigma\sim\D'}\left[\left|\frac{1}{\left|X_\ell\right|}b\paren{\sigma}\right| \right] &\leq&  2 \paren{{k \choose \ell} \sqrt{\ell!d}} \paren{\frac{1}{\left|X_\ell\right|}} \paren{\paren{\ln d'/c_0}^{\ell/2}   } \max_{\alpha \in \mathbb{Z}_d^\ell} \sum_{i=0}^{d-1} \|h_{S_\alpha}\paren{\sigma} \|_2  \sqrt{\left|X_\ell \right|} \\
 &\leq& 2\paren{{k \choose \ell} \sqrt{ \ell!d}}  \paren{\ln d'/c_0}^{\ell/2} \max_{\alpha \in \mathbb{Z}_d^\ell} \sum_{i=0}^{d-1} \frac{\|h_{S_\alpha}\paren{\sigma} \|_2}{\sqrt{\left|X_\ell \right|}  } \ .
\end{eqnarray*}

\end{proof}

We will use Fact \ref{fact:hBound} to prove Lemma \ref{lemma:difference}. The proof of Fact \ref{fact:hBound} is found in \cite[Lemma 7]{feldman2013complexity}. We include it here for completeness.
\begin{fact} \cite{feldman2013complexity} \label{fact:hBound}
If $h:X_k\times \mathbb{Z}_d \rightarrow \mathbb{R}$ satisfies $\|h\|_2^2 = 1$ then for any $i \in \mathbb{Z}_d$, $0 \leq \ell \leq k$ and $S\subseteq [k]$ of size $|S|=\ell$ we have $\sum_{i=0}^{d-1} \|h_S^i\|_2^2 \leq d$.
\end{fact}
\begin{proof}
First notice that for any $C_\ell$, $S\subseteq [k]$ s.t $\left|S\right| = \ell$
 \[ \left|\{C \in X_k ~\vline~ C_{|S} = C_\ell    \} \right| = \frac{\left|X_k\right|}{\left|X_\ell \right|} \ . \]
By applying the definition of $h_\ell$ along with the Cauchy-Schwartz inequality
\begin{eqnarray*}
\sum_{i=0}^{d-1} \|h_S^i\|_2^2 &=& \sum_{i=0}^{d-1} \mathbb{E}_{C_\ell\sim X_\ell}\left[h_S^i\paren{C_\ell}^2 \right] \\
&=& \paren{\frac{\left|X_\ell \right|}{\left|X_k \right|}}^2 \sum_{i=0}^{d-1} \mathbb{E}_{C_\ell\sim X_\ell}\left[ \paren{ \sum_{C\in X_k, C_{|S} =C_\ell} h\paren{C,i}       }^2 \right] \\
&\leq& \paren{\frac{\left|X_\ell \right|}{\left|X_k \right|}}^2 \sum_{i=0}^{d-1} \mathbb{E}_{C_\ell\sim X_\ell}\left[ \frac{\left|X_k\right|}{\left|X_\ell\right|} \paren{ \sum_{ C\in X_k, C_{|S} =C_\ell} h\paren{C,i}^2       } \right] \\
&\leq& \paren{\frac{\left|X_\ell \right|}{\left|X_k \right|}}  d \mathbb{E}_{C_\ell\sim X_\ell, i \sim \mathbb{Z}_d}\left[  \paren{ \sum_{C\in X_k, C_{|S} =C_\ell} h\paren{C,i}^2       } \right] \\
&=& d \mathbb{E}_{C\sim U_k}\left[h\paren{C}^2 \right] = d\|h\|_2^2 = d \ . 
\end{eqnarray*}
\end{proof}

\begin{lemma} \label{lemma:difference}
Let $r= r(f)$, let $\D' \subseteq 
\left\{0,\ldots,d-1\right\}^n$ be a set of secret mappings and let $d' = d^n/|\D'|$. Then
$\dc(\D') = O_k\left((\ln d'/n)^{r/2}\right)$
\end{lemma}
\begin{proof}
Let $h:X_k \rightarrow \mathbb{R}$ be any function such that $\mathbb{E}_{U_k}\left[h^2 \right] = 1$. Using Lemma \ref{lemma:degree} and the definition of $r$,
\begin{eqnarray*}
\left|\Delta\paren{\sigma,h} \right| &=& \left| \sum_{\ell=r}^k  \frac{1}{\left|X_\ell \right|}b_\ell\paren{\sigma} \right| \\
&\leq&\sum_{\ell=r}^k  \left| \frac{1}{\left|X_\ell \right|}  b_\ell\paren{\sigma}  \right|  \ .
\end{eqnarray*} 

We apply Lemma \ref{lemma:boundQueryh} and Fact \ref{fact:hBound} to get 
\begin{eqnarray*}
\mathbb{E}_{\sigma\sim\D'} \left[\left|\Delta\paren{\sigma,h} \right|\right] &\leq&2 \sum_{\ell=r}^k  {k \choose \ell} \sqrt{\ell!d} \paren{\ln d'/c_0}^{\ell/2} \max_{\alpha \in \mathbb{Z}_d} \sum_{i=0}^{d-1} \frac{\|h_{S_\alpha}^i\paren{\sigma} \|_2}{\sqrt{\left|X_\ell \right|}  }  \\
&\leq& \sum_{\ell=r}^k \paren{{k \choose \ell} d\sqrt{\ell!}}  \paren{\frac{2 \paren{\ln d'/c_0}^{\ell/2}}{\sqrt{\left|X_\ell \right|}} } \\
&\leq& O_{k}\paren{\frac{ \paren{\ln d'}^{\ell/2}}{ n^{r/2} }} \ .
\end{eqnarray*}
\end{proof}

\begin{remindertheorem}{\ref{thm:StatisticalDimension}}
\thmStatisticalDimension
\end{remindertheorem}

\begin{proofof}{Theorem \ref{thm:StatisticalDimension}}
 First note that, by Chernoff bounds, for any solution $\tau \in \mathbb{Z}_d^n$ the fraction of assignments $\sigma \in \mathbb{Z}_d^n$ such that $\tau$ and $\sigma$ are $\close$-correlated (e.g., $H\paren{\sigma,\tau} \leq \frac{n(d-1)}{d}-\close\cdot n $) is at most $e^{-2n \cdot \close^2}$. In other words $|\D_\sigma| \geq \paren{1- e^{-2n \cdot \close^2}} |\mathbb{Z}_d^n|$, where $\D_\sigma = \mathbb{Z}_d^n \setminus \left\{\sigma' ~\vline~H\paren{\sigma,\sigma'} \leq \frac{n(d-1)}{d}-\close\cdot n  \right\}$ Let $\D' \subseteq \D_\sigma$ be a set of distributions of size $|\D_\sigma|/q$. Then for $d' = d^n/|\D '| = q \cdot d^n/|\D_\sigma|$, by Lemma \ref{lemma:difference}  we get \begin{eqnarray}
\dc(\D') &=& O_k\left(\frac{(\ln d')^{r/2}}{n^{r/2}}\right)  \\ 
&=& O_k\left(\frac{(\ln q)^{r/2}}{n^{r/2}}\right) \label{eqn:SizeOfD'} \ ,
\end{eqnarray}
where the last line follows by Sterling's Approximation
\[q = d'|\D_\sigma|/d^n = d'|\D_\sigma|/d^n \approx d'c' \sqrt{\frac{d}{n}} \, \]
for a constant $c'$. The claim now follows from the definition of $\SDN$.
\end{proofof}

The proof of Theorem \ref{thm:SecurityLowerBound} follows from Theorem \ref{thm:StatisticalDimension} and the following result of Feldman et al. \cite{feldman2013complexity}.

\begin{remindertheorem}{\ref{thm:avgvstat-random} \cite[Theorems 10 and 12]{feldman2013complexity}}
\thmAverageStatRandom
\end{remindertheorem}

\begin{remindertheorem}{\ref{thm:SecurityLowerBound}}
\thmSecurityLowerBound
\end{remindertheorem}


\section{Security Proofs} \label{apdx:securityProofs}

\begin{remindertheorem}{\ref{thm:EvenDistributedPredict}}
\ThmEvenDistributedPredict
\end{remindertheorem}

\begin{proofof}{Theorem \ref{thm:EvenDistributedPredict}}
Let $f:\mathbb{Z}_d^k\rightarrow \mathbb{Z}_d$ be a function with evenly distributed output. We select fix $C^{-1} \sim X_{k-1}$ and $i \sim [n]\setminus C^{-1}$. Given $j \in [n]\setminus C^{-1}$ we let  $C_j = (C^{-1},j)\in X_k$ denote the corresponding clause. 
Now we generate the mapping $\sigma'$ by selecting $\sigma'(i)$ at random, and setting $\sigma'(j) = \sigma'(i) + \ell_{C_j}-\ell_{C_i} \mod{d}$ for $j \in [n]\setminus C_i$. For $j \in C^{-1}$ we select $\sigma'(j)$ at random.  We let $\mathbf{GOOD}\paren{C^{-1},i,\sigma'}$ denote the event that \[\Pr_{j \sim [n]\setminus C^{-1}}\left[\ell_{C_j} = f\paren{\sigma\paren{C_j}} \right] \geq \frac{1}{d}+\epsilon/2 \ , \] $\sigma'(i) = \sigma(i)$ and $\ell_{C_i} = f\paren{\sigma\paren{C_i}}$. Assume that the event $\mathbf{GOOD}\paren{C^{-1},i}$ occurs, in this case for each $j$ s.t. $\ell_{C_j} = f\paren{\sigma\paren{C_j}}$ we have
\begin{eqnarray*}
\sigma'(j) - \sigma(j) &\equiv& \paren{\sigma'(i) + \ell_{C_j}-\ell_{C_i}}-\sigma(j) \mod{d} \\
&\equiv& \paren{\ell_{C_j}-\ell_{C_i}}+\sigma(i)-\sigma(j) \mod{d} \\
 &\equiv& g\paren{\sigma\paren{C^{-1}}} + \sigma\paren{j} - \paren{g\paren{\sigma\paren{C^{-1}}} + \sigma\paren{i}} +\sigma(i)-\sigma(j)  \mod{d}\\
& \equiv & 0 \mod{d} \ . 
\end{eqnarray*}  
Therefore, we have 
\[\frac{n-H\paren{\sigma,\sigma'}}{n}\geq \frac{1}{d}+\epsilon/2-\frac{k-1}{n} \ .\]
We note that by Markov's inequality the probability of success is at least 
\begin{eqnarray*}
\Pr_{C^{-1} \sim X_{k-1}}\left[\mathbf{GOOD}\paren{C^{-1}}\right] 
 &\geq&\frac{\epsilon}{2d^2} \ . 
\end{eqnarray*} 

\end{proofof}

Before proving Theorem \ref{thm:secFinal} we introduce some notation and prove an important claim. 
 We use $\Adversary_{C_1,\ldots,C_m}:\paren{X_k}^\challengelength\rightarrow \mathbb{Z}_d^\challengelength$ to denote an adversary who sees the challenges $C_1,\ldots, C_m \in X_k$ and the corresponding responses $f\paren{\sigma\paren{C_1}},\ldots,f\paren{\sigma\paren{C_m}}$. $\Adversary_{C_1,\ldots,C_m}\paren{C_1',\ldots,C_\challengelength'} \in \mathbb{Z}_d^\challengelength$ denotes the adversaries prediction of $f\paren{\sigma\paren{C_1'}},\ldots,f\paren{\sigma\paren{C_\challengelength'}}$. Given a function $b:\paren{X_k}^\challengelength\rightarrow \mathbb{Z}_d^\challengelength$, challenges $C_1',\ldots,C_\challengelength' \in X_k$ and responses $f\paren{\sigma\paren{C_1'}},\ldots,f\paren{\sigma\paren{C_\challengelength'}}$ we use $\mathcal{P}_{b,i,C_1',\ldots,C_m'}:X_k\times[t]\rightarrow \mathbb{Z}_d\cup \{\bot\}$ to predict the value of a clause $C \in X_k$
\[
    \mathcal{P}_{b,C_1',\ldots,C_\challengelength'}\paren{C,i} = 
\begin{cases}
   b\paren{\hat{C}_1,\ldots,\hat{C}_\challengelength}[i]  ,& \text{if } f\paren{\sigma\paren{\hat{C}_j}} = b\paren{\hat{C}_1,\ldots,\hat{C}_\challengelength}[j]~~ \forall j < i \\
    \bot,              & \text{otherwise}
\end{cases}
\]
where $\hat{C}_i = C$ and $\hat{C}_j = C_j'$ for $j \neq i$. We allow our predictor $\mathcal{P}_{b,C_1',\ldots,C_\challengelength'}\paren{C,i}$ to output $\bot$ when it is unsure. Informally, Claim \ref{claim:secFinal} says that for $b= \Adversary_{C_1,\ldots,C_m}$  our predictor $\mathcal{P}_{b,i,C_1',\ldots,C_m'}$ is reasonably accurate whenever it is not unsure. Briefly, Claim \ref{claim:secFinal} follows because for $b= \Adversary_{C_1,\ldots,C_m}$ we have \[\Pr\left[\mathbf{Wins}\paren{\Adversary,n,m,\challengelength}\right] = \prod_{i=1}^d \Pr_{\substack{C\sim X_k\\C_1,\ldots,C_m\sim X_k \\C_1',\ldots,C_{\challengelength}'\sim X_k}}\left[\mathcal{P}_{b,C_1',\ldots,C_\challengelength'}\paren{C,i} = f\paren{\sigma\paren{C}} ~\vline~\mathcal{P}_{b,C_1',\ldots,C_\challengelength'}\paren{C,i} \neq \bot \right]  \ .\] 

\newcommand{\claimSecFinal}{Let $\Adversary$ be an adversary s.t $\Pr\left[\mathbf{Wins}\paren{\Adversary,n,m,\challengelength}\right] >\paren{\frac{1}{d} +\epsilon}^\challengelength $ and let $b = \Adversary_{C_1,\ldots,C_m}$ then

\[\Pr_{\substack{i \sim[\challengelength],C\sim X_k\\C_1,\ldots,C_m\sim X_k \\C_1',\ldots,C_{\challengelength}'\sim X_k}}\left[\mathcal{P}_{b,C_1',\ldots,C_\challengelength'}\paren{C,i} = f\paren{\sigma\paren{C}}~\vline~\mathcal{P}_{b,C_1',\ldots,C_\challengelength'}\paren{C,i} \neq \bot \right] \geq \paren{\frac{1}{d} +\epsilon} \ . \]
}

\begin{claim} \label{claim:secFinal}
\claimSecFinal
\end{claim}

\begin{proofof}{Claim \ref{claim:secFinal}}
We draw examples $\paren{C_1,f\paren{\sigma\paren{C_1}}},\ldots,\paren{C_m,f\paren{\sigma\paren{C_m}}}$ to construct $b=\Adversary_{C_1,\ldots,C_m}$. Given a random length-$\challengelength$ password challenge $\paren{C_1',\ldots,C_\challengelength'} \in \paren{X_k}^t$ we let \[p_j = \Pr_{C,C_1,\ldots,C_m,C_1',\ldots,C_\challengelength'\sim X_k}\left[\mathcal{P}_{b,j,C_1',\ldots,C_t'}\paren{C} = f\paren{\sigma\paren{C}} ~\vline ~ \mathcal{P}_{b,j,C_1',\ldots,C_\challengelength'}\paren{C} \neq \bot \right] \] denote the probability that the adversary correctly guesses the response to the $j$'th challenge conditioned on the event that the adversary correctly guesses all of the earlier challenges. Observe that 
\[\Pr_{C,C_1,\ldots,C_m,C_1',\ldots,C_{\challengelength-1}'\sim X_k,i \sim[t]}\left[\mathcal{P}_{b,i,C_1',\ldots,C_\challengelength'}\paren{C,i} = f\paren{\sigma\paren{C}} \right] = \sum_{i=1}^\challengelength p_i/\challengelength \ , \]
so it suffices to show that $\sum_{i=1}^\challengelength p_i/\challengelength \geq  \frac{1}{d} +\epsilon$. We obtain the following constraint
\begin{eqnarray*}
\prod_{i=1}^\challengelength p_i &=&  \prod_{i=1}^\challengelength \Pr_{C,C_1,\ldots,C_m,C_1',\ldots,C_\challengelength'\sim X_k}\left[\mathcal{P}_{b,j,C_1',\ldots,C_\challengelength'}\paren{C} = f\paren{\sigma\paren{C}} ~\vline ~ \mathcal{P}_{b,j,C_1',\ldots,C_\challengelength'}\paren{C} \neq \bot \right]  \\
&=& \prod_{i=1}^\challengelength \Pr_{C_1,\ldots,C_m,C_1',\ldots,C_\challengelength'\sim X_k} \left[ \Adversary_{C_1,\ldots, C_m}\paren{C_1',\ldots,C_\challengelength'}[i] = f\paren{\sigma\paren{C_i'}} ~\vline~\forall j < i.~ \Adversary_{C_1,\ldots, C_m}\paren{C_1',\ldots,C_\challengelength'}[j] = f\paren{\sigma\paren{C_j'}} \right] \\
&=&  \Pr_{C_1,\ldots,C_m,C_1',\ldots,C_\challengelength'\sim X_k} \left[ \Adversary_{C_1,\ldots, C_m}\paren{C_1',\ldots,C_\challengelength'} = \paren{f\paren{\sigma\paren{C_1'}},\ldots, f\paren{\sigma\paren{C_\challengelength'}} } \right] \\
&\geq& \paren{\frac{1}{d} +\epsilon}^\challengelength\ . 
\end{eqnarray*}  

If we minimize $\sum_{i=1}^t p_i/\challengelength$ subject to the constraint $\prod_{i=1}^\challengelength p_i \geq \paren{\frac{1}{d} +\epsilon}^\challengelength$ then we obtain the desired upper bound $\sum_{i=1}^\challengelength p_i/\challengelength \geq  \frac{1}{d} +\epsilon$. 
\end{proofof}

\begin{remindertheorem}{\ref{thm:secFinal}}
\thmSecurityFinal
\end{remindertheorem}

\begin{proofof}{Theorem \ref{thm:secFinal}}
Given random clauses $C_1,\ldots,C_m,C_1',\ldots,C_\challengelength' \sim X_k$ we let $\mathbf{Good}\paren{C_1,\ldots,C_m,C_1',\ldots,C_\challengelength'}$ denote the event that \[\Pr_{i \sim[t],C\sim X_k}\left[\mathcal{P}_{b,C_1',\ldots,C_\challengelength'}\paren{C,i} = f\paren{\sigma\paren{C}}~\vline~\mathcal{P}_{b,C_1',\ldots,C_\challengelength'}\paren{C,i} \neq \bot \right] \geq \paren{\frac{1}{d} + \frac{\epsilon}{2}} \ . \]
By Markov's Inequality and Claim \ref{claim:secFinal} we have $\Pr\left[\mathbf{Good}\paren{C_1,\ldots,C_m,C_1',\ldots,C_\challengelength'}\right] \geq  \frac{\epsilon}{2}$. Here, $b=\Adversary_{C_1,\ldots,C_m}$ and 
\[
    \mathcal{P}_{b,C_1',\ldots,C_\challengelength'}\paren{C,i} = 
\begin{cases}
   b\paren{\hat{C}_1,\ldots,\hat{C}_\challengelength}[i]  ,& \text{if } f\paren{\sigma\paren{\hat{C}_j}} = b\paren{\hat{C}_1,\ldots,\hat{C}_\challengelength}[j]~~ \forall j < i \\
    \bot,              & \text{otherwise}
\end{cases}
\]
Assuming that the event $\mathbf{Good}\paren{C_1,\ldots,C_m,C_1',\ldots,C_\challengelength'}$ occurs we obtain labels for each clause $C \in X_k$ by selecting a random permutation $\pi:[\challengelength]\rightarrow [\challengelength]$, setting $i = 1$ and setting $\ell_C = \mathcal{P}_{b,C_1',\ldots,C_\challengelength'}\paren{C,\pi(i)}$ --- if $\ell_C \neq \bot$ then we increment $i$ and repeat. Note that we will always find a label $\ell_C \neq \bot$ within $t$ attempts because  $\mathcal{P}_{b,C_1',\ldots,C_\challengelength'}\paren{C,1} \neq \bot$. Let $\mathbf{GoodLabels}$ denote the event that
\[ \Pr_{C \sim X_k}\left[G_C \right] \geq \frac{1}{d} + \frac{\epsilon}{4} \ , \] 
where $G_C$ is the indicator random variable for the event $\ell_C = f\paren{\sigma\paren{C}}$. We have \[\mathbb{E}\left[\frac{1}{\left| X_k\right|}\sum_{C \in X_k} G_C\right] \geq \frac{1}{d}+\frac{\epsilon}{2} \ ,\] 
so we can invoke Markov's inequality again to argue that $\Pr\left[\mathbf{GoodLabels}~\vline~\mathbf{Good} \right] \geq \frac{\epsilon}{4}$. If the event $\mathbf{GoodLabels}$ occurs then we can invoke Theorem \ref{thm:EvenDistributedPredict} to obtain $\sigma'$ that is $\epsilon/8$-correlated with $\sigma$ with probability at least $\frac{\epsilon}{8d^2}$. Our overall probability of success is 
\[\frac{\epsilon}{8d^2} \times \frac{\epsilon}{4} \times \frac{\epsilon}{2} = \frac{\epsilon^3}{\paren{8d}^2} \ .\]

\end{proofof}

\section{Security Parameters of $f_{k_1,k_2}$} \label{apdx:ProofOfClaims}
\begin{reminderclaim}{\ref{claim:CandidateFamilySecurityParameters}}
\clmCandidateFamilySecurityParameters
\end{reminderclaim}

\begin{proofof}{Claim \ref{claim:CandidateFamilySecurityParameters}} 
Let $f(x) = f_{k_1,k_2}(x)  = x_{\paren{\sum_{i=10}^{9+k_1} x_i \mod{10}}} + \sum_{i = 10+k_1}^{9+k_1+k_2} x_i \mod{10} $. We first observe that if we fix the values of $x_{10},\ldots,x_{9+k_1} \in  \mathbb{Z}_{10}$ and let $i' = \sum_{i=10}^{9+k_1} x_i \mod{10}$ then $f'\paren{x_0,\ldots,x_9,x_{10+k_1},\ldots,x_{9+k_1+k_2}} = x_{i'} + \sum_{i=10+k_1}^{9+k_1+k_2} x_i \mod{10}$ is a linear function. Similarly, if we fix the values of $x_0=\ldots=x_9 = c$ then the resulting function $f'\paren{x_{10},\ldots,x_{9+k_1+k_2}}=c+\sum_{i = 10+k_1}^{9+k_1+k_2} x_i \mod{10}$ is linear.  Thus, $g(f) \leq \min\{10,k_1\}$. Now suppose that we don't fix all of the values $x_{10},\ldots,x_{9+k_1} \in  \mathbb{Z}_{10}$ and at least one of the variables $x_0,\ldots,x_9$ is not fixed. In this case the resulting function will not be linear. Thus, $g\paren{f} \geq \min\{k_1,10\}$. We also note that for any $\alpha \in \mathbb{Z}_{10}^{10+k_1+k_2}$ s.t. $H\paren{\alpha} \leq k_2$ and $i,t \in  \mathbb{Z}_{10}$ that
\[ \Pr_{x\sim \mathbb{Z}_{10}^{10+k_1+k_2}}\left[f(x) = t ~\vline ~\alpha\cdot x \equiv i \mod{10} \right] =  \Pr_{x\sim \mathbb{Z}_{10}^{10+k_1+k_2}}\left[f(x) = t  \right] = \frac{1}{10} \ . \]
Therefore,
\begin{eqnarray*}
\hat{Q}^{f,t}_\alpha &=& \mathbb{E}_{x\sim \mathbb{Z}_{10}^{10+k_1+k_2}}\left[Q^{f,t}\paren{x} \chi_\alpha\paren{x} \right] \\
&=& \sum_{i =0}^9 \Pr\left[\alpha\cdot x \equiv i \mod{10}\right] \mathbb{E}_{x\sim \mathbb{Z}_{10}^{10+k_1+k_2}}\left[Q^{f,t}\paren{x} \chi_\alpha\paren{x}~\vline~\alpha\cdot x \equiv i \mod{10} \right] \\
&=& \sum_{i =0}^9 \exp\paren{\frac{-2\pi i\sqrt{-1}}{10}} \Pr\left[\alpha\cdot x \equiv i \mod{10}\right] \mathbb{E}_{x\sim \mathbb{Z}_{10}^{10+k_1+k_2}}\left[Q^{f,t}\paren{x} ~\vline~\alpha\cdot x \equiv i \mod{10} \right] \\
&=& \frac{1}{10} \sum_{i =0}^9 \exp\paren{\frac{-2\pi i\sqrt{-1}}{10}} \mathbb{E}_{x\sim \mathbb{Z}_{10}^{10+k_1+k_2}}\left[Q^{f,t}\paren{x} ~\vline~\alpha\cdot x \equiv i \mod{10} \right] \\
&=& 0 \ ,
\end{eqnarray*}
 which implies that $r(f) \geq k_2+1$. Similarly, if we set $\alpha = (\alpha_0,\ldots,\alpha_{9+k_1+k_2})$ such that $\alpha_0=1$ and $\alpha_{10+k_1}=\ldots=\alpha_{9+k_1+k_2} =1$ so that $\alpha$ has hamming weight $k_2+1$ then we can verify that $\hat{Q}^{f,t}_\alpha \neq 0$. 
\end{proofof}

\section{Security Upper Bounds} \label{apdx:subsec:securityUpperBound}

\subsection{Statistical Algorithms}
Theorem \ref{thm:StatisticalUpperBound} demonstrates that our lower bound for statistical algorithms are asymptotically tight for our human computable password schemes $f_{k_1,k_2}$.  In particular, we demonstrate that $m=\tilde{O}\paren{n^{(k_2+1)/2}}$ queries to $\MSAMPLE$ are sufficient for a statistical algorithm to recover $\sigma$.

\newcommand{\thmStatisticalUpperBound}{For $f =f_{k_1,k_2}$ there is a randomized algorithm that makes $O\paren{n^{\max\left\{1,(k_2+1)/2\right\}} \log^2 n}$ calls to the $\MSAMPLE\paren{n^{\lceil r(f_i)/2 \rceil}}$ oracle and returns $\sigma$ with probability $1-o(1)$.}
\begin{theorem}\label{thm:StatisticalUpperBound}
\thmStatisticalUpperBound
\end{theorem}

For binary functions $f':\{0,1\}^k\rightarrow\{0,1\}$, Feldman et al. \cite{feldman2013complexity} gave a randomized statistical algorithm to find $\sigma' \in \{0,1\}^n$ using just $O\paren{n^{r(f)/2} \log^2 n}$ calls to the $\MSAMPLE\paren{n^{\lceil r(f)/2 \rceil}}$ oracle. Their main technique is a discrete spectral iteration procedure to find the eigenvector (singular vector) with the largest eigenvalue (singular value) of a matrix $M$ sampled from a distribution $M_{\sigma',p}$ over $\left|X_{\lfloor r(f)/2\rfloor}\right|\times\left|X_{\lceil r(f)/2\rceil}\right| $ matrices. With probability $1-o(1)$ this eigenvector will encode the value $\sum_{i \in C}\sigma'\paren{i} \mod{2}$ for each clause $C \in X_{r(f)/2}$. We show that the discrete spectral iteration algorithm of Feldman et al \cite{feldman2013complexity} can be extended to recover $\sigma \in \mathbb{Z}_{10}$ when $f_{k_1,k_2}$ is one of our candidate human computable functions. 

\paragraph{Discussion} We note that Theorem \ref{thm:StatisticalUpperBound} cannot be extended to arbitrary functions $f:\mathbb{Z}_d^k\rightarrow\mathbb{Z}_d$. Consider for example the unique function $f:\mathbb{Z}_{10}^6\rightarrow \mathbb{Z}_{10}$ s.t. $f(x_1,\ldots,x_6) \equiv  f'\paren{x_1 \mod{2},\ldots,x_6 \mod{2}}   \mod{2}$ and  $f(x_1,\ldots,x_6) \equiv f''\paren{x_1 \mod{5},\ldots,x_6 \mod{5}} \mod{5}$, where $f':\mathbb{Z}_2^6 \rightarrow\mathbb{Z}_2$ and $f'':\mathbb{Z}_5^6\rightarrow\mathbb{Z}_5$. By the Chinese Remainder Theorem instead of picking a secret mapping $\sigma \in \mathbb{Z}_{10}^n$ we could equivalently pick the unique secret mappings $\sigma_1 \in \mathbb{Z}_2^n$ and $\sigma_2 \in \mathbb{Z}_5^n$ s.t $\sigma \equiv \sigma_1 \mod{2}$ and $\sigma \equiv \sigma_2 \mod{5}$. Now drawing challenge response pairs from the distributions $Q^{f}_\sigma$ is equivalent to drawing challenge-response pairs from the distributions $Q^{f'}_{\sigma_1}$ and  $Q^{f''}_{\sigma_2}$. Suppose that $f'(x_1,\ldots,x_6) = x_1x_2 + x_3+x_4+x_5+x_6 \mod{2}$, and $f''(x_1,\ldots,x_6) = x_1$. Then we have $r(f) = \min\paren{r(f'),r(f'')} = r(f'') = 1$, but $r(f') = 4$. We can find $\sigma_2$ using $O\paren{n \log^2 n}$ calls to $\MSAMPLE(n)$, but to find $\sigma$ we must also recover $\sigma_1$. This provably requires at least $\tilde{\Omega}\paren{n^{r(f')/2}} = \tilde{\Omega}\paren{n^{2}}$ calls to $\MSAMPLE\paren{n^2}$.

\paragraph{Background} The proof of Theorem \ref{thm:StatisticalUpperBound} relies on the discrete spectral iteration algorithm of \cite{feldman2013complexity}. We begin by providing a brief overview of their algorithm. In their setting the secret mapping $\sigma$ is defined over the binary alphabet $\mathbb{Z}_2^n$.  Let $c_1 = \lceil \frac{r(f)}{2} \rceil, c_2 = \lfloor \frac{r(f)}{2} \rfloor$ and let $\delta \in [0,2]\backslash\{1\}$. They use $\sigma$ to define a distribution over $\left|X_{c_1}\right| \times \left|X_{c_2}\right|$ matrices $M_{\sigma,\delta,p} = \hat{M}\paren{Q_{\sigma,\delta,p}}- Jp$, where $J$ denotes the all ones matrix. For $\paren{C_1} \in  X_{ c_1}$, $\paren{C_2} \in X_{c_2}$ such that $C_1\bigcap C_2 = \emptyset$ we have
\[\hat{M}\paren{Q_{\sigma,\delta,p}}\left[\paren{C_1},\paren{C_2} \right] =  \begin{cases}
    1,& \text{with probability $\paren{p\paren{2-\delta}}$ if } \sum_{j \in C_1\cup C_2} \sigma\paren{j}  \equiv 0 \mod{2}\\
    1,& \text{with probability $\paren{p\delta}$ if }  \sum_{j \in C_1\cup C_2} \sigma\paren{j} \not\equiv 0 \mod{2}\\
    0,& \text{otherwise }
\end{cases}
\ .
\]
Given a vector $x \in \{\pm 1\}^{\left|X_{c_2}\right|}$ (resp. $y \in \{\pm 1\}^{\left|X_{c_1}\right|}$) $M_{\sigma,\delta,p}x$ defines a distribution over vectors in $\mathbb{R}^{\left|X_{c_1}\right|}$ (resp. $M_{\sigma,p}^Ty$ defines a distribution over vectors in $\mathbb{R}^{\left|X_{c_1}\right|}$). \\

If $r(f)$ is even then the largest eigenvalue of $\mathbb{E}\left[M_{\sigma,\delta,p}\right]$ has a corresponding eigenvector $x^* \in \{\pm 1\}^{X_{r(f)/2}}$, where for $C_i \in X_{r(f)/2}$ we have $x^*\left[C_i\right] = 1$ if $\sum_{j \in C_i} \sigma(j) \equiv 1 \mod{2}$; otherwise $x^*\left[C_i\right] = -1$ (if $r(f)$ is odd then we consider the top singular value instead). Feldman et al \cite{feldman2013complexity} use discrete spectral iteration to find $x^*$ (or $-x^*$). Given $x^*$  it is easy to find $\sigma$ using Gaussian Elimination. \\
The discrete spectral iteration algorithm of Feldman et al \cite{feldman2013complexity} starts with a random vector $x^0 \in \{0,1\}^{\left|X_{c_2}\right|}$. They then sample $x^{i+1} \sim M_{\sigma,p}x^i$ and execute a normalization step to ensure that $x^{i+1} \in \{0,1\}^{\left|X_{k_2}\right|}$. When $r(f)$ is odd, power iteration has two steps: draw a sample $y^i \sim M_{\sigma,\delta,p}x^i$ and sample from the distribution $x^{i+1} = M_{\sigma,\delta,p}^T y^i$. They showed that $O\paren{\log \left|X_{r(f)}\right|}$ iterations suffice to recover $\sigma$ whenever $p = \frac{K \log \left|X_{r(f)}\right|}{\paren{\delta-1}^2\sqrt{\left|X_{r(f)}\right|}}$, and that for a vector $x \in \{0,1\}^{\left|X_{k_2}\right|}$ (resp. $y \in \{\pm 1\}^{\left|X_{k_1}\right|}$)  it is possible to sample from  $M_{\sigma,\delta,p}x$ (resp. $M_{\sigma,\delta,p}^Ty$) using $O\paren{1/p}$ queries to $\MSAMPLE\paren{\left|X_{c_1} \right|}$.

\paragraph{Our Reduction} 
The proof of Theorem \ref{thm:StatisticalUpperBound} uses a reduction to the algorithm of Feldman et al\cite{feldman2013complexity}.


\begin{proofof}{Theorem \ref{thm:StatisticalUpperBound}} (sketch)
 Given a mapping $\sigma \in \mathbb{Z}_d^n$ and a number $i \in\mathbb{Z}_d$ we define a mapping $\sigma_i \in \mathbb{Z}_2^{n}$ where 
\[\sigma_i(j)  =  \begin{cases}
    1,& \text{\bf if } \sigma\paren{j} = i\\
    0,& \text{otherwise }
\end{cases}
\ .
\]
Clearly, to recover $\sigma$ it is sufficient to recover $\sigma_i$ for each $i \in \mathbb{Z}_d$. Therefore, to prove Theorem \ref{thm:StatisticalUpperBound} it suffices to show that given  $x \in \{\pm 1\}^{\left|X_{k_2}\right|}$ (resp. $y \in \{\pm 1\}^{\left|X_{k_1}\right|}$) we can sample from the distribution $M_{\sigma_i,\delta,p}x$ (resp. $M_{\sigma_i,\delta,p}^Ty$) using $O\paren{1/p}$ queries to $\MSAMPLE\paren{\left|X_{\lceil r(f)/2 \rceil} \right|}$  for each $i \in \{0,\ldots,d-1\}$, where $\MSAMPLE$ uses the distribution $Q^f_\sigma$. In general, this will not possible for arbitrary functions $f$. However, Lemma \ref{lemma:StatisticalUpperBound} shows that for our candidate human computable functions $f_{1,3},f_{2,2}$ we can sample from the distributions $M_{\sigma_i,\delta,p}x$ (resp. $M_{\sigma_i,\delta,p}^Ty$). The proof of Lemma \ref{lemma:StatisticalUpperBound} is similar to the proof of \cite[Lemma 10]{feldman2013complexity}.
\end{proofof}

\begin{lemma} \label{lemma:StatisticalUpperBound}
Given vectors $\vec{x} \in \{\pm 1\}^{\left|X_{c_1}\right| }, \vec{y}\in \{\pm 1\}^{\left|X_{c_2}\right| } $  we can sample from $M_{\sigma_i,\delta,p}x$ and $M_{\sigma,\delta,p}^Ty$ using $O\paren{n^{(k_2+1)/2} \log^2 n}$ calls to the $\MSAMPLE\paren{n^{\lceil r(f)/2 \rceil}}$ oracle for $f= f_{k_1,k_2}$.
\end{lemma}

The proof of Lemma \ref{lemma:StatisticalUpperBound} relies on Fact \ref{fact:Candidate1StatisticalUpperBound}. 

\begin{fact} \label{fact:Candidate1StatisticalUpperBound}
For each $j,t \in \mathbb{Z}_{10}$ we have
\[\Pr_{\paren{x_0,\ldots,x_{9+k_1+k_2}}\sim \mathbb{Z}_{10}^{10+k_1+k_2}}\left[x_t +\sum_{i=10+k_1}^{10+k_1+k_2} x_{i} \equiv j  ~\vline~f_{k_1,k_2}\paren{x_0,\ldots,x_{9+k_1+k_2}} \equiv j \mod{10} \right]= \frac{\paren{\frac{9}{10}\paren{\frac{1}{10}}+\frac{1}{10}}\paren{\frac{1}{10}}}{\paren{\frac{1}{10}}} = \frac{19}{100}\ , \]
and
\[\Pr_{\paren{x_0,\ldots,x_{9+k_1+k_2}}\sim  \mathbb{Z}_{10}^{10+k_1+k_2}}\left[x_t +\sum_{i=10+k_1}^{10+k_1+k_2} x_{i} \equiv j ~\vline~f_{k_1,k_2}\paren{\sigma\paren{x_0,\ldots,x_{9+k_1+k_2}}} \not\equiv j \mod{10} \right] = \frac{\paren{\frac{9}{10}\paren{\frac{1}{10}} +\frac{1}{10}\paren{0}}\paren{\frac{1}{10}}}{\paren{\frac{1}{10}}} = \frac{9}{100}\ . \]
\end{fact}

\cut{\begin{fact}\label{fact:Candidate2StatisticalUpperBound}
For each $j,t \in  \mathbb{Z}_{10}$ we have
\[\Pr_{\paren{x_0,\ldots,x_{13}}\sim \mathbb{Z}_{10}^{14}}\left[ x_{11}+x_{12}+x_{13}+x_t \equiv j  ~\vline~f_{1,3}\paren{\sigma\paren{x_0,\ldots,x_{13}}} \equiv j \mod{10} \right]= \frac{\paren{\frac{9}{10}\paren{\frac{1}{10}}+\frac{1}{10}}\paren{\frac{1}{10}}}{\paren{\frac{1}{10}}} = \frac{19}{100}\ , \]
and
\[\Pr_{\paren{x_0,\ldots,x_{13}}\sim  \mathbb{Z}_{10}^{14}}\left[ x_{11}+x_{12}+x_{13}+x_t \equiv j  ~\vline~f_{1,3}\paren{\sigma\paren{x_0,\ldots,x_{13}}} \not\equiv j \mod{10} \right] = \frac{\paren{\frac{9}{10}\paren{\frac{1}{10}} +\frac{1}{10}\paren{0}}\paren{\frac{1}{10}}}{\paren{\frac{1}{10}}} = \frac{9}{100}\ . \]
\end{fact}}

\begin{proofof}{Lemma \ref{lemma:StatisticalUpperBound}}
Given a value $j \in \mathbb{Z}_{10}$ and a value $i \in \mathbb{Z}_{10}$ we let $x_j^i \in \{0,1\}$ be the indicator variable for the event $x_j = i$. By Fact \ref{fact:Candidate1StatisticalUpperBound} it follows that 
\begin{eqnarray*}& &\Pr_{\paren{x_0,\ldots,x_{k_1+k_2+9}}\sim \mathbb{Z}_{10}^{k_1+k_2+10}}\left[x_0^i+ x_{9+k_1}^i+\ldots+x_{7+k_1+c_1}^i \equiv 1 \mod{2}  ~\vline~f_{k_1,k_2}\paren{\sigma\paren{x_0,\ldots,x_{9+k_1+k_2}}} \equiv ic_1 \mod{10} \right] \\&\neq& \Pr_{\paren{x_0,\ldots,x_{k_1+k_2+9}}\sim \mathbb{Z}_{10}^{k_1+k_2+10}}\left[ x_0^i+ x_{9+k_1}^i+\ldots+x_{7+k_1+c_1}^i \equiv 1 \mod{2}  ~\vline~f_{k_1,k_2}\paren{\sigma\paren{x_0,\ldots,x_{9+k_1+k_2}}} \not\equiv ic_1 \mod{10} \right] \ . 
\end{eqnarray*}
Now for $f_{k_1,k_2}$ we define the function $h^{i,+}:X_{k_1+k_2+10}\times\mathbb{Z}_{10}\rightarrow X_{c_1}\cup \{\bot\}$  as follows 
\[h^{i,+}\paren{x_0,\ldots,x_{9+k_1+k_2},f_{k_1,k_2}\paren{\sigma\paren{x_0,\ldots,x_{9+k_1+k_2}}}} = 
\begin{cases}
\paren{x_0,x_{9+k_1},\ldots,x_{7+k_1+c_1}} & \text{if} f_{k_1,k_2}\paren{\sigma\paren{x_0,\ldots,x_{13}}} \equiv ic_1  \mod{10}  \\
\bot & \text{otherwise.} 
\end{cases} \ .\]
Intuitively, given a clause $C_1 \in X_{c_1}$ the probability that $h^{i,+}$ returns $C_1$ is greater if $\sum_{j \in C_1} \sigma_i(j) \equiv 1 \mod{2}$.

Given a vector $x \in \{\pm 1\}^{\left|X_{c_1}\right|}$ we query our $\MSAMPLE\paren{\left|X_{c_1} \right|+1}$ oracle  $\lceil 10/p\rceil$ times with the function $h^{i,+}$ to sample from $M_{\sigma_i,\delta,p}x$. Let $q_1,\ldots,q_{\lceil 10/p\rceil} \in X_{c_1}$ denote the responses and let $x\left[q_j\right]$ denote the value of the vector $x$ at index $q_j$. We observe that for some $\delta \neq 1$ we have \[ M_{\sigma_i,\delta,p}x[C] \sim \sum_{\substack{j \in \lceil 10/p\rceil\\q_j \neq \bot }} x\left[q_j\right]  - p  \sum_{C' \in X_{c_1}}  x[C'] \ ,\]
for every  $C \in X_{c_2}$.  
\end{proofof}

\paragraph{Open Question} Can we precisely characterize the functions $f:\mathbb{Z}_d^k\rightarrow\mathbb{Z}_d$ for which we can efficiently recover $\sigma$ after seeing $\tilde{O}\paren{n^{r(f)/2}}$ challenge-response pairs? Feldman et al. \cite{feldman2013complexity} gave a statistical algorithm that recovers the secret mapping whenever $d=2$ after making $\tilde{O}\paren{n^{r(f)/2}}$ queries to $\MSAMPLE\paren{n^{r(f)/2}}$. While we show that the same algorithm can be used to recover $\sigma$ after making $\tilde{O}\paren{n^{r(f)/2}}$ queries to $\MSAMPLE\paren{n^{r(f)/2}}$ in our candidate human computable password schemes with $d=10$, we also showed that these results do not extend to all functions $f:\mathbb{Z}_d^k\rightarrow\mathbb{Z}_d$.

\subsection{Gaussian Elimination} \label{subsec:GuassianElimination}
Most known algorithmic techniques can be modeled within the statistical query framework. Gaussian Elimination is a notable exception. As an example consider the function $f(x_1,\ldots,x_7) = x_1+\ldots+x_7 \mod{10}$ (in this example $r(f) = 7$ and $g(f) = 0$). Our previous results imply that any statistical algorithm would need to see at least $m = \tilde{\Omega}\paren{n^{7/2}}$ challenge response pairs $\paren{C,f\paren{\sigma\paren{C}}}$ to recover $\sigma$. However, it is trivial to recover $\sigma$ from $O(n)$ random challenge response pairs using Gaussian Elimination. In general, consider the following attacker shown in algorithm \ref{alg:GaussianAttack}, which uses Gaussian Elimination. Algorithm \ref{alg:GaussianAttack} relies on the subroutine $\mathbf{TryExtract}\paren{C,f\paren{\sigma\paren{C}},S,\alpha}$, which attempts to extract a linear constraint from $\paren{C,f\paren{\sigma\paren{C}}}$ under the assumption that $\sigma\paren{S} = \alpha$. We assume $\mathbf{TryExtract}\paren{C,f\paren{\sigma\paren{C}},S,\alpha}$ returns $\emptyset$ if it cannot extract a linear constraint. For example, if we assume that $\sigma(1) = 4$ and $\sigma(2) = 8$ and let $C = \paren{i_0,i_1,...,i_{9}, 1, 2, i_{10},i_{11}}$ (with $i_j \in [n] \setminus \{1,2\}$) then we have $f_{2,2}\paren{\sigma\paren{C}}$ $= \sigma\paren{i_{4+8 \mod{10}}} + \sigma\paren{i_{10}} + \sigma\paren{i_{11}} \mod{10}$. In this case, $\mathbf{TryExtract}\paren{C,f\paren{\sigma\paren{C}},\{1,2\},\{4,8\}}$ would return the constraint $f\paren{\sigma\paren{C}}= \sigma\paren{i_{2}} + \sigma\paren{i_{10}} + \sigma\paren{i_{11}} \mod{10} $. However, $\mathbf{TryExtract}\paren{C,f\paren{\sigma\paren{C}},\{i_{0},2\},\{4,8\}}$ would return $\emptyset$.

\begin{algorithm}[H]
\caption{$\mathbf{GaussianAttack}$} \label{alg:GaussianAttack}
\SetKwInOut{Input}{input}
\Input{Clauses $C_1,\ldots, C_m \sim X_k$, and labels $f\paren{\sigma\paren{C_1}},\ldots,f\paren{\sigma\paren{C_m}}$.}
\ForAll{$S \in X_{g(f)}$, $\alpha \in \mathbb{Z}_d^{g(f)}$}{
      $\mathbf{LC} \leftarrow \emptyset$ \; 
\tcp{$\mathbf{LC}$ is the set of linear constraints extracted}
      \ForAll{$C \in \left\{C_1,\ldots,C_m\right\}$}{
          $\mathbf{LC} \leftarrow \mathbf{LC}  \bigcup \mathbf{TryExtract}\paren{C,f\paren{\sigma\paren{C}},S,\alpha}$ \;
\If{$\left| \mathbf{LC} \right| \geq n$}{ 
 $\sigma' \leftarrow \mathbf{LinearSolve}\paren{LC}$ \;
\If{ $\forall i \in [m]$. $f\paren{\sigma'\paren{C_i}}=f\paren{\sigma\paren{C_i}} \in C$ }{ \Return{$\sigma'$} 
}
}
}
}
\cut{\begin{algorithmic}
\State {\bf Input:} Clauses $C_1,\ldots, C_m \sim X_k$, and labels $f\paren{\sigma\paren{C_1}},\ldots,f\paren{\sigma\paren{C_m}}$. 

\ForAll{$S \in X_{g(f)}$, $\alpha \in \mathbb{Z}_d^{g(f)}$}
      \State $\mathbf{LC} \gets \emptyset$ \Comment{$\mathbf{LC}$ is the set of linear constraints extracted}
      \ForAll{$C \in \left\{C_1,\ldots,C_m\right\}$}
          \State $\mathbf{LC} \gets \mathbf{LC}  \bigcup \mathbf{TryExtract}\paren{C,f\paren{\sigma\paren{C}},S,\alpha}$ 
\If{$\left| \mathbf{LC} \right| \geq n$} 
\State $\sigma' \gets \mathbf{LinearSolve}\paren{LC}$
\If{ $\forall i \in [m]$. $f\paren{\sigma'\paren{C_i}}=f\paren{\sigma\paren{C_i}} \in C$ } \Return{$\sigma'$} 
\EndIf

\EndIf
\EndFor
\EndFor

\end{algorithmic}}
\end{algorithm}

 Fact \ref{fact:GaussianSecurity} says that an attacker needs at least $m=\tilde{\Omega}\paren{n^{1+g(f)}}$ challenge-response pairs to recover $\sigma$ using Gaussian Elimination. This is because the probability that $\mathbf{TryExtract}\paren{C,f\paren{\sigma\paren{C}}S,\alpha}$ extracts a linear constraint is at most $O\paren{\paren{\frac{\left|S\right|}{n}}^{-g(f)}}$, which is  $O\paren{n^{-g(f)}}$ for $\left|S\right|$ constant. The adversary needs $O(n)$ linearly independent constraints to run Gaussian Elimination. If the adversary can see at most $\tilde{O}\paren{n^{s(f)}}$ examples neither approach (Statistical Algorithms or Gaussian Elimination) can be used to recover $\sigma$. 

\begin{fact} \label{fact:GaussianSecurity}
Algorithm \ref{alg:GaussianAttack} needs to see at least  $m=\tilde{\Omega}\paren{n^{1+g(f)}}$ challenge-response pairs to recover $\sigma$.
\end{fact}

Remark \ref{remark:security} explores the tradeoff between the adversary's running time and the number of challenge-response pairs that an adversary would need to see to recover $\sigma$ using Gaussian elimination. In particular the adversary can recover $\sigma$ from $\tilde{O}\paren{n^{1+ g(f)/2}}$ challenge-response pairs if he is willing to increase his running time by a factor of $d^{\sqrt{n}}$. In practice, this attack may be reasonable for $n \leq 100$ and $d=10$, which means that it may be beneficial to look for candidate human computable functions $f$ that maximize $\min\{r(f)/2,1+g(f)/2\}$ instead of $s(f)$ whenever $n \leq 100$.

\begin{remark} \label{remark:security}
If the adversary correctly guesses value of $\sigma\paren{S}$ for $\left|S\right| = n^\epsilon$ then he may be able to extract a linear constraint from a random example with probability $\Omega(1/n^{\paren{1-\epsilon}g(f)})$.  The adversary would only need $\tilde{O}\paren{n^{1+\paren{1-\epsilon}g(f)}}$ examples to solve for $\sigma$, but his running time would be proportional to $d^{\epsilon n}$ --- the expected number of guesses before he is correct.
\end{remark}


\section{Rehearsal Model} \label{apdx:Rehearsal}
In this section we review the usability model of Blocki et al. \cite{NaturallyRehearsingPasswords}. Their usability model estimates the `extra effort' that a user needs to expend to memorize and rehearse all of his secrets for a password management scheme. In this section we use $(\hat{c},\hat{a})$ to denote a cue-association pair, and we use the variable $t$ to denote time (days). In our context $(\hat{c},\hat{a})$ might denote the association between a letter (e.g., `e') and the secret digit associated with that letter (e.g., $\sigma\paren{e}$). If the user does not rehearse an association $(\hat{c},\hat{a})$ frequently enough then the user might forget it. Their are two main components to their usability model: rehearsal requirements and visitation schedules. Rehearsal requirements specify how frequently a cue-association pair must be used for a user to remember the association. Visitation schedules specify how frequently the user authenticates to each of his accounts and rehearses any cue-association pairs that are linked with the account.

\subsection{Rehearsal Requirements}
 
Blocki et al. \cite{NaturallyRehearsingPasswords} introduce a rehearsal schedule to ensure that the user remembers each cue-association pair.  
\begin{definition} \cite{NaturallyRehearsingPasswords}
\label{def:RehearsalRequirement} A rehearsal schedule for a
cue-association pair $(\hat{c},\hat{a})$ is a sequence of times $t_0^{\hat{c}} < t_1^{\hat{c}} <...$. For each $i \geq 0$ we have a {\em rehearsal requirement}, the cue-association pair must be rehearsed at least once during the time window $\left[t_i^{\hat{c}},t_{i+1}^{\hat{c}}\right) = \{x \in \mathbb{R} ~ \vline~
t_i^{\hat{c}} \leq x < t_{i+1}^{\hat{c}} \}$.  
\end{definition} 

A rehearsal schedule is {\em sufficient} if a user can maintain the association $\left(\hat{c},\hat{a}\right)$ by following the
rehearsal schedule. The length of each
interval $\left[t_{i}^{\hat{c}},t_{i+1}^{\hat{c}}\right)$ may depend on the strength of the mnemonic techniques used to memorize and rehearse a cue-association pair $\paren{\hat{c},\hat{a}}$
as well as $i$ --- the number of prior rehearsals \cite{memory:ExpandingRehearsal,memory:alternatevan}. 

{\bf Expanding Rehearsal Assumption \cite{NaturallyRehearsingPasswords}:} The rehearsal schedule given by $t_i^{\hat{c}} = 2^{i\AssociationStrength{}}$ is sufficient to maintain the association $(\hat{c},\hat{a})$, where $\AssociationStrength{} > 0$ is a constant.  \\

\subsection{Visitation Schedules}
Suppose that the user has $m$ accounts $A_1,\ldots,A_m$. A visitation schedule for an account $A_i$ is a sequence of real numbers
$\tau_0^i < \tau_1^i < \ldots$, which represent the times when the account
$A_i$ is visited by the user. Blocki et al. \cite{NaturallyRehearsingPasswords} do not assume that the exact visitation
schedules are known a priori. Instead they model visitation schedules using a
random process with a known parameter $\lambda_i$ based on
$E\left[\tau_{j+1}^i-\tau_j^i\right]$ --- the average time between consecutive
visits to account $A_i$. 

 A rehearsal requirement $\left[t_i^{\hat{c}},t_{i+1}^{\hat{c}}\right)$
can be satisfied naturally if the user visits a site $A_j$ that
uses the cue $\hat{c}$ $\paren{\hat{c} \in c_j}$  during the given time window. Here, $c_j$ denote the set of cue-association pairs that the user must remember when logging into account $A_j$.  Formally,

\begin{definition} \cite{NaturallyRehearsingPasswords} We say that a rehearsal requirement
$\left[t_i^{\hat{c}},t_{i+1}^{\hat{c}}\right)$ is {\em naturally satisfied} by a visitation
schedule $\tau_0^i < \tau_1^i < \ldots$ if $\exists j \in [m],k \in \mathbb{N}$ s.t $\hat{c} \in c_j$ and $\tau_k^j
\in \left[t_i^{\hat{c}},t_{i+1}^{\hat{c}}\right)$.  We use 
\[\ExtraRehearsals{t}{\hat{c}} = \left|\left\{i ~\vline ~ t_{i+1}^{\hat{c}}
\leq t \wedge \forall j,k. \left(\hat{c} \notin c_j \vee \tau_k^j \notin
\left[t_i^{\hat{c}},t_{i+1}^{\hat{c}} \right) \right) \right\} \right| \ , \]
to denote the number of rehearsal requirements that are not naturally satisfied by the visitation schedule during the time interval $[0,t]$.
\end{definition}

{\bf Example: } Consider the human computable function $f_{2,2}$ from section \ref{sec:Candidates}, and suppose that the user has to compute $f_{2,2}\paren{\sigma\paren{C_i}}$ to authenticate at account $A_j$, where $C_i = \paren{x_0,\ldots,x_{13}}$. When the user computes $f_{2,2}$ he must rehearse the associations $\paren{x_{10},\sigma\paren{x_{10}}}$, $\paren{x_{11},\sigma\paren{x_{11}}}$, $\paren{x_{12},\sigma\paren{x_{12}}}$, $\paren{x_{13},\sigma\paren{x_{13}}}$ and $\paren{x_{i},\sigma\paren{x_{i}}}$ where $i =\paren{ \sigma\paren{x_{10}}+ \sigma\paren{x_{11}} \mod{10}}$. Thus $c_j\supset \{x_i,x_{10},x_{11},x_{12},x_{13}\}$. When user authenticates he naturally rehearses each of these associations in $c_j$. \\

If a cue-association pair $\left(\hat{c},\hat{a}\right)$ is not rehearsed naturally during the interval $\left[t^{\hat{c}}_i,t^{\hat{c}}_{i+1}\right)$ then the user needs to perform an extra rehearsal to maintain the association. Intuitively, $\ExtraRehearsals{t}{\hat{c}}$ denotes the total number of extra rehearsals of the cue-association pair $\left(\hat{c},\hat{a}\right)$ during the time interval $[0,t]$, and $\TotalExtraRehearsals{t} = \sum_{\hat{c} \in C}
\ExtraRehearsals{t}{\hat{c}}$ denotes the total number of extra rehearsals during the time interval $[0,t]$ to maintain all of the cue-association pairs. Thus, a smaller value of $E\left[\TotalExtraRehearsals{t}\right]$ indicates that the user needs to do less extra work to rehearse his secret mapping. \\

\paragraph{Poisson Arrival Process}
The visitation schedule for each account $A_j$ is given by a Poisson arrival process with parameter $\lambda_j$, where $1/\lambda_j=E\left[\tau_{j+1}^i-\tau_j^i\right]$ denotes the average time between consecutive visits to account $A_j$. 

\begin{table}
\centering
\begin{tabular}{| l | c | c | c | c | c |}
\hline
Schedule $\vline~ \lambda$& $\frac{1}{1 }$  & $\frac{1}{3 }$  & $\frac{1}{7}$  & $\frac{1}{31 }$  & $\frac{1}{365 }$ \\
\hline
Very Active & 10 & 10 & 10 & 10 & 35 \\
\hline
Typical & 5 & 10 & 10 & 10 & 40 \\
\hline
Occasional & 2 & 10 & 20 & 20 & 23 \\
\hline
Infrequent & 0 & 2 & 5 & 10 & 58 \\
\hline 
\end{tabular}
\caption{Visitation Schedules - number of accounts visited with frequency $\lambda$ (visits/days)}
\label{tab:userSchedules}
\end{table}

\paragraph{Evaluating Usability}  Blocki et al. \cite{NaturallyRehearsingPasswords} prove the following theorem. Given a sufficient rehearsal schedule and a visitation schedule, 
Theorem \ref{thm:ExtraRehearsals} predicts the value of $\TotalExtraRehearsals{t}$, the total number of extra rehearsals that a user will need to do to remember all of the cue-association pairs required to reconstruct all of his passwords. 
\newcommand{\thmExtraRehearsals}{Let $i_{\hat{c}}* = \left(\arg\max_x t^{\hat{c}}_{x} < t \right)-1$ then \begin{eqnarray*}
E\left[\TotalExtraRehearsals{t} \right] = \sum_{\hat{c} \in C} \sum_{i=0}^{i_{\hat{c}}*} \exp \left(-\left(\sum_{\substack{j~s.t. \\ \hat{c} \in c_j}} \lambda_j \right)\left(t^{\hat{c}}_{i+1}-t^{\hat{c}}_i \right) \right) \ 
\end{eqnarray*}  

}

\begin{theorem} \cite{NaturallyRehearsingPasswords} \label{thm:ExtraRehearsals}
\thmExtraRehearsals
\end{theorem}

We use the formula from Theorem \ref{thm:ExtraRehearsals} to obtain the usability results in Table \ref{tab:Usability}. To evaluate this formula we need to be given the rehearsal requirements, a visitation schedule ($\lambda_i$) for each account $A_i$ and a set of public challenges $\vec{C}_i \in \paren{X_{14}}^{10}$ for each account $A_i$. The rehearsal requirements are given by the Expanding Rehearsal Assumption \cite{NaturallyRehearsingPasswords} (we use the same association strength parameter $\AssociationStrength{} = 1$ as Blocki et al. \cite{NaturallyRehearsingPasswords}), and the visitation schedules for each user are given in Table \ref{tab:userSchedules}. We assume that each password is $10$ digits long and that the challenges $\vec{C}_i \in \paren{X_{14}}^{10}$ are chosen at random by Algorithm \ref{alg:GenStories}. Notice that each time the user responds to a single digit challenge he rehearses the secret mapping at five locations (see discussion in Section \ref{subsec:Usability}). Because the value of $\mathbb{E}\left[\TotalExtraRehearsals{365}\right]$ depends on the particular password challenges that we generated for each account, we ran Algorithm \ref{alg:GenStories} and computed the resulting value  $\mathbb{E}\left[\TotalExtraRehearsals{365}\right]$ one-hundred times. The values in Table \ref{tab:Usability} represent the mean value of $\mathbb{E}\left[\TotalExtraRehearsals{365}\right]$ across all hundred instances.

\begin{table}
\parbox{.65\linewidth}{
\centering
\begin{tabular}{| l | l | l | l || l | l | l | }
\hline
 &\multicolumn{3}{c||}{Our Scheme $\paren{\sigma \in \mathbb{Z}_{10}^n}$} & \multicolumn{3}{|c|}{\sharedCues} \\
\hline
 User & $n=100$ & $n=50$ & $n=30$ & SC-0 & SC-1 & SC-2 \\
\hline
Very Active & $0.396$ & $0.001$ & $\approx 0$ & $\approx 0$ & $3.93$ & $7.54$ \\
\hline
Typical & $2.14$ & $0.039$ & $\approx 0$ & $\approx 0$ & $10.89$ & $19.89$ \\
\hline 
Occasional & $2.50$ & $0.053$ & $\approx 0$ & $\approx 0$ & $22.07$ & $34.23$ \\
\hline
Infrequent & $70.7$ & $22.3$ & $6.1 $ & $\approx 2.44$ & $119.77$ & $173.92$  \\
\hline
\end{tabular} 
\caption{$\mathbb{E}\left[\TotalExtraRehearsals{365}\right]$: Extra Rehearsals over the first year to remember $\sigma:\{1,\ldots,n\}\rightarrow \mathbb{Z}_{10}$ in our scheme with $f_{2,2}$ or $f_{1,3}$. Compared with  \sharedCues~schemes SC-0,SC-1 and SC-2\cite{NaturallyRehearsingPasswords}. }
\label{tab:Usability}
}
\hfill
\parbox{.23\linewidth}{
\begin{tabular}{| l | l || l | l | }
\hline
A & B & C & D \\
\hline \hline
0 & E & 5 & J   \\
1 & F & 6&K \\
2 & G & 7&L \\
3 & H & 8&M \\
4 & I & 9&N \\
\hline
\end{tabular}
\caption{Single-Digit Challenge Layout. Given a random mapping $\sigma$ from letters to digits the user can compute $f_{2,2}\paren{\sigma\paren{`C'}}$ by executing the following steps (1) Recall $\sigma(`A')$ --- the number associated with the letter A, (2) Recall $\sigma(`B')$, (3) Compute $i = \sigma(`A')+\sigma(`B') \mod{10}$ --- without loss of generality suppose that $i = 8$, (4) Find the letter  at index $i$---`M' if $i=8$, (5) Recall $\sigma(`M')$ (6) Recall $\sigma(`C')$ (7) Compute $j = \sigma(`M')+\sigma(`C') \mod {10}$ (8) Recall $\sigma(`D')$ (9) Return $j+\sigma(`D') \mod{10}$.} \label{tbl:challengeLayout}
}
\hfill
\end{table}

\section{Sum of $k$-Mins}
In the basic Hopper-Blum \cite{hopper2001secure} Human Identification Protocol the user memorizes a subset $S \subseteq [n]$ of $k=|S|$ secret indices. A single digit challenge consisted of a vector $x \in \mathbb{Z}_{10}^n$ of $n$ digits and the user responded by  with the $\mod{10}$ sum of the digits at $k \leq n$ secret locations plus an error term $e$ 
\[ \sum_{i \in S} x_i \mod{10} + e \ . \]
Typically, the user will set $e=0$, but occasionally the user is supposed to respond with a completely random digit instead of the correct response (e.g., so that the adversary cannot simply use Gaussian Elimination to find $S$). Thus, the human user must occasionally generate random numbers in his head to execute the Hopper-Blum protocol. This is potentially problematic because humans are not good at consciously generating random numbers \cite{wagenaar1972generation,humanRandom:figurska2008humans,seventeenMostRandom}.  In fact, hard learning problems like noisy parity might even be easy to learn when humans are providing the source of noise.

Hopper and Blum~ \cite{hopper2001secure} also proposed a deterministic human identification protocol call sum of $k$-mins. In this protocol the user memorized $k$ secret pairs $(i,j)$ of indices $S \subseteq [n]^2$. As before a single digit challenge consists of a vector $x \in \mathbb{Z}_{10}^n$ of $n$ digits. However, now the response to the challenge is deterministic
\[ \sum_{(i,j) \in S} \min\{ x_i x_j\} \mod{10} \ . \]
 We observe that for any constant $k$ the protocol is not secure against polynomial time attackers who have seen $O\big(k\cdot\log n \big)$ examples. The adversary can simply enumerate all possible sets $S$ of $k$ pairs and cross out the ones that are inconsistent with the challenge-response pairs he has already seen. Even for larger $k$ (e.g., greater human work) Hopper and Blum~ \cite{hopper2001secure} observed that the protocol was not secure against an adversary who has seen $O(n^2)$ examples. To see this observe that we can create an indicator variable $y_{i,j}$ for each pair $(i,j)$. Each challenge response pair $(x,r)$ yields a linear constraint
 \[ \sum_{(i,j)} y_{i,j} \min\{x_i,x_j\}  = r \mod{10} \ .  \]

\end{document}